\documentclass[12pt]{amsart}
\usepackage{graphicx} 
\usepackage[margin=1.5in,marginpar=.8in]{geometry}
\usepackage{subcaption}
\usepackage{cite}
\def\R{\mathbb{R}}
\def\Z{\mathbb{Z}}
\def\C{\mathbb{C}}
\def\Q{\mathbb{Q}}
\def\ot{\otimes}
\def\s{\sigma}
\def\del{\delta}
\def\G{\Gamma}
\def\bk{\mathbf{k}}
\def\bs{\boldsymbol{\sigma}}
\def\bv{\mathbf{v}}
\def\BZ{BZ}
\def\kdk{,\dots,}

\def\deg{\mathrm{deg}}

\usepackage{amsthm}
\usepackage{braket}
\usepackage{amsmath}
\usepackage{amssymb}

\usepackage{mathtools}

\newtheorem{theorem}{Theorem}[section] 
\newtheorem{proposition}{Proposition}[section] 
\newtheorem{lemma}{Lemma}[section] 
\newtheorem{corollary}{Corollary}[section] 

\theoremstyle{definition}
\newtheorem{definition}{Definition}[section]

\usepackage[foot]{amsaddr} 
 \usepackage{hyperref}
 \hypersetup{
    colorlinks=true,
    citecolor=blue, 
    linkcolor=blue, 
    filecolor=magenta,      
    urlcolor=blue,
}

\usepackage{caption} 

\title[Topological Quantum Phase diagrams]{Quantum Phase diagrams and transitions for Chern topological insulators}

\author{Ralph Kaufmann$^{1,2,3}$}%
 \email{rkaufman@math.purdue.edu}

\address{$^1$ Department of Mathematics, Purdue University, West Lafayette, Indiana, 47907.}
\address{$^2$ Department of Physics and Astronomy, Purdue University, West Lafayette, Indiana, 47907.}

\address{$^3$ PQSEI, Purdue University, West Lafayette, Indiana, 47907.}

\author{Mohamad Mousa$^{2}$}
  \email{mmousa@purdue.edu}

\author{Birgit Wehefritz–Kaufmann$^{1,2,3}$}%
 \email{ebkaufma@purdue.edu}

\date{\today}

\begin{document}
 
\begin{abstract}
Topological invariants such as Chern classes are by now a standard way to classify topological phases. Introducing and varying parameters in such systems leads to phase diagrams, where the Chern classes may jump when crossing a critical locus. These systems appear naturally when considering slicing of higher dimensional systems or when considering systems with parameters. 

As the Chern classes are topological invariants, they can only change if the ``topology breaks down''. We give a precise mathematical formulation of this phenomenon and show that synthetically any phase diagram of Chern topological phases can be designed and realized by a physical system,  using covering, aka.\ winding maps.
Here we provide explicit families realizing arbitrary Chern jumps. The critical locus of these maps is described  by the classical rose curves. These realize the lower bound on the number of Dirac points necessary obtained from viewing them as local charges. We treat several concrete models and show that they have the predicted generic behavior.

In particular, we focus on different types of lattices and tight--binding models, and show that effective winding maps, and thus higher Chern numbers, can be achieved using $k$--th nearest neighbors. 
We give explicit formulas for a family of 2D lattices using imaginary quadratic field extensions and their norms.  Our study includes the square, triangular, honeycomb and Kagome lattices.

\end{abstract}

\maketitle
\section{Introduction}
The use of topological invariants in condensed matter physics has by now been well established.
The most prominent among these are Chern classes and especially the first Chern class. These characteristic classes are defined for bundles and take values in cohomology.
To associate them to a given physical system, one has to find the appropriate geometric setting that contains the bundle of which one takes the Chern class. Following Berry \cite{Berry1984}, in a non-degenerate situation, this can be done by looking at a line bundle defined by a non--degenerate state and then integrating the Berry curvature obtained from the Berry connection which corresponds to adiabatic transport. As Chern classes are topological invariants, they do not actually depend on the choice of the connection, and one can use any other convenient one, see, e.g. Simon \cite{Simon1983} for examples.
In condensed matter, this bundle is usually obtained via Bloch theory where the base of the bundle is the Brillouin zone and the bundle is the bundle of occupied bands, for details on the geometry see, e.g. \cite{KKWK16}.
Another method is to utilize projectors and K-theory \cite{Connes1994,Bellissard}. This has the advantage of carrying over to the noncommutative case, which physically appears when considering non--trivial magnetic fields.

One area of continued further study is the realization of higher Chern numbers and their appearance in phase diagrams.
Studies in the physics literature have indicated that adding long-range interactions can result in a higher Chern number. In particular,  \cite{Sticlet2013} contains a systematic numerical and analytical study of extending Haldane's model in graphene, which showed that the Chern number generally increases with adding more distant neighbors. 
Additionally, other methods have been proposed to realize tight-binding lattices with higher Chern numbers, such as adding bands (orbitals) for the unit cell, and stacking multilayers of the system. An important note is that adding orbitals and stacking are related. After stacking, there are typically more atoms in the new unit cell. The atoms from different layers can be thought of as different orbitals \cite{Sarma2012}.
Studies of the appearance of higher Chern numbers along these lines can be found in
\cite{Chen2011,Sarma2012,Alase2021,Mondal2022,Sticlet2012,Sticlet2013,Bena2011,Fruchart2013,Lee2015,Eslam2022,Woo2024,Wang2015}.

Mathematically, there are three standard ways to construct bundles with higher Chern numbers that are at ones' disposal: forming direct sums, taking tensor products, and pulling back along higher degree maps, see \S\ref{par:background} for details.
Using this background and mathematical tools,
we present a general method for studying the
appearance of higher Chern number phases and show that any quantum phase diagram in which the phases correspond to non--degenerate ground states with given Chern number can be realized. This is achieved using concrete standard families of Hamiltonians and will allow to design such systems. We expect this to have applications in condensed matter physics and quantum computation.

To make contact with existing geometries and systems, we then focus on the particular structure of 2D discrete tight-binding lattices that realize topological insulators. These models have a quantized Hall conductivity given by the first Chern number \cite{Thouless1982}. It is natural for physics models to have interactions that are local in space and whose magnitude depends on the physical Euclidean distances between lattice points. This is the physical basis for the general study of commensurate sublattices, see \S\ref{par:highernbh}. In \S\ref{par:commensurate}, we specialize to several lattices that model concrete condensed matter systems and find that these instantiations match  both the general mathematical models and analysis, as well as the higher order neighborhood constructions.

In summary, our analysis answers the following basic questions:
\begin{itemize}
        \item [Q1:] What Chern classes, transitions, and phase diagrams can be achieved? 
   \item[Q2:] Can these be effectively understood and computed?
        \item [Q3:] Can one effectively design phase diagrams and realize them?
\item [Q4:] How can these be implemented on a lattice using lattice geometry?
 \end{itemize}
The answers we provide are as follows:\\
{\bf A1}: Arbitrarily high Chern numbers can appear. These are already realized by basic spin Hamiltonians as studied by Berry \cite{Berry1984}. Using Clebsch-Gordon rules, the bundles corresponding to higher spins can also be obtained from $spin\frac{1}{2}$, by using tensor products and projectors, see \S\ref{par:spin}. 
 We show that all Chern numbers and all transitions are possible and that they can effectively be obtained from the Chern number $\pm 1$ $spin\frac{1}{2}$ systems by pull--back  along higher degree covering maps.
The 1--parameter family describing a single wall transition from Chern number $d$ to Chern number $d'$ is just a suspension of a family of curves.
        The critical locus is the suspension of the critical curve which is a classical rose curve $r=\cos(k\theta)$, whose history goes back to \cite{Grandi}, with $k=\frac{d-d'}{d+d'}$ parameterized over $[0,\delta\pi]$ where $\delta=|d-d'|$.
        The number of crossings, which are Dirac points in 2 band systems, is $\delta$,
        which is the number of times $(0,0)$ is crossed.
        This number is the minimal for the two-band systems. Details are in \S\ref{par:rose}.\\
{\bf A2:}  Based on classical differential topology, the computation can be done via the curvature of a connections, as mentioned above. In special cases, the bundle is a pull--back from a bundle with known Chern class, such as one of the examples above. Then  the  Chern number can be computed by the mapping degree of the function which is used to pull--back. The quickest computation of the degree is done with a ray method, which we review in \S\ref{subsection:calcdegree}. This allows one to understand and ``read off'' the transitions and provide graphical versions to represent them.\\
{\bf A3:} As we show, any prescribed phase diagram of transitions between topological phases with different Chern numbers whose crossings are polytopic can be realized, see Theorem \ref{thm:main}.  Concretely, if the critical locus near a point of the phase diagram corresponds to a polytope fan, then these families can be interpolated using the normal cone.  These standard families are minimal for 2--band systems, in the sense that they have a minimal number of level crossings.

\noindent{\bf A4:}  The implementation is abstractly given by successive quotients by lattices and sublattices, see \S \ref{par:supercell}. Physically,  these can come by including interaction terms with higher neighbor interactions. The effects of moving to sublattices, or dually to bigger unit cells, on the band structure  is discussed in \S\ref{par:bands}. 
Concretely, the theory is implemented in physical lattice systems by using commensurate lattices of higher order nearest neighbors, see \S\ref{par:lattices}.
We classify these for lattices formed by the integers $\mathcal{O}(\sqrt{-d})$ in imaginary quadratic field extensions in Appendix \ref{par:quadratic}. The results are summarized in Theorem \ref{theorem:main2}. In particular. for 
$d=1$ these are the Gaussian integers, aka.\ the square lattice, and for $d=3$ the Eisenstein integers, aka.\ the triangular lattice.
The method is to use knowledge about primes in these fields. This classifications lets us identify infinite families of commensurate neighbors. We extend the analysis to sublattices, such as the hexagonal and Kagome lattices, see \S\ref{par:commensurate}.

\section{Geometry}
We introduce the geometric cast of characters for the analysis of topological phase diagrams.
\subsection{Background}
\label{par:background}
The following is a review of the material, we need to formally write and prove the statements. Although most of it is standard,  we include the presentation for completeness and to gather it in one place, as it comes from different perspectives and sources.
\subsubsection{Ground state line--bundle and n-band systems}
Given a parametrized family of Hamiltonians,
with spectrum bounded from below, and a non--degenerate ground state, one obtains a line bundle of this ground state on the parameter space. 
This generalizes to the case of the lowest $n$ bands,
which form a rank $n$ vector bundle.
If the bands do not cross, the respective Eigenfunctions yield line bundles $L_i,i=1\kdk n$, and the vector bundle splits as $V=L_1\oplus \dots \oplus L_n$.
In the standard condensed matter setup, $V$ is the Bloch bundle and $B$ is the Brillouin zone. The two standard cases are $B=S^n$, the $n$--sphere, or $B=T^n$, the $n$--torus.
\subsubsection{Chern classes}

Cohomological invariants of vector bundles $\pi:E\to B$ are given by the Chern classes  $c_i(E)\in H^{2i}(B,\Z)$. In the usual 2D case $B=S^2,T^2$, there is only the class $c_1\in H^2(S^2)\simeq \Z$ or $c_1\in H^2(T^2)\simeq \Z$. The isomorphism with $\Z$ is given by choosing an orientation. By choosing this isomorphism, the invariant becomes integer--valued.
More generally, for a smooth $n$ dimensional connected manifold $H^n(M^n)\simeq \Z$. If $M$ is even-dimensional $n=2m$,
one would have $c_m(V)\in H^{2m}(M^{2m})\simeq \Z$.  

Choosing an orientation fixes the isomorphism with $\Z$ concretely as follows. An orientation is given by a choice of the fundamental class $[M^n]\in H_n(B)$, and to obtain Chern numbers one caps with the fundamental class $-\cap [M]:H^n(M)\to H^0(M)=\Z$.  Since one can evaluate any degree $n$--cohomology class in this way, if  $n=2m$ is even and if $i_1+i_2+\dots i_k=m$, then $(c_{i_1}\cup c_{i_2}\cdots \cup c_{i_k})\cap[M]\in \Z$. These are the Chern numbers.
It is convenient to put together the Chern classes into a polynomial, $c_t(V)=\sum_k c_k(V) t^k$, where $c_0=1$.
Note that there is a vanishing result for Chern classes which states that $c_i(E)=0$ if $rk(E)<i$.

Then the Chern classes satisfy $c_t(V\oplus W)=c_t(V)c_t(W)$. In particular, for a line bundle $c_t(L)=1+tc_1(L)$ 
 and hence $c_1(L_1\oplus L_2)=c_1(L_1)+c_1(L_2)$. 

The Chern character which is defined for a line bundle as $ch(L)=\exp(c_1(L))$ and then extended by the splitting principle takes values in the even cohomology with $\Q$ coefficients $H^{ev}(B,\Q)=\bigoplus_i H^{2i}(B,\Q)$. That is, if $x_1,\dots x_n$ are the Chern roots of $V$ --that is the formal expression $c_t(V)=\prod_j(1+x_jt)$-- then $ch(V)=\sum_j e^{x_i}$;  the  coefficients are elementary symmetric functions in the $x_i$ which can be identified with Chern classes, see \cite{Hirzebruch}:
$\operatorname {ch} (V)=\operatorname {rk} (V)+c_{1}(V)+{\frac {1}{2}}(c_{1}(V)^{2}-2c_{2}(V))+{\frac {1}{6}}(c_{1}(V)^{3}-3c_{1}(V)c_{2}(V)+3c_{3}(V))+\cdots$. 
It is a ring homomorphism, viz.\ it satisfies 
$ch(E\oplus F)=ch(E)+ch(F)$ and $ch(E\otimes F)=ch(E)ch(F)$ where the product is the cup product. From this it follows that for line bundles $c_1(L_1\otimes L_2)=c_1(L_1)+c_1(L_2)$.

\subsubsection{Chern--Weil theory}
Note that in the differentiable case, by Chern--Weil theory, see e.g.\ \cite{KobayashiNomizu}, these cohomology classes have representatives in terms of differential forms of degree $2i$.  Concretely, picking a connection $\nabla$ and denoting its curvature form $\Omega$, then
$ \det \left(\frac {i\Omega}{2\pi}t +I\right) = \sum_k c_k(V) t^k$. The surprising fact is that this is independent of the choice of connection. Using $ det(X)=\exp(\mathrm {tr} (\ln(X)))$, the  de Rham Chern forms which represent the Chern classes are explicitly given by

\begin{multline}
    c^{dR}_t(V)= \left[ I 
       + i \frac{\mathrm{tr}(\Omega)}{2\pi} t 
       +   \frac{\mathrm{tr}(\Omega^2)-\mathrm{tr}(\Omega)^2}{8\pi^2} t^2\right.\\ 
       \left. 
       + i \frac{-2\mathrm{tr}(\Omega^3)+3\mathrm{tr}(\Omega^2)\mathrm{tr}(\Omega)-\mathrm{tr}(\Omega)^3}{48\pi^3} t^3
       + \cdots  \right]
\end{multline}

The Chern character in this notation is
$  {ch} (V)=\left[\operatorname {tr} \left(\exp \left({\frac {i\Omega }{2\pi }}\right)\right)\right]$.
If $\omega_1,\dots \omega_m$ are the forms representing $c_i(E)$, then the cup product is given by the wedge product of forms and  capping is given by integration. The Chern number is $\int_{[M]}\omega_{i_1}\wedge \dots \omega_{i_k}$.

It is important to note that the Chern classes in cohomology and the Chern numbers do not depend on the choice of connection.

\subsubsection{Pull--back and mapping degree}
Being characteristic classes means that Chern classes behave well under pull--back. If $f:B'\to B$ is a continuous map, then $c_i(f^*(V))=f^*(c_i(V))$, where on the left-hand side the bundle is pulled back and on the right-hand side the cohomology class is pulled back. This holds on the form level as well. Given a map $f:M^n\to N^n$ 
between two compact oriented $n$--manifolds, one can define the mapping degree by $f_*[M]=deg(f)[N]$, and for the Chern numbers we have $(f^*c_{i_1}\cup f^*c_{i_2}\cdots \cup f^*c_{i_k})\cap[M]=deg(f)(c_{i_1}\cup c_{i_2}\cdots \cup c_{i_k})\cap[N]$.

In the non--compact case one can compute the degree of a proper map using compactly supported forms \cite{BottTu}.
This is needed if one thinks about maps to $\R^n\setminus \{O\}$. Any map from a compact manifold to $\R^n\setminus \{O\}$ is proper and one can compute the degree by pulling back a local generator which is given by a bump function. 
Hopf's Theorem \cite{MilnorBook} states that the mapping degree of a map $f:M\to S^n$ is a complete invariant for the homotopy groups $[M,S^n]$.
In particular, the mapping degree of the map $f_d:z\to z^d:S^1\to S^1\subset \C$, or $\theta \mapsto d\theta$ in polar coordinates, is of degree $d$ and hence any map $f:S^1\to S^1$ is homotopic to one of the maps $f_d$.

Concretely for the spaces $T^2$ and $S^2$, we can  raise the Chern number by using the following maps: $f_{d_1,d_2}=f_{d_1}\times f_{d_2}:T^2\to T^2, (\theta_1,\theta_2)\mapsto (d_1\theta_1, d_2\theta_2)$ whose degree is $d_1d_2$. This follows from the fact that the map $f_d$ has degree $d$.
Likewise, a degree $d$ map on $S^2$ is induced by the $f_d$ map via suspension $Sf:SS^1=S^2\to S S^1=S^2$.
Recall that the suspension of a space is given by  $SX= (X\times [0,1])/(X\times \{0\}){\big /}(X\times \{1\})$, viz.\ identifying $X\times \{0\}$ to one point and $X\times \{1\}$ to another. The suspension of a map is given by  $f([x,t])=[f(x),t]$.
In particular,  in spherical coordinates, $Sz^d$ is the map $(\phi,\theta)\mapsto (d\phi,\theta):SS^1=S^2\to SS^1=S^2$.

\subsubsection{Calculating degrees}\label{subsection:calcdegree}
There are several ways to compute the mapping degree. 
For compact $n$ manifolds and $f:M^n\to N^n$, one can pull back any $n$--form $\omega$, then $deg(f)=\frac{\int_M f^*\omega}{\int_N\omega}$. A standard choice for $\omega$ is the volume form.

In particular, if $N=S^2$ and $\omega=xdy\wedge dz-y dx\wedge dz+z dy\wedge dz$ is the standard volume form,
then $\int_{S^2} \omega=4\pi$. If $f$ in a coordinate patch of $N$ is given by
 $f(k_x,k_y)=(f_x,f_y,f_z):T^2\to S^2\subset \R^3$, then $f^*{\omega}=f\cdot(\frac{\del}{\del k_x}f\times \frac{\del}{\del k_y}f)$. Given any map $h:N^2\to \R^3\setminus\{O\}$, we can normalize it to $f=\frac{h}{||h||}$ and then $deg(h)=deg(f)$. 
If for instance $M=T^2$ with coordinates $k_x,k_y$ and
 $h:T^2\to \R^3\setminus \{O\}$, then we obtain:

\begin{equation}
\label{eq:omegadeg}
deg(h)=deg(f)=\frac{1}{4 \pi} \int_{\mathrm{T^2}}  f\cdot\left(\partial_{k_x}f \times \partial_{k_y} f\right) dk_x\, dk_y
\end{equation}

Another standard theory for computing the mapping degree is as follows, see e.g.\ \cite[p 40ff]{BottTu}, and \cite{MilnorBook}.
By Sard's Theorem, the set of critical values of a smooth map
$f : M \to N$ between two manifolds has measure zero.
Picking a regular point $q$ and a preimage $p$, the map $f$ is a local diffeomorphism $\phi_p$. Let $sgn(\phi_p)=\pm 1$ be the sign of the determinant of the Jacobian at $p$. I.e.\ $1$ if $\phi_p$ preserves orientation and $-1$ if it reverses orientation. The mapping degree is then given 
by 
\begin{equation}
\label{eq:degree}
    \sum_{p\in f^{-1}(q)}sgn(\phi_p)
\end{equation}
Note in the non--compact case $f$ needs to be a proper map.

Using this method for a map $f:B^n\to \R^{n+1}\setminus \{O\}$ considered in the homotopy class $\tilde f\in [B^n,S^n]$, one has to identify the $p\in \tilde f^{-1}$ at a regular point and the local orientation at these points. Being a local diffeomorphism, $f$ is an immersion and by Sard's theorem one can isotope the map to have transversal intersections on a set of measure $0$. This means that we have an immersion. After this preparation, one can use the ray method to compute the degree.
For this consider the image $Im(f)\subset \R^3\setminus \{O\}$ and a ray $R=(0,\infty)\mathbf{v}, \mathbf{v}\in S^2$ which intersects $Im(f)$ transversely outside the locus of self--intersections.  Then there is a bijection of points $p\in {\tilde f}^{-1}(\mathbf{v})\leftrightarrow R\cap Im(f)$. The local orientation can be found by computing the orientation of the frame $f_*(e_i),\mathbf{v}$ in $\R^{n+1}$, where $e_i$ is a basis of $TB_p$.  

Note that in physics literature the points $p\in {\tilde f}^{-1}(\mathbf{v})\leftrightarrow R\cap Im(f)$, especially for the choice of the positive z--axis as the ray are sometimes called pre-Dirac points. For maps $B^n\to \R^{n+1}$ Dirac/Weyl points are the points in $B^n$ that are preimages of the origin.
Pre-Dirac points will already have two coordinates zero, whence the name. A more robust definition would be, points near the origin that in the limit of the phase transition have the origin as their limit.

Consider $g:B^n\to R^{n+1}\setminus \{O\}\sim S^n$ ---here and in the following $\sim$ means homotopic--- and let $f=Sg, g:SB'\to S\R^{n}\setminus \{O\}\sim S^{n+1}$, then by the formula \eqref{eq:degree} $deg(f)=deg(g)$. Again, by  Hopf's Theorem, we always have that $f\sim Bg$ for some $g$ as the mapping degree is the only homotopy invariant.
Specializing to a map $g:S^1\to \R^2\setminus \{O\}$, we can see 
that $deg(g)=wind_0(g)$ is the classical winding number around $0$ of $g$  
viewed as a closed curve. If $g$ is in transverse position, this 
can be computed by the ray method with $\mathbf{v}\in S^1$. We 
then get for $f\sim Sg:SS^1=S^2\to S(\R^2\setminus \{O\})\sim  SS^1=S^2$, $deg(f)=deg(g)$.

In particular, if $\R^2\cap im(f)=\bigcup_i c_i$ and let the 
orientation of $c_i$, be such that the orientation $\mathbf{t}$ 
of $(\mathbf{t}, \mathbf{v})$ is that of $(f_*(e_1),f_*(e_2),
\mathbf {v})$ for at most transversely intersecting closed 
curves, then $deg(f)=\sum_i wind_0(\pm c_i)$. By the previous 
arguments 
a) the winding number of the multicurve is the degree of the map $f$ 
by the ray method in $\R^3\setminus\{O\}$ for a $\mathbf{v}$ 
in the equator of $S^2$ and b) any map $f\simeq bg$ by  Hopf's Theorem and the degree of $g$ is the winding number of the equator. To move from a) to b) one can use a homotopy on the map $f$. This consolidates the multicurve into one curve. The difference between a) and b) is that in b) one automatically has the right orientation, while in a) one has to first compute it.
Note that if one does not compute the orientation, then the method in b) only gives an integer mod 2, as one does not know the signs of the local contributions, cf.\ \cite{MilnorBook}.
The contribution in the jump of the Chern number is given by the local charge as defined in \cite[Section 1.5]{KKWK20}.
In particular, for a two--band crossing the local charges are $\pm 1$, see \cite[Corollary 2.4 ]{KKWK20}.

This analysis shows that 
\begin{enumerate}
    \item Chern classes jump only when the degenerate locus is crossed. This happens when one closes a gap, and that is the description of a quantum phase transition.
    \item Generically the Chern numbers will jump by $1$ for each  gap closing at a non--higher--multiplicity point.
    \item  They can jump by multiples if either a non--transversal intersection point passes through the degenerate locus or several bands cross at once.
\end{enumerate}

\subsubsection{Momentum space and Families of Hamiltonians} \label{subsec:momentumspace}

Consider a family of Hamiltonians depending on parameters in a base space $B$, that is $H:B\to Herm_d$. Consider the trivial bundles $E=B\times \C^d\to B$. This has a fiberwise action $E_b\to E_b:(b,v)\mapsto (b,H(b)v)$ and
decomposes $E=\bigoplus_{i=1}^lV_i$ into block vector bundles $V_i$ with $\sum_i rk(V_i)=d$, where the Eigenvalues in 
$V_i$ and $V_j$ never cross if $i\neq j$ and do cross inside the $V_i$.
Let $B_0$ be the subspace on which the family is non--degenerate. Over $B_0$ the bundle splits into line bundles $V=\bigoplus_{i=1}^d L_i$. As the total bundle is trivial, $\sum_i c_1(L_i)=0$, see \cite{KKWK16} for more details.

On $B_0$, the topological invariants are then the Chern classes $c_1(L_i)\in H^2(B_0,\Z)$. It is often assumed that $B_0$ is $B$, but we are also interested in the case that there are
defects, as we wish to study phase diagrams.
To compute the Chern number, one can utilize  the adiabatic connection following Berry \cite{Berry1984}, but can alternatively calculate with any connections, see e.g.\ \cite{Simon1983}.
For a 2D connected compact component $B_0^i$ of $B_0$,
for the $n$th band  $C_n^i=\frac{1}{2 \pi} \int_{B_0^i}  F_{12}(k)dS$. Where the curvature, aka.\ field strength tensor for the connection $A$ is given by $F_{12}(k)=\frac{\partial}{\partial k_1} A_2(k)-\frac{\partial}{\partial k_2} A_1(k)$. Choosing a Hermitian metric, $A_\mu(k)=-i\left\langle n_{\mathbf{k}}\right| \frac{\partial}{\partial k_\mu}\left|n_{\mathbf{k}}\right\rangle$, where, $\left|n_{\mathbf{k}}\right\rangle$ is a normalized local section, aka.\ normalized wavefunction of the $n$th band. Note, such a choice is possible and the value of the Chern number does not depend on the choice.
For the non--compact components, we consider, as usual, their one--point compactifications.

If we have a map $f^*:B'\to B$ and a family of Hamiltonians $B\to Herm_d$, we have the pull back family of Hamiltonians $f^*H=Hf:B'\to Herm_d$ given by $b\mapsto H(f(b))$ 
and the bundles defined by this family are the pull backs under $f$ of the bundles of the original family. E.g.\ 
if $V=\bigoplus_{i=1}^d L_i=B\times \C^d$, 
then $B'\times \C^d=\bigoplus_{i=1}^d f^*(L_i)$ with first Chern classes $f^*c_1(L_i)$ and Chern numbers $deg(f)c_1(L_i)$.

    Note that any family is homotopic to a traceless family, and for non--degenerate families the homotopy can be chosen to stay in the space of non--degenerate families.
    The homotopy is given by $(1-t)H-t\frac{1}{d}(H-tr(H)I)$.
    The scaling has the effect of shifting all Eigenvalues by the same amount at the same time, so there are no new crossings.
    Thus, we can assume that the families are traceless, when considering homotopy invariants.

To obtain phase transitions which are detected by changes in Chern numbers, by varying parameters, one has to pass through families that have  degenerate locus $B_{deg}=B\setminus B_0$. We will assume that this is not wild ---the technical details are below--- and, for convenience, that $B_0$ has finitely many components.

\subsection{Standard examples with higher Chern numbers and universality}

\subsubsection{Spin Hamiltonians}
\label{par:spin}
The standard example of a Hamiltonian family is given by the family $H(\mathbf{k})=\mathbf{k}\cdot \boldsymbol{\sigma}$, where $\boldsymbol{\sigma}=(\sigma_x,\sigma_y,\sigma_z)$ is a spin $s\in \Z[\frac{1}{2}]$ which is of dimension $d=2s+1$.
This is a family on $\R^3$ which is only degenerate at the origin. Restricting it to $S^2\subset \R^3$, the trivial bundle $S^2\times \C^d\to S^2$ splits as a direct sum of line bundles $L_i, i=-s,\dots, s$ corresponding to the Eigenvalues $-s,\dots, s$, whose Chern numbers are $\int_{S^2}c_1(L_i)=2i$, \cite{Berry1984,Simon1983}.
The ground state bundle for this Hamiltonian has $c_1=-2s$.
Flipping $H$ to $-H$ the ground state bundle has $c_1=2s$.
For example, for $s=\frac{1}{2}$, we obtain two bundles with Chern numbers $\pm 1$.
The bundle with Chern number $1$ is classically known as the Hopf bundle. This is the associated vector bundle to the Hopf fibration $S^1\to S^3\to S^2$, aka.\ the Bloch sphere obtained by quotienting out a global phase from states in $\mathbb{C}^2$.

Thus there is a line bundle over $S^2$ with any given Chern number that comes from a Hamiltonian family. We can pull back this bundle to $T^2$ under the standard  degree--$1$ map which contracts the one--skeleton $S^1\vee S^1$ to a point without altering the Chern number. More generally, if $\BZ$ is an oriented 2D connected CW complex with one 2-cell, collapsing its 1-skeleton gives a degree--$1$ map $\pi:\BZ\to S^2$, along which we can pull back. This yields a line bundle over any such $\BZ$ with a specified Chern class.

Note that taking tensor products of these bundles is equivalent to the Clebsch--Gordon rules and thus taking tensor products and splitting off bundles via projection also yields these higher Chern number bundles. E.g.\  $Spin\frac{1}{2}\otimes Spin\frac{1}{2}=Spin \,1\oplus Spin \,0$ where the first is a rank 3 bundle that splits into line bundles with Chern numbers $-2,0,2$ over $S^2$ and the second summand is a trivial line bundle. Namely, if $L$ is the bundle with $c_1(L)=1$ then $V=L\oplus \bar L$, $V\ot V=(L\ot L)\oplus (L\ot \bar L)\oplus (\bar L\ot L)\oplus (\bar L\ot \bar L)$ where the summands have first Chern numbers $1+1=2,1-1=0,-1+1=0,-1-1=-2$.

\subsubsection{Universality}
More generally, we can pull back the standard spin $s$ representations as families $\R^3\to Herm_d$, $d=2s+1$ via maps 
$f:B\to \R^3$. For traceless $2\times 2$, $d=2,s=\frac{1}{2}$ families the spin family is universal.

\begin{theorem}\label{theorem:ksigma}
For any base space $B$, any given family of $2\times 2$ Hermitian Hamiltonians $B\to Herm_2$ is homotopic to the standard family $H(\mathbf{k})=\mathbf{k}\cdot \boldsymbol{\sigma}$, where $\boldsymbol{\sigma}=(\s_x,\s_y,\s_z)$ is given by the Pauli--matrices. And, moreover, non--degenerate families are pulled backs via maps $f:B\to \R^3\setminus \{O\}$. 
\end{theorem}
\begin{proof}
The identity matrix $\sigma_0=I$ and the Pauli matrices 
$\sigma_1,\sigma_2,\sigma_3$ are a basis of the Hermitian 
$2\times 2$ matrices and the latter are a basis for 
the traceless Hermitian matrices. Using this basis, a 
map $H:B\to Herm_2$ can be written as $\sum_{i=0}^3 
f_i\sigma_i$. The straight line homotopy $(1-t)f_0$ 
to $0$ makes the map homotopic to $\sum_{i=1}
^3f_i\sigma_i$. The map $B\to \R^3$ is then given  by 
$b\mapsto (f_1(b),f_2(b),f_3(b))$. The degenerate points in such a family are the inverse images of the origin $O$, as $\lambda_1=\lambda_2=-\lambda_1$ means that $\lambda_1=\lambda_2=0$.
\end{proof}

\begin{corollary}
    The Chern numbers of the line bundles corresponding to non--degenerate $2\times 2$ family of Hamiltonians on a base space $B$ can be computed by pull back. In particular, if the base space is $2$--dimensional, then $c_1$ can be computed as the mapping degree of the function $f$ from the theorem above.
\end{corollary}
For the mapping degree we note that  $\R^3\setminus \{O\}$ is homotopic to  $[B,\R^3\setminus \{O\}]\simeq [B,S^2]$.
The map is induced by the  homotopy $(1-t)f+tf/||f||$. The effect of this homotopy on $H$ is spectral flattening, which means that the family of Hamiltonians is 
$(1-t)H+tH/|\lambda|$, where $\lambda,-\lambda$ are the Eigenvalues of $H$. The resulting Eigenvalues after the homotopy that is at $t=1$ are then just $\pm 1$.
\begin{corollary}
    For any bundle $V$ over a compact connected 2--manifold $\BZ$ that splits as a sum of line bundles $V\simeq L_1\oplus \dots \oplus L_n$
    with first Chern classes $c^1_1,\dots, c_1^n$, there are functions $f_i:\BZ\to S^2$ with $deg(f_i)=c_1^i$ such that  $V\simeq\bigoplus_i f_i^*(H)$, where $H$ is the Hopf bundle, that is the $-H_{spin\frac{1}{2}}$ ground state bundle. 
    Alternatively, using the degree $1$ map $\pi:\BZ\to S^2$, $L_i=\pi^*(\tilde L_i)$
    where $\tilde L_i=g_{c_1^i}^*H$.
    This can be identified with the highest state bundle of the spin $s=\frac{1}{2}c_1^i$ system.
\end{corollary}

This in particular also allows to reconstruct bundles stemming from given families of Hamiltonians using pull--backs of pairs of line  bundles stemming from $spin\frac{1}{2}$. 
As the total bundle is trivial in this case, the last line bundle $L_d$ is equivalent to the quotient of the trivial $d$--dimensional bundle $\tau_d$ by $L_1\oplus\dots\oplus L_{d-1}$ which has first Chern number $-(c^1_1+\dots+ c_1^{d-1})$.

\begin{lemma}
    Let $H\ot \bar H$ be the  bundles of the $spin \frac{1}{2}$ system, a fixed $n$ and
 a collection of Chern numbers $c_1^i,i=1,\dots n$ with $\sum_i c_1^i=0$.
Then 
\begin{equation}
    H^{\ot c^1_1}\oplus\dots \oplus H^{\ot c^{n-1}_1}
\oplus( \bar H^{\ot c^1_1}\otimes\dots \otimes H^{\ot c^{n-1}_1})
\end{equation}
is the trivial rank--$n$ vector bundle over $S^2$ which splits as line bundles with the given Chern classes.
\end{lemma}
 \begin{proof}
 This follows from the fact that the bundle is $\bigoplus_{ij} H^{\ot c^1_i}\ot \bar H^{\ot c_j^1}$ and $H\oplus \bar H$ is trivial.
\end{proof}
 
\begin{theorem}
    A given non--degenerate $n$--band  
    structure $V=L_1\oplus\dots \oplus L_n\to \BZ$,
    where  $L_1$ are line bundles with Chern 
    numbers $c_1^i$ and  $V$ is trivial can be 
    realized via $spin \frac{1}{2}$ bundles using the above construction.
\qed
\end{theorem}

This is commensurate with the analysis for $su(3)$ Hamiltonians yielding  three-band models  in \cite{Lee2015}.


\subsection{Tight-binding Lattice Hamiltonians}
One standard setup is to consider a periodic Hamiltonian $H(x)=H(x+l)$ where $l\in \Lambda\subset \R^n$ a translational lattice group $\Lambda \simeq \Z$, cf.\ e.g.\ \cite{KKWK20}; we will call this the mathematical lattice. Using Fourier transform in momentum space, this give a family of Hamiltonians $H(k)$ $k\in \R^n/\Z^n=T^n=(S^1)^{\times n}$. If the space is discretized into a site lattice $\Gamma$, we will call this the physical lattice, which is  translationally invariant $l\Gamma=\Gamma$ for all $l\in \Lambda$. $\G$ is actually a graph, which means that it also has edges. If only the vertices are given, then the edges are defined to be between nearest neighbors.  Then the elementary cell is $\bar \Gamma=\Gamma/\Lambda$.  We set $\pi:\G\to \bar\G$ the projection, $\Gamma_v=\pi^{-1}(v)$ and $d=|\bar \Gamma|$ the number of sites in an elementary cell.

A basic example is then given by the standard Harper Hamiltonian. The Hilbert space is $\mathcal{H} =\ell^2(\Gamma)$, were $\ell^2(\Gamma)$ are the $\ell^2$ functions on the vertices and the Hamiltonian is given by translation along the edges. Splitting the Hilbert space 
as $\bigoplus_{v\in \bar \Gamma} \ell^2(\Gamma_v)$, the Hamiltonian becomes a $d\times d$ Hermitian family of matrices on momentum space: $H(k):T^n\to Herm_d$.
Note that as defined this is an operator valued matrix. To have a ``true'' matrix, that is an element in $End(\mathcal{H})$ of some Hilbert space, one needs to fix a spanning tree.
Practically this is done by choosing a root vertex and set $E$ of generating directed edges, then 

\begin{equation}
\label{eq:Harper}
H=\sum_{e\in E} U_e+U^{\dagger}_e
\end{equation}
where $U_e$ are Zak  (magnetic) translation operators, see e.g.\ \cite{Bellissard}.
There are several such choices that are related by re-gauging transformations, see \cite{KKWKsym}. 
This can be seen as a family of Hamiltonians over the Brillouin zone. The corresponding Bloch  bundle is the trivial bundle $E=T^n\times \C^d\stackrel{\pi}{\to} T^n$, where $\pi=\pi_1$ is the projection to the first factor. Being trivial, this carries no information, but on the locus $T^n_{0}\subset T^n$
where the Hamiltonians are non-degenerate, $E$ splits as $\bigoplus_{i=1}^d L_i$ where $L_i$ 
are the Eigenbundles to 
the Eigenvalues $\lambda_i$. Since the Eigenvalues are real and do not cross, we can order them $\lambda_1<\lambda_2\dots <\lambda_d$.
One can add different orbitals for more bands and on--site potentials for additional diagonal elements.

The basic example is the $\Z^2$ lattice whose Hamiltonian $H=U_1+U_1^*+U_2+U_2^*$ where $U_i$ is the translations along $e_i$.

For the hexagonal lattice, with vectors $f_1,f_2,f_3=-f_1-f_2$ moving from the A to the B sites, the Hamiltonian is
\begin{equation}
    H_{hex}=\begin{pmatrix}
        0&U+V+W\\
        U^\dagger+V^\dagger+W^\dagger&0
    \end{pmatrix}
\end{equation}
where $U=T_{e_1},V=T_{e_2},W=T_{e_3}=(UV)^{-1}$ as $f_1+f_2+f_3=0$.
Using Fourier transform, or equivalently characters on the $C^*$ algebra, such a matrix becomes
\begin{equation}
    H_{hex}=\begin{pmatrix}
        0&1+e^{ia}+e^{ib}\\
        1+e^{-ia}+e^{-ib}&0
    \end{pmatrix}
\end{equation}
This family has no $\s_z$ component and 
it is also degenerate at the Dirac points  
$(e^{2\pi i\frac{ 1}{3}} , e^{-2\pi i\frac{ 1}{3}} )$ and $(e^{-2\pi i\frac{ 1}{3}} , e^{2\pi i\frac{ 1}{3}} )$. Therefore, there will be no non--zero Chern class
as discussed in \cite[Section 4.11]{KKWK16}. 
To get a non--trivial Chern number, following Haldane, one can  invoke higher order interactions, see \S\ref{par:haldane}.

As an example of how to construct a model with an arbitrary Chern number, we can implement the strategy of using higher degree maps. In this way, we recover the following example of 
\cite{circuits2023}. 
Starting from the particular tight-binding model with 
\begin{equation}
    h_1(\bk)= \alpha \cos(k_x),  h_2(\bk)=-\beta\cos(k_y),  h_3(\bk)=m_0+\gamma_1 \sin k_x+\gamma_2 \sin k_y
\end{equation}
composing with a higher-degree map  $f_d: S^2 \to S^2$
corresponding to the map $Sz^d$, which using   coordinates $\R^3=\C\times \R$ is given by $(x+iy,z)\to ((x+iy)^d),z)$,
translates to constructing the new Hamiltonian for $d \in \mathbb{Z}^+$
\begin{equation}
\label{eq:tight}
\begin{aligned}
&\tilde h_1(\boldsymbol{k})+i \tilde  h_2(\boldsymbol{k})=\left(\alpha \cos k_x-i \beta \cos k_y\right)^d\\
&\tilde h_3(\boldsymbol{k})=m_0+\gamma_1 \sin k_x+\gamma_2 \sin k_y
\end{aligned}\end{equation} \\
Here, $\alpha, \beta, \gamma_{1,2} \text {, and } m_0 \in \R$ are parameters. 
The resulting Hamiltonian will have  Chern number $d$.


\subsubsection{Higher neighborhood interactions}
\label{par:highernbh}

One way to obtain new terms in the Hamiltonian which correspond to higher degree maps is to consider interactions which are next nearest or further neighbors. For the square lattice example \eqref{eq:tight} this was done in \cite{circuits2023}. 
The breakdown of the Hamiltonian in real space will generally involve terms for all neighbors up to the $d$th neighbors. We will give a more precise analysis in \S\ref{par:commensurate}.

To write down the Hamiltonian, given a lattice $\G\subset \R^n$ with $0\in \Gamma$ one has to find the sublattice $\G^{n}$ of $n$-th nearest neighbors.
If this sublattice is generated by the  vectors $w_1,\dots w_k$, and $\G$ is generated by the edge vectors $v_i$, then one can express the $w_i=\sum_j a_{ji}v_j$. The $n$-th nearest neighbor hopping is then given by 
\begin{equation}\label{equation:neighborhood}
    H^{(n)}=\sum_{i} \prod_j  U_j^{a_{ij}} 
\end{equation}
where the $U_j$ are the original translation operators.

In the square lattice $\Z^2$, for the 2nd nearest neighbors, of which there are four, the lattice vectors are
$w_1=(1,1),w_2=(-1,1), w_3=(-1,-1)$ 
and $w_4=(-1,-1)$. These are at distance $\sqrt{2}$ and the Hamiltonian is  
$H^{(2)}=U_1U_2+U_1U_2^*+U_1^*U_2^*+U_1^*U_2$. Transforming $V_1=U_1U_2$ and $V_2=U_1U_2^*$ we transform $H^{(2)}$ to the form $H$. This is the fact that the sublattice of second nearest neighbors is given by the sublattice $\Z(1,1)+\Z(1,-1)\subset \Z^2$.
There are again four 3rd nearest neighbors $(\pm 2,0),(0,\pm 2)$ at distance $2$ with Hamiltonian $H^{(3)}= U_1^2+U_2^2+U_1^{-2}+U_2^{-2}$.  Rewriting $V_i=U_i^2$,
we see that $H^{(3)}$ transforms into $H$. However, evaluating with a character, we have $V_1=e^{2ia}, V_2=e^{2ib}$, that is we get the pullback under $z\to z^2$.
This expresses the fact that the sublattice of 3rd nearest neighbors is $(2\Z)^2\subset \Z^2$. The full systematic treatment is in \S\ref{subsection:gaussian}.

The second order interactions for the hexagonal lattice are diagonal after gauging $W$ to $1$ and the corresponding Hamiltonian is given by 
\begin{equation}
\label{eq:hex1}
    H_{hex}^{(2)}=\begin{pmatrix}
        U+U^*+V+V^*+UV^*+VU^*&0\\0&
         U+U^*+V+V^*+UV^*+VU^*
    \end{pmatrix}
\end{equation}
Adding first and second order interactions with a coupling,
we obtain
\begin{multline}
\label{eq:hex2}
 H_{hex}+t  H_{hex}^{(2)}=\\  
 \begin{pmatrix}
        t[e^{ia}+e^{-ia}+e^{ib}+e^{-ib}+e^{ia}e^{-ib}+e^{-ia}e^{ib}]& 1+e^{ia}+e^{ib}\\1+e^{-ia}+e^{-ib}&
t[e^{ia}+e^{-ia}+e^{ib}+e^{-ib}+e^{ia}e^{-ib}+e^{-ia}e^{ib}] 
\end{pmatrix}
\end{multline}

\begin{multline}
=\begin{pmatrix}
t(2 \cos(a) +2 \cos(b)+2 \cos(a-b)) & 1 + \cos(a)+\cos(b) +i (\sin(a)+\sin(b))\\
1 + \cos(a)+\cos(b) -i (\sin(a)+\sin(b))& t(2 \cos(a) +2 \cos(b)+2 \cos(a-b))\nonumber
\end{pmatrix}
\end{multline}
\begin{multline} \nonumber
=    (1+\cos(a)+\cos(b))\sigma_x-(\sin(a)+\sin(b))\sigma_y+ t (2 \cos(a)+2\cos(b)+2 \cos(a-b))\sigma_0
\end{multline}
The full systematic treatment  is in \S\ref{par:honeycriteria}.
 
Note that the families \eqref{eq:hex1} and \eqref{eq:hex2} ---even on the non--degenerate locus-- are pull--backs with mapping degree $0$ as there are no $\sigma_z$ terms. 
To alter this, one has to introduce an asymmetry in the coupling, as Haldane did, see \S\ref{par:haldane} for details.

\subsubsection{Super--cells}
\label{par:supercell}
The way to understand how higher order neighborhood interactions generate higher Chern classes is via super--cells.
A typical situation in condensed matter theory is given by  covers coming from a sublattice in a lattice.  
Let $\Lambda'\subset \Lambda$ both be mathematical lattices of rank $n$ ---the 
standard example being $(2\Z)^n\subset \Z^n$. Then the 
quotient map 
\begin{equation}
\label{eq:covers}
    \pi: T^n=\R^n/\Lambda'\to (\R^n/
\Lambda)/(\Lambda/\Lambda') = \R^n/\Lambda'=T^n
\end{equation}
is a degree $|\Lambda/\Lambda'|$ covering.
We can now pull back a family from $T^2$ by $\pi$ to obtain higher Chern classes on $T^2$.
This corresponds to looking at a super--cell which covers the regular cell. In particular, using the volume form $\omega$ induced from the standard volume form on $\R^n$, we can compute  the degree of the map $\pi$ as $deg(\pi)=\frac{Vol(\R^n/\Lambda)}{Vol(\R^n/\Lambda')}=Vol(\Lambda/\Lambda')$.
If one has a physical lattice $\G$, then one has the two quotients $\bar \G=\G/\Lambda$ and $\bar \G'=\G/\Lambda'$ and a covering map $\pi_\G:\bar \G'\to \bar\G$ which is a discrete cover of degree $|\Lambda/\Lambda'|$.

\begin{definition}
We call a sublattice of $k$--th nearest neighbors 
$\G^{k}\subset \G\subset \R^n$ {\em commensurate}, 
if $\G'=l(k)\Gamma$, where $l(k)$ is the distance of the nearest neighbors. 
\end{definition}

In this case the cover is of degree $l(k)^n$ and the coefficients of \eqref{equation:neighborhood} satisfy $a_{ij}=l(k)$.
The following is now straightforward:
\begin{proposition}
For a commensurate $k$--th nearest neighbor lattice,
    the family of Hamiltonians $H^{(k)}:T^n\to Herm_d$
    is the pull back of the original lattice Hamiltonian under the diagonal scaling map $(\theta_1\kdk \theta_n)\mapsto l(k)(\theta_1\kdk \theta_n)$ whose degree is $l(k)^n$, and hence has Chern classes scaled by $l(k)^n$.
    \qed
\end{proposition}

Thus the question of finding natural lattice implementations turns into the question of finding commensurate higher order neighborhoods. For a large family of 2D--lattices this is done in \S\ref{par:commensurate}.

\subsubsection{Band structures for sublattices}
\label{par:bands}
In this section, we analyze  the effect of larger Brillouin zones, e.g.\ those coming from stacking. 
For a sublattice and a lattice, the matrix dimension of the two families of Harper, tight binding Hamiltonians differs by $|\Lambda/\Lambda'|=\frac{d'}{d}$,
where $|\G/\Lambda'|=d'$ and $|\G/\Lambda|=d$. To compare the two, one should re--sum the larger block--Hilbert spaces as
$\mathcal{H}_{\bar \G'}=\bigoplus_{v\in \bar \G'}\ell(\G_v)\simeq
\bigoplus_{w\in \bar \G}(\bigoplus_{v\in \pi_\G^{-1}}\G_{w})$.
This amounts to summing the entries given by the translation operators. 
As an example on $\Lambda=\Z^n$, we have the $1\times 1$ Harper Hamiltonian
$H=\sum_{i=1}^n U_i+U^\dagger_i$.
For the sublattice $\Lambda'=(2\Z)^2\subset \Z^2$ the operator in the Harper Hamiltonian derived from the lattice $\Z^2$ is given by a $4\times 4$ matrix $H_4=\tilde U_1+\tilde U_1^\dagger+\tilde U_2+\tilde U_2^\dagger$, where the operators $\tilde U_1$ and $\tilde U_2$ are
\begin{equation}
\tilde U_1=    \begin{pmatrix}
0&0&0&W_1\\
0&0&W_{3}^\dagger&0\\
0&W_{3}&0&0\\
W_{1}^\dagger&0&0&0\\
    \end{pmatrix}
\quad 
\tilde U_2=\begin{pmatrix}
    0&0&W_{2}&0\\
    0&0&0&W_{4}^\dagger\\
    W_{2}^\dagger&0&0&0\\
    0&W_{4}&0&0\\
\end{pmatrix}
\end{equation}
with the $W_i$ representing the translations along the edges in the new, bigger, unit cell.
 The original bundle embeds under the diagonal map on the fibers 
$\Delta^{[4]}:\C\to \C^4$. 
In particular, $H_4 (1,1,1,1)^T=H_1(1,1,1,1)$  where $H_4$ and $H_1$ are the $4\times 4$ and the $1\times 1$ families, respectively.
Note hat in $C^*$ geometry, see e.g.\ \cite{Bellissard,KKWK}, the entries are from the 4--torus and the corresponding map  $T^2\to T^4$ along which the family restricts to is
given by $W_1,W_3\to U_1$ and $W_2,W_4\to U_2$ on generators.

If one would simply start with the quotient graph $\Z^2/(2\Z)^2$,
then one would have the more general Hamiltonian
\begin{equation}
\hat H_4 =   \begin{pmatrix}
    0&0&V_2+V_6^\dagger&V_1+V_5^\dagger\\
0&0&V_3^\dagger+V_7&V_4^\dagger+V_8\\
V_2^\dagger+V_6&V_3+V_7^\dagger&0&0\\
V_1^\dagger+V_5&V_4+V_8^\dagger&0&0\\
\end{pmatrix}
       \end{equation}
This more general Hamiltonian has more operators and hence parameters.
These can  stand for more sites resulting from stacking or more orbitals.
The specialization equations that turn $\hat H_4$ into $H_1$ are given by $V_1=V_3=V_5=V_7=U_1$ and $V_2=V_4=V_6=V_8=U_2$, which corresponds to the iterated diagonal embedding $T^2\to T^4\to T^8$.

\subsubsection{Band folding}
There is more generally the possibility to push forward a family of Hamiltonians. 
This has the effect of what is known as band--folding.
Given a family of Hamiltonians $H:B'\to Herm_d$ and a covering $\pi:B'=\R^n/\Lambda'\to B=\R^n/\Lambda$ with $\Lambda'\subset \Lambda$,
for $\bv \in \Lambda/\Lambda'$ consider the translation operator $T_{\mathbf{v}}:\R^n/\Lambda'\to \R^n/\Lambda': T_\bv([l])=[l+\bv]$.
Fix one section $s_0$ of $\pi$ and set $s_\bv=T_\bv s_0$, i.e.\ the section obtained by shifting by the action of translation by $\bv$. Now consider
$\phi_*(V)=\bigoplus_{\mathbf{v}\in \Lambda/\Lambda'} s_\mathbf{v}^*V$.
The corresponding push--forward Hamiltonian in momentum space is 
given as follows:
$\pi_*(H)=\bigoplus_{\mathbf{v}\in \Lambda/\Lambda'}T_\bv(H)$, i.e.\ $\pi_*(H)(\bk)= \bigoplus_{\mathbf{v}\in \Lambda/\Lambda'}H(\mathbf{k}+\mathbf{v})$.
 This also has the effect of increasing the size of the matrix by a factor $\Lambda/\Lambda'$.
 The problem of unfolding is identifying the original band structure in the pull--back of the push--forward, $\pi^*\pi_*(H)$.

\section{Phase diagrams}
We will first give natural examples of phase diagrams, where phases are Chern phases of topological insulators, and the define the notion of an abstract phase diagram, which subsumes all of them. For this we prove that it is realizable by physical systems.

\subsection{Phase diagrams from families of Hamiltonians}
When dealing with families (or families of of families) of Hamiltonians, one obtains phase diagrams. A common situation is that one has a space of parameters $P$ which parameterizes Hamiltonians $H_p$ in real space. Each of these Hamiltonians gives a family $H_p(\bk)$ of $d$--dimensional Hermitian matrices with $\bk$ in the Brillouin zone $B$. 
At each point $p\in P$, we have the band structure of the family $H_p(\bk)$ as in \S\ref{par:background}, and if the ground state is non--degenerate, as it is for a Chern insulator, we have the first Chern class. In the case that $B$ is a  compact oriented surface, this yields the first Chern number which is a locally constant function, as it is a topological invariant. 

It will not be defined on the locus $P_{deg}$, where the ground state is degenerate, i.e.\ where we have a level crossing in the parameter space. Here, we just set the function to $0$.
All in all, we obtain a family $P\times B\to Herm_d$.
 The  dimensions of $P$ can be called synthetic dimensions in contrast to the real or momentum space dimensions. 
 
\subsection{Slicing and Phase diagrams}
Higher dimensional theories naturally lead to phase diagrams.
Given a family $B\to Herm_d$ without restriction on the dimension of $B$, as  Chern classes lie in $H^{2i}(B,\Z)$, to obtain a number, one in general has to pair with a degree $2i$ dimensional homology class.
Such degree $2i$ classes arise from the embedding of $2i$--dimensional manifolds $j:N^{2i}\subset B^m$. For these  one can evaluate $c_i(j^*(L))$ as an integral: $\int_{N}f^*(c_i(V))=\int_N c_i(V))$.
This works analogously for all Chern numbers. Thus choosing a family of embeddings, parameterized by $P$, one obtains a phase diagram. The degenerate locus is comprised of   those parameter values for which submanifold inclusions hit a degenerate point, i.e.\ $im(j)\cap B_\deg\neq \emptyset$.

A typical situation which come from 3d lattices  and is called slicing is given by an embedding $i_h:T^2\hookrightarrow T^3=T^2\times S^1$ by $(\theta_1,\theta_2)\mapsto (\theta_1,\theta_2, h)$, with the parameter space being $S^1$.
More generally, writing $T^n=T^{n-2}\times T^{2}$ with $T^{2}$ in the $i$ and $k$th slot $(i<k)$ we obtain $n\choose 2$ families with $P=T^{n-2}$.  This makes   $n-2$ of the coordinates synthetic coordinates or parameters. 

If the slices pass though the degenerate locus of the family, the Chern classes can jump, see \cite{KKWK16,KKWK20} for detailed examples. 
 If there are only finitely many degenerate points, then these numbers determine the class completely; see e.g.\ \cite{KKWK16}.

\subsection{Abstract phase diagrams}
\begin{definition}
An abstract $n$ parameter quantum phase diagram $PD$ will be given by a paracompact $n$--manifold  $P$, together with a codimension $1$ subspace $P_{deg}$ such that $P_0=P\setminus P_{deg}$ has finitely many components, and a locally constant function $C_1:P_0\to \Z$.
We allow that multiple critical lines or codimension-1 degeneracy submanifolds  meet, but the overall behavior is tame in the following technical sense. We
 restrict to the case that each point $p$ of $P_0$ is has a local neighborhood $U$ that is homeomorphic to  a neighborhood $V$ of the origin in  $\R^n$ with $U\cap B_0$ being homeomorphic to the codimension 1-cones of a polyhedral fan.
    
\end{definition}

In 1D, locally $P_\deg$ is just the 
origin, while in 2D finitely locally $P_\deg$ can be the union of finitely many rays emanating from the origin. In 3D locally $P_\deg$ can be the union of finitely many boundaries of finitely many polyhedral cones of a complete fan and so on. We recall that given a polyhedron which contains the origin, its fan has maximal cones given by the convex hull of the vertex vectors of each face of the polygon. The lower dimensional cones are then the intersection of the higher dimensional cones and correspond to the convex cones of the lower dimensional faces.

We will deal with the 2D case here and a Brillouin zone $\BZ$ is a compact orientable surface.
A family of Hamiltonians on a 2D Brillouin zone $\BZ$ parameterized by a manifold $P$ {\em will be called tame}, if together with the locus $P_{\rm deg}$ and the locally  constant function $C_1:P_0\to \Z$ given by the first number of class of the line bundle of $\BZ$ defined by the ground state it yields the data of an abstract phase diagram. 

\subsection{Designing phase diagrams} 
\begin{definition}
    Given an abstract phase diagram,
a realization is a family of Hamiltonians and ground states parameterized by $P$  which has the a degenerate locus $P_{\rm deg}$ and regions with different phases classified by the Chern number function $C_1$. This means that for $p\in P_0$, $C_1(p)$ is the first Chern number of a non--degenerate ground state. Designing a phase diagram means to pick an abstract phase diagram and realizing it. 
\end{definition}

\begin{theorem}
\label{thm:main}
   For any abstract phase diagram $PD$ and any 2D compact oriented Brillouin zone $\BZ$, there is a  family of Hamiltonians parameterized by $P$ whose phase diagram is $PD$ which when on the wall between one domain with Chern number $d$ to another domain with Chern number $d'$ has $|d-d'|$ Dirac points. This gives an effective lower bound on the number of Dirac points, which is sharp for two band systems.
\end{theorem}

\begin{proof}
We first reduce to the case that $\BZ=S^2$ by using the degree $1$ map $\pi$ obtained by collapsing the 1--skeleton of $\BZ$.

The basic ingredient is the standard spin $\frac{1}{2}$ 
family $H(\bk)=-\bk\cdot \bs$, which has Chern number $C_1=1$ for the ground state. We can change the Chern number to $C_1=d$ by pulling back along a map $f_d:S^2\to S^2$ of degree $d$. Note for concreteness, we can use $f_d=Sz^d$. These are the locally constant maps on the components of $P_0$.
The local models at points of $p\in B_{\rm deg}$ are given by the construction \S\ref{par:wall}, and \S\ref{par:polyhedral} below. We can glue these together using local partitions of unity.

These local models satisfy the condition for wall crossings, as  the  $\delta=|d-d'|$ points where the Hamiltonians are degenerate are indeed Dirac points. First they are a discrete set, by \cite[Theorem 1.2]{KKWK2012}.
One can compute the local charges using the linear information \cite[Corollary 2.4]{KKWK20}. The sign of the linear information given by the sign of the Jacobian of the family around each level crossing which is $\mathrm{sgn}(d'-d)$ and hence carries local charges $\pm 1$. There are $\delta$ many of these points.
These local charges correspond to the jumps in the Chern number along a path $[-\epsilon,\epsilon]\to P$ that passes through a wall, \cite[\S 3.5.9.]{KKWK16}. The minimum number of Dirac points needed is $\delta$ which is attained.  
\end{proof}

Note that the families from the Theorem are generic in the sense that on the wall the level crossings are at isolated points, charges are $\pm 1$ and there are no annihilations, that is spurious Dirac pairs, or higher Chern number crossings. The latter two  are non--generic, as can be seen from the characteristic map, see \cite{KKWK2012}. Without the  extra conditions for wall crossings on the families, if $P$  is metrizable, one can realize an abstract phases diagram by  taking the families coming from pull--back with $f_d$ inside the connected components and then scale them by the distance to the degenerate locus. Each $\tilde{f_d}(p)=
\inf_{q\in P_{\rm deg}} dist(p,q)f_d(p)$ will then go smoothly to zero on the walls; consequently will have a continuum of degenerate Hamiltonians in these families if the walls are of positive dimension.
Non--generic crossing do appear for instance in the Kagome lattice, see \S\ref{par:kagome}, which exhibits 2-band crossings with local charge $2$, and the Gyroid, which has spin $1$ points, that is $3$ band crossings, see \cite{KKWK20}.

    

\subsubsection{Wall crossings}
\label{par:wall}

Let $I=[0,1]$. 
The following construction give rise to families on $I\times S^2$ regarded as a one parameter family. 
To get a phase transition from Chern number $d$ to $d'$  consider 
$F=[(1-t)f_d+tf_{d'}]:I\times S^2\to B^3$, where $B^3$ is 
the 3-ball, that interpolates between a degree $d$ map $f_d$ and a degree $d'$ map $f_{d'}$. 
If for some $p$,  [$(1-t)f_d(p)+tf_{d'}(p)]=O$,
then $(1-t)f_d(p)=t f_{d'}(p)$, and since $||f_d(p)||=||f_{d'}(p)||=1$
this means that $1-t=t$ and hence $t=\frac{1}{2}$. Also at $t=\frac{1}{2}$ there must be some $p$ where  $F(\frac{1}{2},p)=0$, because otherwise the degree cannot change.

Reparameterizing $[0,1]$ to $[-1,1]$ we can pull back this family to neighborhoods of points in $p\in P_{\rm deg}$ which locally  are modeled by two regions in $R^d$ separated by the wall $x_d=0$, by pulling back along the projection $\pi_2:R^{d-1}\times I\to I$.

\subsubsection{Polyhedral cone}
\label{par:polyhedral}
If the local structure near $p\in P_{\deg}$ is given by the polytope $P$,
then let $N$ be the normal fan, \cite[7.1]{ZieglerGTM}. This is the fan in the dual space $(\R^d)^*$
consisting of the cones $\{C_{F}\}_{F\in \operatorname {face} (P)}$
where 
$C_{F}=\{w\in (\mathbb {R} ^{d})^{*}\mid F\} \subseteq \{x\in P: w(x)={max}_{y\in P}w(y)\}$.  This is also the face fan of its polar polytope, \cite[7.3]{ZieglerGTM}.
If the convex generators for $C_F$ are $v_1,\dots, v_k$, we define the pull--back function  as
$\sum_{i=1^k} t_if_i$, where $f_i$ is the function in the cone which is dual to the ray $v_i$ and $t_i$ are the coordinate functions of $C_F$. This realizes all wall crossings at once, viz.\ $p=\sum_{i=1}^k t_i(p)V_i$.

\subsubsection{Standard wall crossing and rose curves}
\label{par:rose}
For the 1-parameter crossing,
choosing the standard maps $f_d=Sz^d$, where $z=e^{i\phi}\in S^1\subset \C$, we have that $F=SG$ where  $G(t,z)=[(1-t)z^d+tz^{d'}]$ as a map $[0,1]\times S^1\to \R^2$. Now $G(t,z)=(0,0)$ if and 
only if $t=\frac{1}{2}$, and in this case, the curve is $G(\frac{1}{2},z)=g(z)=\frac{1}{2}[z^d+z^{d'}]$. The number of times that the curve passes through zero, that is the number of Dirac 
points, is given as $\delta=|d-d'|$. This is 
seen by noting that $z^{d}+z^{d'}=0$ if and only if $z^{\delta}=-1$ whose solutions are the odd 
powers of a primitive $2\delta$--th roots of unity $\exp(i\frac{2\pi}{2\delta}(2k+1))$, $k=0,\dots, \delta-1$.
Physically this means that the family has $\delta$  crossings corresponding to Dirac points.
These happen at Eigenvalues $\lambda=0$.
Identifying this with the energy and in common situations, such as graphene or Weyl semi-metals, this number can also be seen as the intersection with the Fermi energy $E_F=0$.
One can also view this as pulling $|d-d'|$ strands through $0$, or alternatively by change of reference point pulling the origin through $|d-d'|$ strands to change the winding number. 
Thus $|d-d'|$ is an upper bound on the minimal number of Dirac points in a wall crossing.
As the jumps from local charges in a 2--band system are only by $\pm 1$ for each Dirac point, the number is minimal for these systems and family provides the minimal model for the phase transition.

Rewriting in real coordinates, 
we get the equations
\begin{equation}
    g(\phi)=(\cos(\frac{d-d'}{2}\phi)\cos(\frac{d+d'}{2}\phi),
\cos(\frac{d-d'}{2})\sin(\frac{d+d'}{2}\phi) )
\end{equation}
This is part of a classical curve called the rose curve, i.e.\ it satisfies the polar equation $r=\cos(k\theta)$ with $k=|\frac{d-d'}{d+d'}|$, if $d\neq -d'$.
This is seen by computing $r^2=g(z)\overline{g(z)}$ and $\theta= \arctan[\Im(g(z))/\Re(g(z))]$.
The parametrized curve is parametrized either in the parameter $\phi\in [0,2\pi]$ or in the parameter $\theta=\frac{d_1+d_2}{2}\in [0,(d_1+d_2)\pi]$.
To transverse a rose with $k=\frac{r}{s}$
with $gcd(r,s)=1$ exactly once fully, one needs to work in the interval $[0,2s]$ if both $r$ and $s$ are odd and  $[0,s]$ else. 

In the case $d=0$, $k=0$ and the curve is simply the circle $\frac{1}{2}+\frac{1}{2}z$ of radius $\frac{1}{2}$ centered at $(\frac{1}{2},0)$.
This is traversed $d$ times ---to go around once one needs the interval $[0,\pi]$, and the interval of $\theta$ is $[0,d\pi]$.

If $d=-d'$, the curve collapses to the curve $(\cos(d\phi),0)$ at the parameter $t=\frac{1}{2}$.
Note that this passes through $(0,0)$ exactly 2d times. In particular, for $d=1$ this is at $(\frac{\pi}{2},\frac{3\pi}{2})$.

For the special case $d=1,d'=2$ we have  $k=-\frac{1}{3}$ and the curve is the lima\c{c}on trisectrix on the interval $[0,3\pi]$ which traverses the whole curve once hitting the origin once. For $d=1,d'=-2$, $k=-3$ and the curve is the trifolium, for the interval $[0,\pi]$, which traverses the curve once and passes three times through the origin.
For $d=-2,d'=4$, $k=\frac{-6}{2}=-3$ we again get the trifolium, but now on the interval $[0,2\pi]$ which traverses the trifolium twice
and passes six times through the origin.

These examples are given in Figure \ref{fig:rose}.

\begin{figure}[h]

\begin{subfigure}[t]{0.3\textwidth}
{\includegraphics[height=1.2
in]{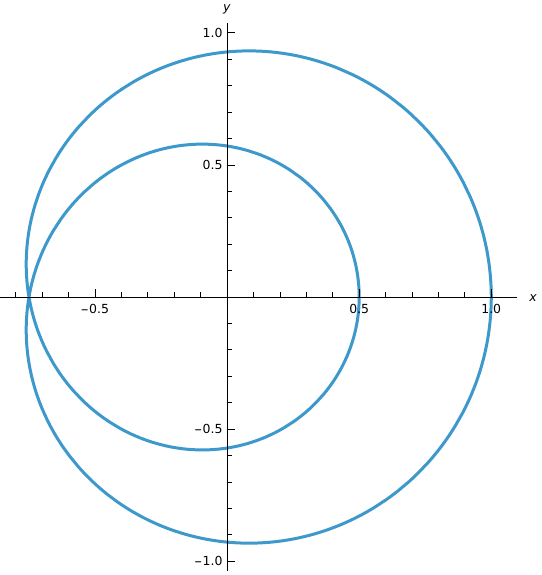}}
    \caption{$t=\frac{1}{4}, d=1, d'=2$}
\label{fig:A}
\end{subfigure}\hfill
\begin{subfigure}[t]{0.3\textwidth}
{\includegraphics[height=1.2
in]{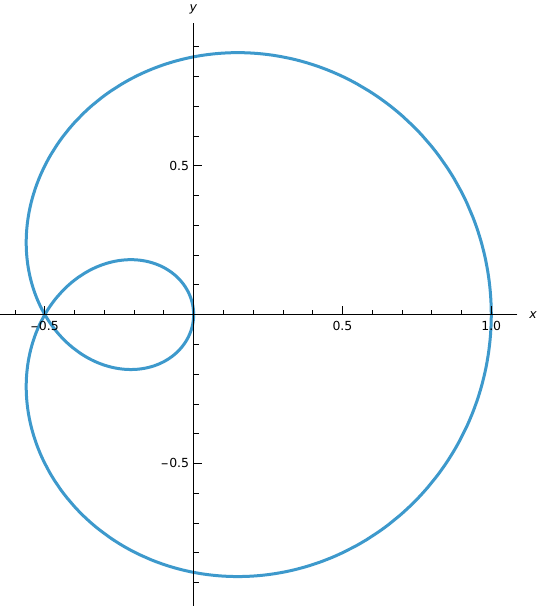}}
    \caption{$t=\frac{1}{2}, d=1, d'=2$}
\label{fig:B}
\end{subfigure}\hfill
\begin{subfigure}[t]{0.3\textwidth}
{\includegraphics[height=1.2
in]{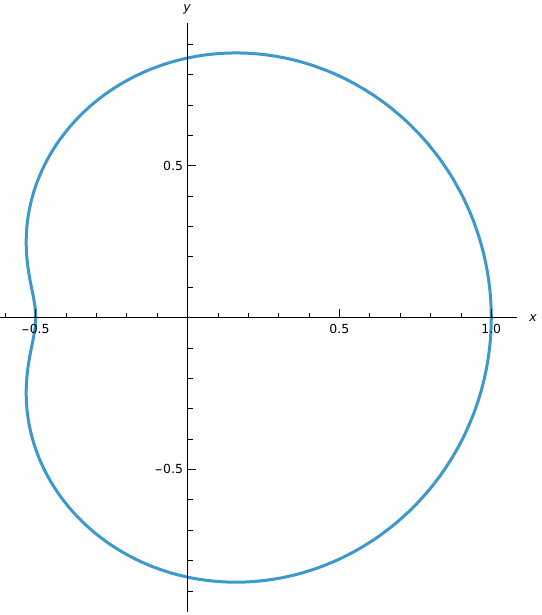}}
    \caption{$t=\frac{3}{4}, d=1, d'=2$}
\label{fig:C}
\end{subfigure}\hfill

\begin{subfigure}[t]{0.3\textwidth}
{\includegraphics[height=1.2
in]{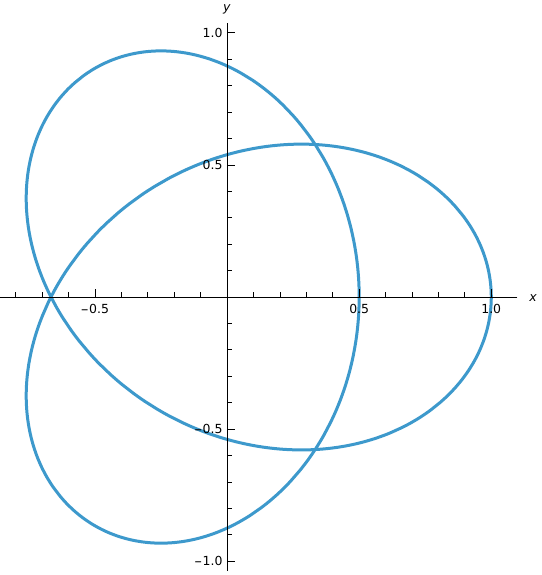}}
    \caption{$t=\frac{1}{4}, d=1, d'=-2$}
\label{fig:D}
\end{subfigure}\hfill
\begin{subfigure}[t]{0.3\textwidth}
{\includegraphics[height=1.2
in]{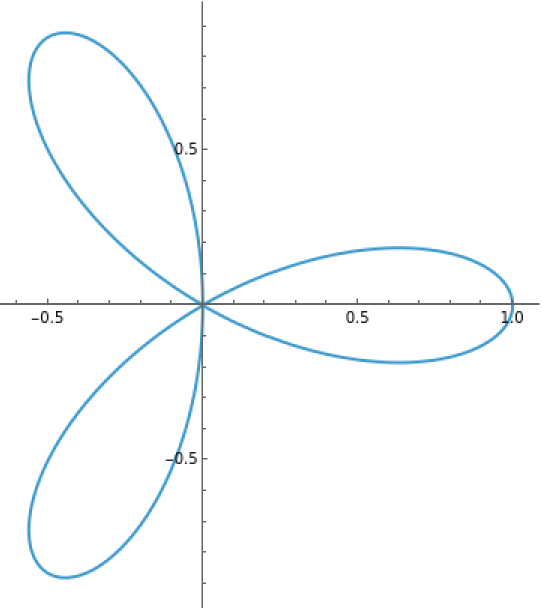}}
    \caption{$t=\frac{1}{2}, d=1, d'=-2$}
\label{fig:E}
\end{subfigure}\hfill
\begin{subfigure}[t]{0.3\textwidth}
{\includegraphics[height=1.2
in]{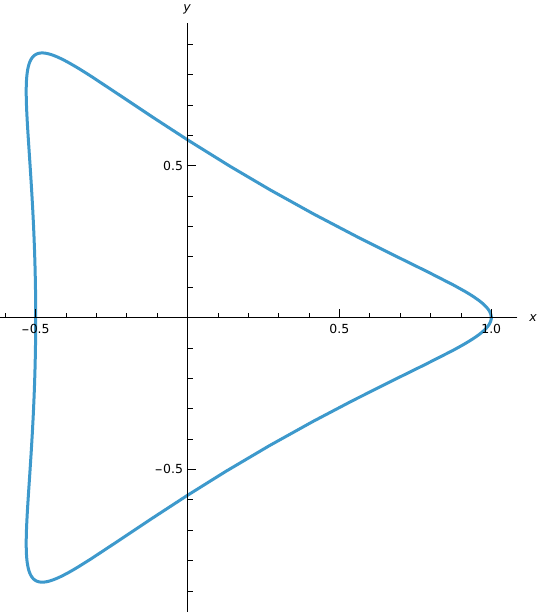}}
    \caption{$t=\frac{3}{4}, d=1, d'=-2$}
\label{fig:F}
\end{subfigure}\hfill

\begin{subfigure}[t]{0.3\textwidth}
{\includegraphics[height=1.2
in]{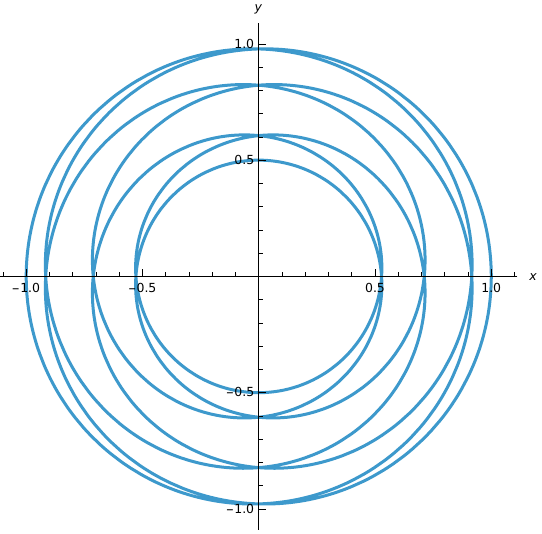}}
    \caption{$t=\frac{1}{4}, d=5, d'=7$}
\label{fig:G}
\end{subfigure}\hfill
\begin{subfigure}[t]{0.3\textwidth}
{\includegraphics[height=1.2
in]{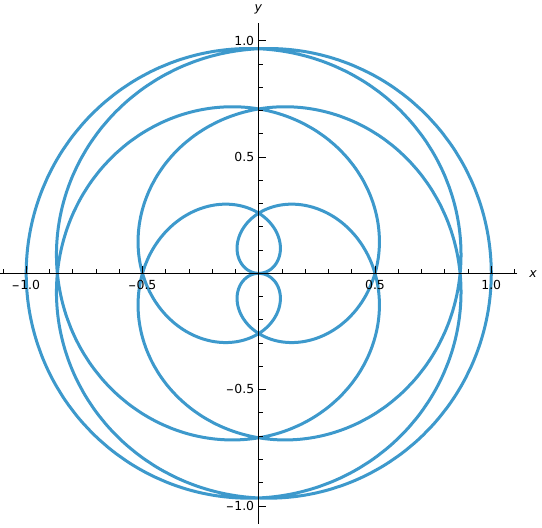}}
    \caption{$t=\frac{1}{2}, d=5, d'=7$}
\label{fig:H}
\end{subfigure}\hfill
\begin{subfigure}[t]{0.3\textwidth}
{\includegraphics[height=1.2
in]{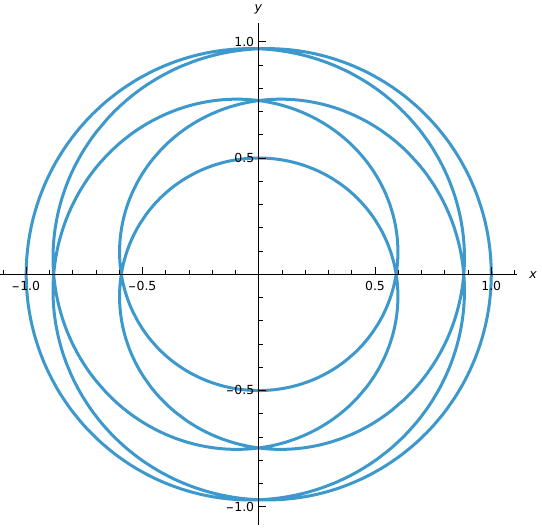}}
    \caption{$t=\frac{3}{4}, d=5, d'=7$}
\label{fig:I}
\end{subfigure}\hfill

 \caption{Rose curves for different values of $t,d$ and $d'$. 
 The origin is in the image for $t=\frac{1}{2}$. The winding numbers are constant for $0\leq t<\frac{1}{2}$ and $\frac{1}{2}<t\leq 1$. These can be readily be read off by using a ray.}\label{fig:rose}
 
\end{figure}

\subsection{Other choices and examples}
One can of course use any maps of the given degrees, not just the suspension maps to define the straight line homotopy $F$.
In the case of a jump from $+1$ to $-1$, we can use the identity $id$ as the degree $1$ map and any map
with image $S^2$ which is degree $-1$ when restricted to its image. 
  The two standard choices for this are the antipodal map $a:\bk\to -\bk$ or the map $a_z:(k_x,k_y,k_z)\to (k_x,k_y,-k_z)$. 
For $a$ there is a singularity at $t=1/2$ where the whole map collapses to the origin $F(1/2,\bk)\equiv O$. 
For $t\neq 1/2$ the image is in $\R^3\setminus \{O\}$ and the map has a mapping degree, which is $1, t<1/2$ and $-1, t>1/2$.
We obtain the family $(1-2t)\bk=\alpha\bk$, where $\alpha=1-2t\in [-1,1]$. As $\frac{|\alpha|}{2}$ is the distance of $t$ to $\frac{1}{2}$ this is an example of scaling the map to zero at the critical locus.

In the second case, also for $t\neq 1/2$ the image is in $\R^3\setminus \{O\}$ and the map has a mapping degree, which is $1, t<1/2$ and $-1, t>1/2$. For $t=1/2$ the image is the standard disc which contains the origin $O$.
Writing this as
\begin{equation}
k_x \s_x+k_y\s_y+m\sigma_z    
\end{equation}
where $m=(1-2t)k_z$, we see that the mass gap vanishes at $t=1/2$.
This particular model appears as a Dirac model, cf.\ \cite[2.3]{KLWK16}.
Note that in both cases, there is a band crossing at the critical point, $O$ is in 
the image, and this has to be the case, since otherwise the Chern class  would be
invariant. The perceived contradiction is lifted by the fact that there is no 
splitting for the two line bundles involved in the crossing, if a crossing point 
is present.

We illustrate the construction \S\ref{par:polyhedral}
for a 2-parameter family with $k$ rays. Up to isotopy fixing the origin and the incidence relations, we can arrange the rays to be in the directions of the $k$--th roots of unity $\zeta_k$. Then the normal cone is given by the odd $2k$--th roots $\zeta_{2k}^{2n+1}$. If in the chamber between $\zeta_k^{i-1}$ and $\zeta_k^i$ the $C_1=d_i$,  then in the chamber between $\zeta_{2k}^{2i-1}$ and $\zeta_{2k}^{2i+1}$ for $p=a_1\zeta_{2k}^{2i-1}+ a_2\zeta_{2k}^{2i+1}$ the function $F$ is given by 
$a_1(p)f_{d_i}+a_2(p)f_{d_{i+1}}$. Similar to the 1--parameter calculation  $F(p)=O$ means that $a_1(p)=a_2(p)$ and hence $p$ is on the ray $\zeta_{2k}^{2i}=\zeta_k^i$ and hence is on the critical locus.

We now consider  a specific 2-parameter model coming from spin--orbit coupling, cf.\ \cite{KKWK16}, given  by 
\begin{equation}
\label{eq:MBphase}
    H(\bk)=(\sin (k_x), \sin(k_y),M-B(\sin^2(k_x/2)+\sin^2(k_y/2))\cdot \bs
\end{equation}

\begin{figure}[htbp]
 \includegraphics[width=0.48\textwidth]{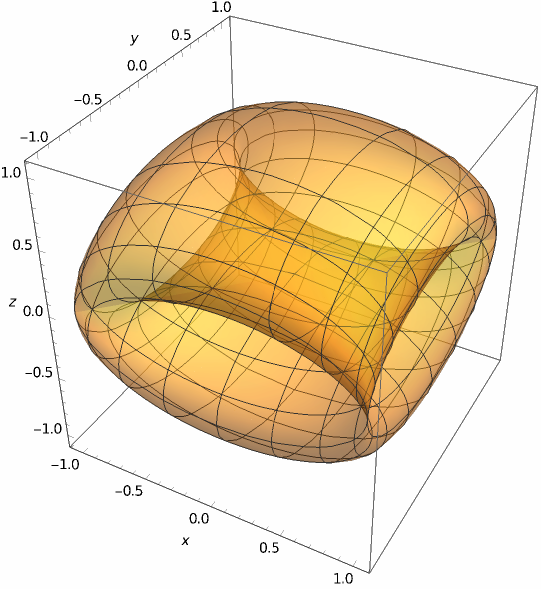}\caption{\label{fig:surface} Surface plot of the function given in Eq.~(\ref{eq:MBphase}) for $M=B=1$.}
\end{figure}

The function from $T^2\to \R^3$ is shown in Figure \ref{fig:surface}. It misses the origin if $M\neq0,M\neq B,M\neq 2B$.
The intersection with the z-axis is given by 
$(0,0,M), (0,\pi,M-B), (\pi,0,M-B), (\pi,\pi, M-2B)$. The indices can be computed as $\cos(k_x)\cos(k_y)$ which are $+1,-1,-1,+1$ respectively.
There are 6 regions; taking into account how many of the 4 points above lie on the ray  $\lambda(0,0,1),\lambda>0$, we obtain Table \ref{tab:phases}.
The corresponding phase diagram is shown in Figure \ref{fig:phasediagram}.
It is a straightforward check that the wall transitions are minimal and of rose--type.

\begin{figure}[htbp]
 \includegraphics[width=0.48\textwidth]{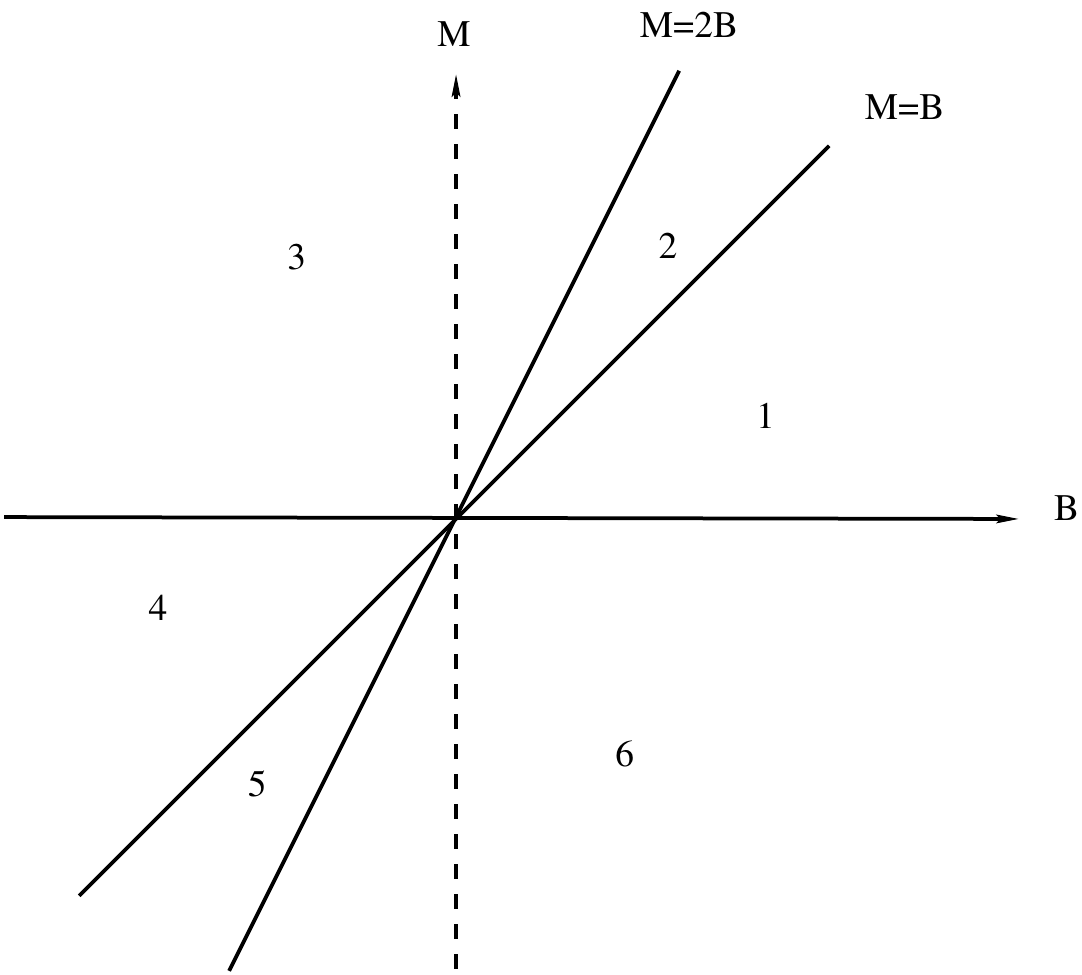}\caption{\label{fig:phasediagram} Phase diagram corresponding to the phases listed in Table \ref{tab:phases}.}
\end{figure}

\begin{table}[h]
    \centering
    \begin{tabular}{c|l|r}
    Region&Hyperplanes&Chern\\
    &&number\\
    \hline
        $1$&$B>M, M>0$ & $1$ \\
        $2$ & $2B>M, M>B$&$-1$\\
        $3$&$M>0, M>2B$&$0$\\
        $4$&$0>M,M>B$&$-1$\\
        $5$&$M<B,M<2B$&$1$\\
        $6$&$M<2B, M<0$&0
    \end{tabular}
    \caption{Chern number of the open regions in the phase diagram of \eqref{eq:MBphase}}
    \label{tab:phases}
\end{table}

\section{Lattice Implementation}
\label{par:lattices}
For the lattice implementations, we work in the usual physical setup of tight-binding models with creation and annihilation operators to realize the hopping terms and other interactions. This facilitates the connection to the literature.
In particular, translation from site $i$ to site $i+k$ in a one--particle system is given  by $U_k=c_{i+k}^\dagger c_j$. For instance in $\Z^2$ there are the operators $U_1=c_{i+1,j}^\dagger c_{i,j}$ and $U_2=c_{i,j+1}^\dagger c_{i,j}$.

We will achieve higher Chern numbers by using commensurate higher order neighborhoods, which correspond to coverings \eqref{eq:covers}.
If the Hamiltonians already involve higher order interactions as the Haldane Hamiltonian, see \S\ref{par:haldane} below, we require that all the participating $n$--th neighborhoods are commensurate in the super--lattice, that is they are all scaled by the same integer $l(k)$.
In particular, given a Hamiltonian with interactions terms $H^{(n)}$, under the condition that all these sublattices are commensurate, under multiplication by $l(k)$ the interaction terms will become $l(k)H^{(n)}=H^{(kn)}$.




We will first show that in the case of lattices from quadratic integers, there are infinite families of commensurate sublattices, see \S\ref{par:QIL}.
After reviewing Haldane's construction in our general setup \S\ref{par:haldane}, we turn to implement our strategy for natural models with higher Chern numbers  for square and triangular, which are lattices of quadratic integers. 
Subsequently we will treat then the honeycomb and Kagome lattices which are triangular-based lattices. By triangular-based we mean lattices that are a subset of a triangular lattice, as will be discussed. These examples show how the construction can be easily extended to other lattices like Lieb, Dice and checkerboard lattices.

\subsection{Quadratic Integer Lattices}
\label{par:QIL}
Importantly, for 2D lattices described by rings of quadratic integers $\Z[\omega]=\{a+\omega b: a, b \in \mathbb{Z}\}$, where $\omega = \frac{1+\sqrt{-d}}{2}$ if $-d \equiv 1 \pmod{4}$ and $\omega = \sqrt{-d}$ otherwise, as we show in Appendix \ref{par:quadratic}, there exist arbitrarily far distant-neighbors with the same structure as the original ones. The tool we use are norms, where the norm after embedding into $\C$ -- that is identifying $\sqrt{-d}=i\sqrt{d}$ -- is the norm of the element as a complex number, cf \S\ref{par:qintro}.
A unit in ring quadratic integer ring is an invertible element $u \in \Z[\omega]$. These are the elements of norm  $1$. 
An associate $\mu \in \Z[\omega]$ of an element $\nu \in \Z[\omega]$ is a quadratic integer that satisfies the relation $\nu=u\cdot \mu$, for some unit $u$. The associates of a number is the set of all its associates and this is a finite set for $\Z[\omega]$ \cite{IrelandRosen1990}. In fact, the number of units $|U|=2,4$ or $6$ depending on $d$, see \eqref{eq:units}.  From the multiplicativity of the norm, all associates will have the same norm. 
\begin{definition}
An element $\mu \in \Z[\omega]$ is said to have an isolated norm if the only numbers with the same norm as
 $\mu$ are its associates.     
\end{definition}
 The elements whose norm is isolated span a sublattice which is a homothety of the original lattice by the square root of the norm, up to a possible rotation.
From \ref{corollary:generalisolated} and Theorems \ref{theorem:Heegner} and \ref{theorem:isolatednorms}, we obtain:

\begin{theorem}\label{theorem:main2}
    In a full 2D lattice $\Gamma$ described by a quadratic integer ring $\Z[\omega]$,  there are 
    exactly $|U|$ nearest neighbors to any lattice 
    point lying at a Euclidean distance  $d = 
    \prod_{i=1}^{r} p_i^{b_i}\prod_{j=1}^{s} 
    q_j^{c_j/2}$ where each $p_i$ is an inert prime 
    $p_i \in \Z$ in $\Z[\omega]$, and $q_j \in 
    \mathbb{N}$ is a ramified prime represented by 
    the norm function, $b_i,c_j,r,s \in \mathbb{N}$, 
    and $|U|$ is the cardinality of the set of units 
    in $\Z[\omega]$.
    In particular if 
$d \in \{-1, -2, -3, -7, -11, -19, -43, -67, -163\}$ then any ramified prime is represented and the elements with these norms generate exactly those lattices with the same minimal number of nearest neighbors as the original lattice. \qed
\end{theorem}
The details and definitions of split, ramified and inert primes are in \S\ref{par:qintro}.

We now illustrate this by treating topological insulators living on a square lattice $(d=-1)$, and  the triangular lattice ($d=-3$).  
\subsubsection{Gaussian Integers/square lattice}
\label{subsection:gaussian} 
Gaussian integers are the quadratic integers : $\mathbb{Z}[i]=\{a+b i: a, b \in \mathbb{Z}\}$. They form a square lattice in the complex plane and are a  unique factorization domain (UFD)  \cite{cox1989primes}. The norm of a Gaussian integer is $N(a+b\omega) = a^2 +b^2$. There are  four units $\{\pm 1,\pm i\}$ that have unit norm. The question of which norms are isolated in a quadratic domain boils down to the problem of finding a representation of integers as a certain quadratic form. In the case of Gaussian integers, non-isolated norms have a representation as the sum of two nonzero squares. This is a relatively old problem that was considered by Fermat. Many proofs for the theorem were given \cite{Euler1760,zagier1990one,cox1989primes}. The results of Appendix \ref{par:ufd} can be used as a proof when specializing to Gaussian integers. Using theorems~\ref{theorem:oddprimes} and \ref{theorem:evenprimes}, the only rational prime that ramifies (i.e. divides the discriminant $D=4$) in $\Z[i]$ is $p=2$. The primes of the form $p \equiv 3 \pmod{4}$ remain inert while those of the form $p\equiv 1 \pmod{4}$ split. A direct application of theorem~\ref{theorem:main2} in the case of Gaussian integers gives the following theorem.

\begin{theorem}[Fermat’s theorem on sums of two squares]\label{theorem:fermat}
    An odd rational prime $p$ can be represented as a sum of two nonzero squares $p=a^2+b^2$ for $a,b \in \mathbb{N}$ iff $p \equiv 1 \pmod{4}$.
\end{theorem} 

Restating this in terms of the distances of nearest neighbors in a square lattice, we have the following corollary.

\begin{corollary}\label{corollary:square4nns}
    In a square lattice, there are exactly four nearest neighbors to any lattice point at a Euclidean distance \( d = 2^{b_0/2}\prod_{i=1}^{r}p_i^{b_i}\), with positive integers $r,b_i \in \mathbb{N}$, and each \( p_i \) is a rational inert prime, i.e.\ \( p_i \equiv 3 \pmod{4} \).
\end{corollary}  

  Note that we will in fact restrict to the case where $b_0$ is an even integer. The parity of the number $b_0$ then divides this family of isolated norms into two classes whose lattice points are rotated from each other with an angle $45^{\circ}$; see Fig.~\ref{fig:squarenns} for an illustration of the first few isolated norms. For commensurate lattices we  take the parity of $b_0$ to be even which yields distances: $d \in \{1,2,3,4,6,7,8,9,11,12,14 \cdots\}$.

\subsubsection{Eisenstein Integers/triangular lattice}\label{subsection:eisenstein}
Eisenstein integers are the quadratic integers: $\mathbb{Z}[\omega]=\{a+b \omega: a, b \in \mathbb{Z}\}$ where $ \omega=e^{\frac{2 \pi i}{3}}=-\frac{1}{2}-\frac{\sqrt{3}}{2}i$ is the cubic root of unity satisfying $1+\omega +\omega^2 =0$.  Importantly Eisenstein integers form a UFD \cite{cox1989primes}. As an Abelian group they form a triangular lattice in the complex plane. The norm of an Eisenstein integer is $N(a+b\omega) = a^2 +b^2 -ab$. There are six units $\{\pm 1,\pm \omega, \pm \omega^2\}$ that have unit norm.
 Moreover, $\mathbb{Z}[\omega]$ is an Euclidean domain where irreducible elements are primes \cite{cox1989primes}. By theorems~\ref{theorem:oddprimes} and ~\ref{theorem:evenprimes}, a rational prime will ramify if it divides the field discriminant which is $-3$. Consequently, the only rational prime that ramifies is $p=3$. Furthermore, for primes other than $3$, if $p \equiv 2 \pmod{3}$ it remains inert, otherwise it splits. We have the following direct application of theorem~\ref{theorem:main2}.
 
\begin{corollary}\label{corollary:eisenisolated}
        In a triangular lattice, there exist exactly six nearest neighbors to any lattice point at a Euclidean distance \( d =3^{b_0/2} \prod_{i=1}^{r} p_i^{b_i} \), for positive integers \( r, b_i \in \mathbb{N}\), and each \(p_i \) is a rational inert prime \( p_i \equiv 2 \pmod{3} \).
\end{corollary}

 Note that similar to the Gaussian integers case, the parity of the number $b_0$ then divides this family of isolated norms into two classes whose lattice points are rotated from each other with an angle $30^{\circ}$; see Fig.~\ref{fig:triangularnns} for an illustration of the first few isolated norms. Thus to obtain a commensurate sublattice, we will restrict to the case where $b_0$ is an even integer in the above corollary, which yields commensurate lattices at distances: $d \in \{1,2,3,4,5,6,8,9,10,11,12,16 \cdots\}$.
  

\subsection{Haldane model on a Honeycomb lattice} 
\label{par:haldane}
A prototypical example of a tight-binding Hamiltonian that realizes a topological insulator is the two-band Haldane model \cite{haldane88}. It was the first model that realizes the integer quantum Hall effect without using an external magnetic field. The model breaks time-reversal symmetry using complex hopping terms. This is necessary to realize a non-zero Hall conductance \cite{haldane88}. This Hall conductance is quantized and given by the first Chern number \cite{Thouless1982}. The model lives on a two-species honeycomb lattice with vertices shown as blue and red circles in Figure ~\ref{fig:haldanemodel}. The tight-binding Hamiltonian has real interactions between first nearest neighbors: $t_1$ and complex interactions between second nearest neighbors: $t_2 e^{\pm i\phi}$, in addition to a species-dependent on-site potential $\pm m$.

\begin{figure}[h]
    \centering
    \includegraphics[scale=1.4]{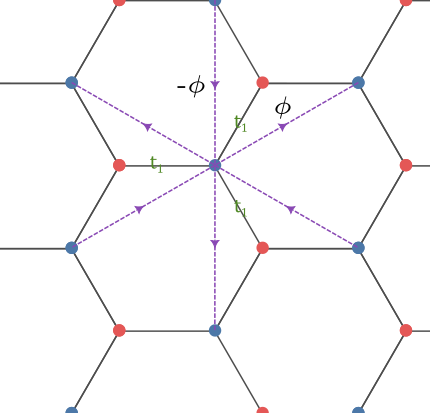}
    \caption{Haldane model on a honeycomb lattice with two different atom species (blue and red). The interactions are \( t_1 \) between the first nearest neighbors and \( t_2 e^{\pm i\phi} \) between the second nearest neighbors using the convention shown.}
    \label{fig:haldanemodel}
\end{figure}

The Hamiltonian then reads:
\begin{equation}\label{equation:HaldaneH}\begin{aligned}
 H_{\text{Haldane}} =&t_1\sum_{\langle i, j\rangle} (c_{i,A}^{\dagger} c_{j,B}+c_{i,B}^{\dagger} c_{j,A})+t_{2} \sum_{\langle\langle i, j\rangle\rangle} ( e^{ \pm i \phi} c_{i,A}^{\dagger} c_{j,A}+e^{ \mp i \phi} c_{i,B}^{\dagger} c_{j,B}) \\&+m \sum_i (c_{i,A}^{\dagger} c_{i,A}-c_{i,B}^{\dagger} c_{i,B})
\end{aligned}\end{equation} 

The summations ${\langle i, j\rangle}$ and ${\langle\langle i, j\rangle\rangle}$ are over first and second-nearest neighbors in real space, respectively. Here, the variables $\{i,j\}$ denote lattice site positions and A(B) denotes blue(red) sites. The convention for $\phi$ is illustrated in Figure ~\ref{fig:haldanemodel}. It is taken such that the total flux in the unit cell is zero. The $t_2$ terms break time-reversal symmetry thus allowing for a non-zero Chern number. Mathematically this avoids having a degree zero map. We set $t_1=1$ without loss of generality. This defines a 3-parameter family of Hamiltonians depending on real variables $\{t_2,\phi,m\}$.  
Note, that since the Honeycomb lattice has two equivalent sublattices, choosing an equivalence to a Hilbert space 
$\mathcal{H'}$ one has
$\mathcal{H}=\mathcal{H}_A\oplus \mathcal{H}_B\simeq \mathcal{H'}\ot \C^2$.
One defines momentum-space creation and annihilation operators $\{c_{A,k}, c_{B,k}\}$ which create/annihilate a state with momentum $k$ living on the A(B) sublattice:
\begin{equation}
 c_{A,k}\coloneqq\frac{1}{\sqrt{N}} \sum_{r} e^{-i k \cdot r} c_{r,A}, \quad c_{B,k}\coloneqq\frac{1}{\sqrt{N}} \sum_{r} e^{-i k \cdot r} c_{r,B}
\end{equation}   The $r$ summation here is over all sublattice sites. These operators are just the Fourier transform of the real-space operators $\{c_{i,A}, c_{j,B}\}$. For each momentum $k$ we can define a basis which consists of the state with momentum $k$ living on one of the two sublattices:  

\begin{equation}\label{equation:basishoneycomb}
\left(\begin{array}{c}
 c_{k} 
\\
0
\end{array}\right)  \coloneqq c_{k,A} \ket{0}\in \mathcal{H}_A  , \quad \left(\begin{array}{c}
0 
\\
 c_{k}
\end{array}\right)  \coloneqq c_{k,B} \ket{0}  \in \mathcal{H}_B \end{equation}   

Here, $\ket{0}$ is the vacuum state of the Hilbert space with no states occupied. 
 After the Fourier transform, the Hamiltonian will be a summation of $2\times2$ Hermitian matrices for each momentum $k$. As an example, we provide the explicit Fourier transform of the first term of the Hamiltonian given by Eq.~\eqref{equation:HaldaneH}.

\begin{equation}\label{equation:haldanefourier}
\begin{aligned}
&H_{1^{st} \text{ term}}(k)  = \sum_{\langle i, j\rangle} (c_{i,A}^{\dagger} c_{j,B}+h.c.)\\ 
& =\sum_i \sum_{m=1}^3\left(\frac{1}{\sqrt{V}} \sum_k e^{i k \cdot r_i} c_{A, k}^{\dagger}\right)\left(\frac{1}{\sqrt{V}} \sum_{k^{\prime}} e^{-i k^{\prime} \cdot\left(r_i+a_m\right)} c_{B, k^{\prime}}\right) +h.c.\\
& =\frac{1}{V} \sum_i \sum_{m=1}^3 \sum_{k, k^{\prime}} e^{i k \cdot r_i} e^{-i k^{\prime} \cdot\left(r_i+a_m\right)} c_{A, k}^{\dagger} c_{B, k^{\prime}} +h.c.\\
&= \frac{1}{V}  \sum_{m=1}^3 \sum_{k, k^{\prime}} \left(\sum_i e^{i\left(k-k^{\prime}\right) \cdot r_i} \right) e^{-i k^{\prime} \cdot a_m}c_{A, k}^{\dagger} c_{B, k'} +h.c.\\
&= \sum_k \sum_{m=1}^3 e^{-i k \cdot a_m} c_{A, k}^{\dagger} c_{B, k} +h.c.\\
&= \sum_k\sum_{i=1}^3 \cos \left(k \cdot a_i\right) \sigma_1 + \sum_k\sum_{i=1}^3 \sin \left(k \cdot a_i\right) \sigma_2
\end{aligned}
\end{equation}

In the second line, $r_i$ represent the lattice vectors, $a_m$ represent nearest-neighbor vectors and $V$ is the total number of lattice points. To go from line 4 to line 5, we used $\sum_i e^{i\left(k-k^{\prime}\right) \cdot r_i}=V \delta_{k, k^{\prime}}$. Rewriting the Hamiltonian $H(k)$ for each momentum $k$ in terms of \eqref{equation:basishoneycomb} the Hamiltonian will be a $2\times2$ Hermitian matrix 
and can be expanded in terms of Pauli matrices and the identity matrix (denoted  here by $\sigma_0$): 
\begin{align} \label{equation:h(k)}
    H (k) &=  h_{0}({k})  \sigma_0 +\sum_i h_i({k}) \sigma_i  \\
    h_0({k}) &= 2 t_2 \cos \phi \sum_{i=1}^3 \cos \left(k \cdot b_i\right),  \quad
    h_{1}({k}) = \sum_{i=1}^3 \cos \left(k \cdot a_i\right)\nonumber\\  
     h_{2}({k}) &=\sum_{i=1}^3 \sin \left(k \cdot a_i\right), \quad
    h_{3}({k}) = m - 2 t_2 \sin(\phi) \sum_{i=1}^3 \sin \left(k \cdot b_i\right)\nonumber
\end{align}
The vectors $\{a_i,b_i\}$ are the first and second-nearest neighbors vectors respectively. The coefficients $h_i$ are evidently real from the hermiticity of the Hamiltonian. Further, we can ignore $h_0({k})$ altogether as it does not affect the energy Eigenstates or the topological properties of the model; see Theorem \ref{theorem:ksigma}. The Hamiltonian in momentum space defines a map from the periodic 2D Brillouin zone $T^2$ to $\R^3$. If the spectrum is non-degenerate, which we demand for an insulator, then it is a map to $\R^3 \setminus\{0\}$. The first Chern class can then be computed using the pull back of this map; see Theorem \ref{theorem:ksigma}. 

The Eigenvalue, aka.\ energy spectrum, is $E_\pm = \pm \sqrt{h_1^2+h_2^2+h_3^2}$. It is easily seen that the Hamiltonian is gapped (non-degenerate) for any value of $\{m,t_2,\phi\}$ except at two points in the Brillouin zone: $K=\left( \frac{4\pi}{3 \sqrt{3}}, 0, 0 \right)$ and $K'=\left( -\frac{4\pi}{3 \sqrt{3}}, 0, 0 \right)$. The two terms $h_1$ and $h_2$ vanish at these two momentum values and these points are called pre-Dirac point. In the limit where $h_3$ also vanishes then they become  Dirac points and the Hamiltonian becomes degenerate. This happens along the two curves: \begin{equation}
m = \pm 3\sqrt{2}\, t_2 \sin{\phi}
\end{equation} 
These curves are the phase transition curves where the first Chern class over the ground state is not well defined. The positive branch $m = 3\sqrt{2}\, t_2 \sin{\phi}$ corresponds to crossing the origin at $K$ (the negative branch is at $K'$).  Since this is 2--band system by Theorem \ref{thm:main} there need to be at least two Dirac points and this example is minimal.
In all other regions, the first Chern number is defined for the ground state and can be computed using either the connection: $C=\frac{1}{2 \pi} \int_{B} d^2 k F_{12}(k)$ or using the mapping degree: $C=deg(h)=\frac{1}{4 \pi} \int_{B}  \hat{h}\cdot\left(\partial_{k_x}\hat{h} \times \partial_{k_y} \hat{h}\right) dk_x\, dk_y = \sum_{p\in h^{-1}(q)}sgn(\phi_p)$ as discussed in \S\ref{subsection:calcdegree}. Here, $\hat{h}$ is the normalized map $\hat{h} = \frac{(h_1,h_2,h_3)}{\sqrt{h_1^2+h_2^2+h_3^2}}$. The phase diagram can be obtained using any of the previously mentioned methods and is presented in Figure~\ref{fig:haldanephase}. A point in the phase diagram corresponding to a Chern number of $1$ will correspond to a map that covers the sphere $\hat{h}(k)$ exactly once. Such a covering is shown in Figure~\ref{fig:HaldaneSphere}. Again this is commensurate with the general theory.
Moreover  according to \ref{par:rose}, the wall crossings are between Chern numbers $0,\pm 1$ which are modeled by the circle of \S\ref{par:rose} while the wall crossing $-1$ to $1$ along the line $\frac{m}{t_2}$ is a collapse to the interval, see {\it loc.\ cit.}.

\begin{figure}[h]
    \centering
    \includegraphics[scale=0.60]{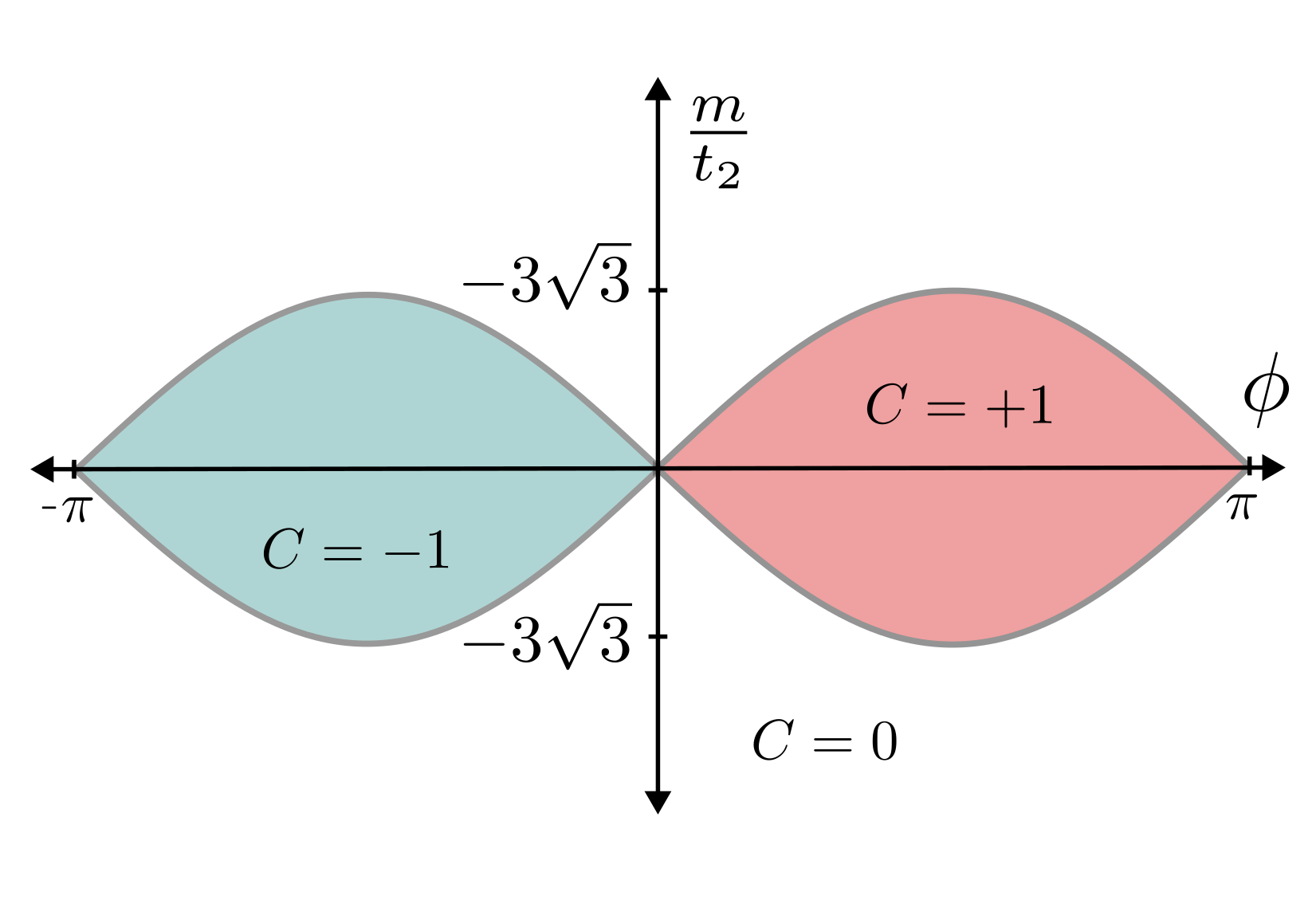}
    \caption{A slice of the 3-parameter family of Hamiltonians with \(t_2 = 1 \) in the Haldane Hamiltonian Eq.~\eqref{equation:HaldaneH}. The Chern number takes values {0,-1,1} for different values of the Hamiltonian parameters $\phi$ and $m$.}
    \label{fig:haldanephase}
\end{figure}
\begin{figure}[h]
    \centering
    \includegraphics[scale=0.6]{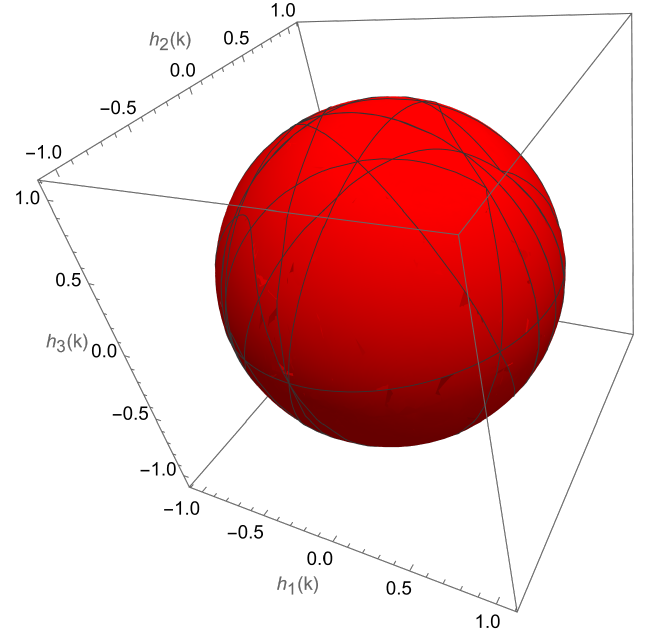}
      \caption{The image of the map $\hat{h}(k)$ from the Brillouin zone to the normalized Hamiltonian parameters for $m=0$ and $\phi=\frac{\pi}{2}$. This point corresponds to Chern number $C=1$ in the phase diagram in Figure~\ref{fig:haldanephase}. Hence, it shows a single covering of the sphere. }
    \label{fig:HaldaneSphere}
\end{figure}
\subsection{Physical Models with Higher Chern numbers $|C| > 1$}
\label{par:commensurate}
In this section, using concrete examples, we show how one can indeed achieve higher Chern numbers by using commensurate higher neighborhood interactions as discussed in \S\ref{par:supercell} using the classification of commensurate lattices given above in \ref{theorem:main2}. 
The two main physical examples are the square and triangular lattices as well as those based upon them, like the hexagonal and Kagome lattice.
Building a Hamiltonian with coupling such higher order terms  gives us access to phase transitions between topological phases with arbitrarily high Chern numbers.





\subsubsection{Square lattice}
A simple Hamiltonian that realizes a Chern insulator on a square lattice was introduced in \cite{BHZ2006}. The model lives on a square lattice with two orbitals per atom. It represents the topological insulator mercury-telluride when we restrict to spin-up electrons \cite{BHZ2006,zoology2012}. The Hamiltonian explicitly reads:

\begin{align}
H = & -\frac{t_1}{2}\sum_i \Big\{ \Big[ c_{i,A}^\dagger\, c_{A,i+\hat{x}} - i\, c_{i,A}^\dagger\, c_{B,i+\hat{x}} - i\, c_{i,B}^\dagger\, c_{A,i+\hat{x}} - c_{i,B}^\dagger\, c_{B,i+\hat{x}} \Big] \nonumber \\
& \quad + \Big[ c_{i,A}^\dagger\, c_{A,i+\hat{y}} - c_{i,A}^\dagger\, c_{B,i+\hat{y}} + c_{i,B}^\dagger\, c_{A,i+\hat{y}} - c_{i,B}^\dagger\, c_{B,i+\hat{y}} \Big] + \text{h.c.} \Big\} \nonumber \\
& + m \sum_i \Big[ c_{i,A}^\dagger\, c_{i,A} - c_{i,B}^\dagger\, c_{i,B} \Big].
\end{align}

Here, $A$ and $B$ are the two orbitals of each atom. It can be written more succinctly as:
\begin{align} \label{equation:squareHshort}
H & =-t_1 \sum_i \left(c_i^{\dagger} \frac{\left(\sigma_z-i \sigma_x\right)}{2} c_{i+\hat{x}}+c_i^{\dagger} \frac{\left(\sigma_z-i \sigma_y\right)}{2} c_{i+\hat{y}}+\text { h.c. }\right) \\ \nonumber
& +m \sum_i c_i^{\dagger} \sigma_z c_i
\end{align}
The creation operators for each site 
$c_i^{\dagger}=\left(c_{A i}^{\dagger} \quad c_{B i}^{\dagger}\right)^T$ 
are now vectors because each site has two orbitals 
and the Hilbert space again splits as 
$\mathcal{H}_{A}\oplus \mathcal{H}_{B}$. 
The model is shown in Figure \ref{fig:square1nn}. The two 
interaction matrices are explicitly: $t_{x1} = -
t\frac{\left(\sigma_z-i \sigma_x\right)}{2}$ and $t_{y1} = -t\frac{\left(\sigma_z-i \sigma_y\right)}{2}$. The model again breaks 
time-reversal symmetry without external magnetic field similar to the Haldane model Eq.~\eqref{equation:HaldaneH}. In momentum space the Hamiltonian becomes:

\begin{equation}
\begin{aligned}
H(k) &=  {h}(k) \cdot {\sigma}, \text{ with}\quad h_1(k) = t_1 \sin k \cdot a_1,\\
h_2(k) &= t_1 \sin k \cdot b_1,\quad h_3(k) = m - t_1 \cos k \cdot a_1 - t_1 \cos k \cdot b_1.
\end{aligned}
\end{equation}
Here, $a_1$($b_1$) is the displacement vector to the nearest neighbor in the positive $x$($y$) direction.
\begin{figure}[h]
    \centering
    \includegraphics[scale=2]{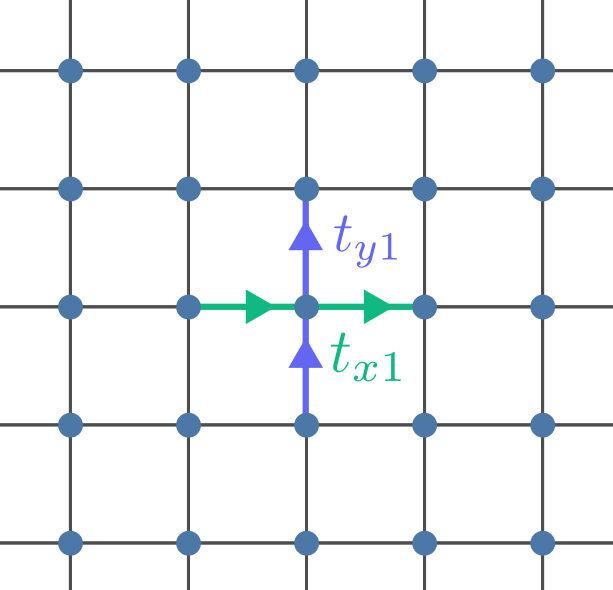}
    \caption{A model for a topological insulator on a square lattice Eq.~\eqref{equation:squareHshort}. Each atom has two orbitals so the interaction terms $t_{x1}$ and $t_{y1}$ are $2\times2$ Hermitian matrices.}
    \label{fig:square1nn}
\end{figure}
The ground state of the model with $\frac{m}{t_1}=-1$ has a unit Chern number \cite{zoology2012}. The Chern number can be calculated through any of the methods in \S\ref{subsection:calcdegree}, and the phase diagram is presented in Figure ~\ref{fig:squarephase}.
\begin{figure}[h]
    \centering
    \includegraphics[scale=0.6]{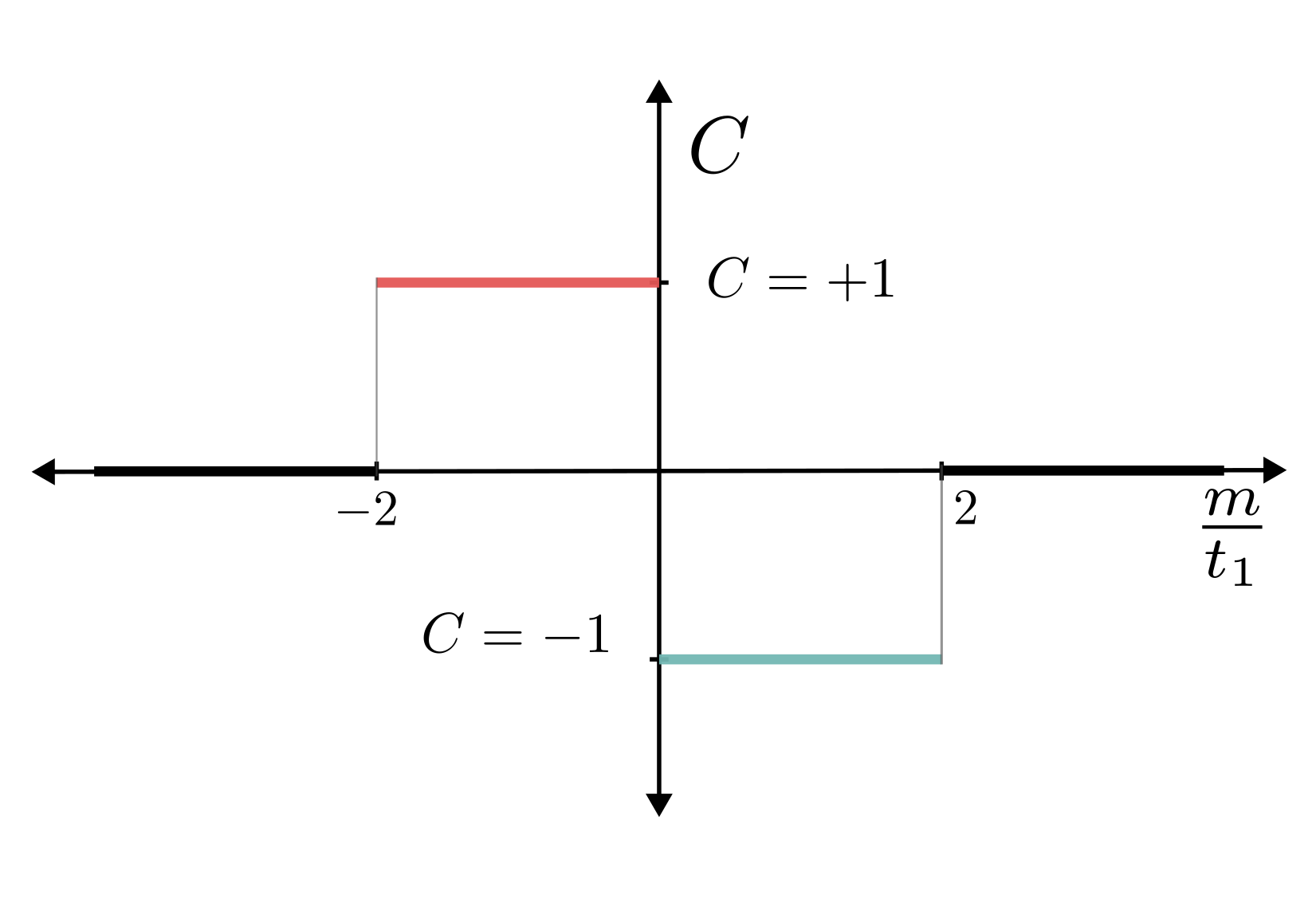}
    \caption{One-dimensional phase diagram of the model for square lattice Eq.~\eqref{equation:squareHshort}. The model is gapless for $\frac{m}{t_1} \in \{0,\pm2\}$, otherwise it is gapped.}
    \label{fig:squarephase}
\end{figure}

\subsubsection{Criteria for choosing distant neighbors in a square lattice}
\label{par:squarecriteria}
The higher order neighborhoods are best handled through identifying the lattice with the Gaussian integers $\Z+i\Z\subset \C$, see \S\ref{subsection:gaussian}. For example, the number of nearest neighbors at a certain distance can be translated into the number of preimages of the norm function of that distance; see corollary~\ref{corollary:square4nns}. We can then extend the range of hopping to any of the distances $N \in \{2,3,4,6,7,8,9,11,\cdots \}$ and obtain a model with Chern numbers $N^2$. The general form of Hamiltonian with Chern number $N^2$ is:
\begin{equation}  
\label{equation:squareN2} 
\begin{aligned}
h_{1}^{[C=N^2]}(k) &=\sin ( k \cdot N a_1) = \sin (N k \cdot a_1), \quad h_{2}^{[C=N^2]}(k) = \sin (k \cdot N b_1) = \sin (N k \cdot b_1), \\
h_{3}^{[C=N^2]}(k) &= \frac{m}{t_N} - \cos ( k \cdot N a_1) - \cos ( k \cdot N b_1)= \frac{m}{t_N} - \cos (N k \cdot a_1) - \cos (N k \cdot b_1).
\end{aligned}
\end{equation}

 Where we emphasized that the distant neighbors effectively implement the mapping $f_{d_1,d_2}:T^2\to T^2, (\theta_1,\theta_2)\mapsto (d_1\theta_1, d_2\theta_2)$ whose degree is $d_1d_2$. Here, $d_1=d_2=N$ and the resulting Hamiltonian has a Chern number $C=N^2$ \S\ref{subsec:momentumspace}. The structure of the nearest neighbors in a square lattice is shown in Figure ~\ref{fig:squarenns}. For distances $d\in \{1,2,3,4\}$, there is no prime $p\equiv 1\pmod{4}$ as a factor of $d$. Consequently, for these distances we recover the same structure of neighbors as in the original model Eq.~\eqref{equation:squareHshort}. Thus they represent valid choices for $N$ in Eq.~\eqref{equation:squareN2}.

\begin{figure}[h]
    \centering
    \includegraphics[scale=2]{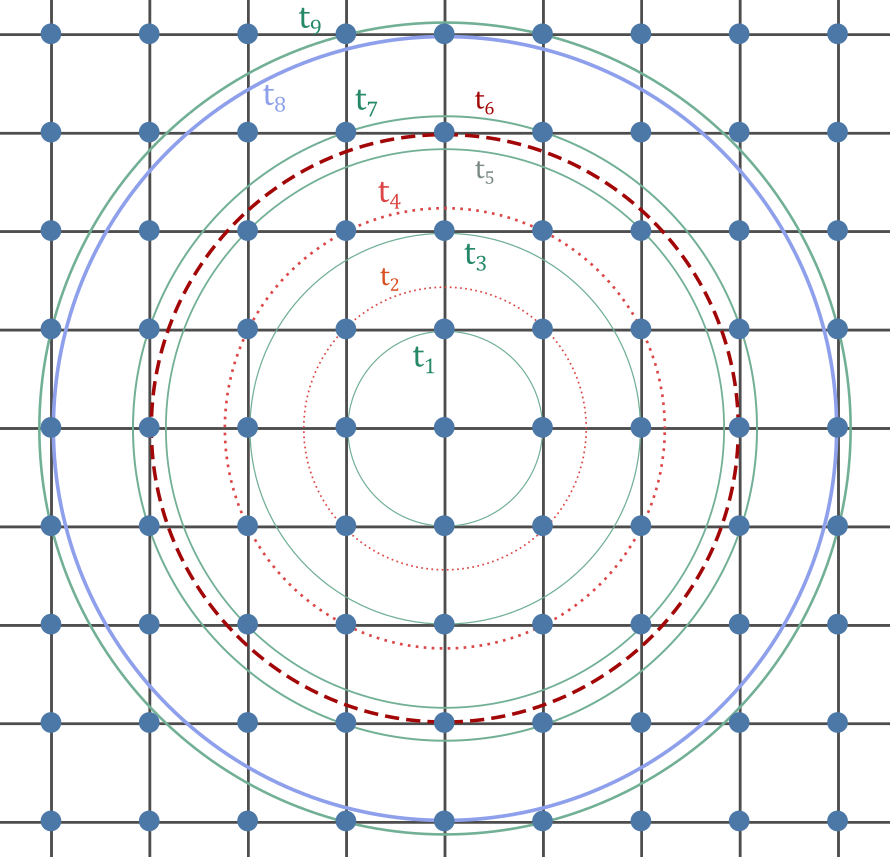}
    \caption{The structure of the nearest-neighbor interactions in a square lattice. For certain ranges ($d\in \{1,2,3,4,6,\cdots\}$) the structure of the first-nearest neighbors appears again. For example, $t_1,t_3,t_6\text{ and } t_8$ have the same structure as $t_1$. They can be used to construct models with Chern numbers $C = 1,4,9,16$, respectively.}
    \label{fig:squarenns}
\end{figure}

\subsubsection{Triangular Lattice}
The next fundamental example is the case of a triangular lattice. The basic Hamiltonian realizing a Chern insulator on a triangular lattice, similar to Haldane's choice \cite{Sarma2012}, is shown in Fig.~\ref{fig:triangularmodel}. It has real nearest-neighbor interaction $t_1$ and complex next-nearest-neighbor interaction $t_2e^{\pm i \phi}$ along with an on-site potential $m$.
\begin{equation}
H=-t_1 \sum_{\langle i, j\rangle} c_{i,A}^{\dagger} c_{j,B}-t_2 \sum_{\langle\langle i, j\rangle\rangle} e^{i \phi_{i j}}\left(c_{i,A}^{\dagger} c_{j,A}+c_{i,B}^{\dagger} c_{j,B}\right) + m \sum_i \left( c_{i,A}^{\dagger}c_{i,A}-c_{i,B}^{\dagger}c_{i,B} \right)
\end{equation}
In this model each atom has two orbitals $A$ and $B$. The second nearest-neighbor hopping $t_2$ has a phase that depends on the the two orbitals as shown in Fig.~\ref{fig:triangularmodel}. 

\begin{figure}[h]
    \centering
    \includegraphics[scale=1.9]{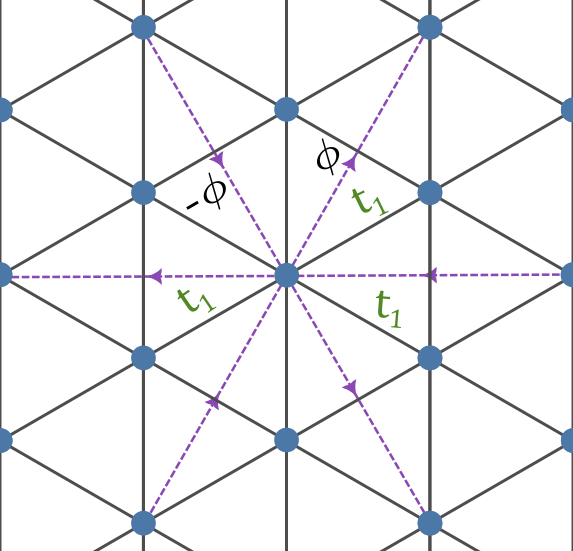}
    \caption{Simple model for a Chern insulator on a triangular lattice. The phase $\phi$ is positive if it is aligned with the arrows and it is from $A$ to $B$ orbitals. Flipping the arrow or the order of the orbitals flips the sign of $\phi$. }
    \label{fig:triangularmodel}
\end{figure}

In momentum space the model reads:
\begin{equation}
\begin{aligned}
    H({k}) &= {h}({k}) \cdot {\sigma}, \\
    h_1({k}) &= -\frac{1}{2} t_1 \sum_{i=1}^{6} \cos\Bigl({k}\cdot{a}_i\Bigr), \quad
    h_2({k}) = -\frac{1}{2} t_1 \sum_{i=1}^{6} (-1)^i \sin\Bigl({k}\cdot{a}_i\Bigr), \\
    h_3({k}) &= m - t_2 \sum_{i=1}^{6} (-1)^i \sin\Bigl({k}\cdot{b}_i\Bigr).
\end{aligned}
\end{equation}\label{equation:triangularH}

Here, $a_i$ and $b_i$ denote the first and second-nearest neighbors vectors. The model in momentum space is the same as the Haldane model Eq.\eqref{equation:HaldaneH}. The only difference is that the Brillouin zone is bigger than the Honeycomb lattice by a factor of $3$. This means that the Chern number of the ground state for $t_1 =1$ is three times that of the same parameters in the Haldane model; see Fig.~\ref{fig:haldanephase}. The phase diagram is then shown in Fig.~\ref{fig:triangular3-3}.

\begin{figure}[h]
    \centering
    \includegraphics[scale=0.6]{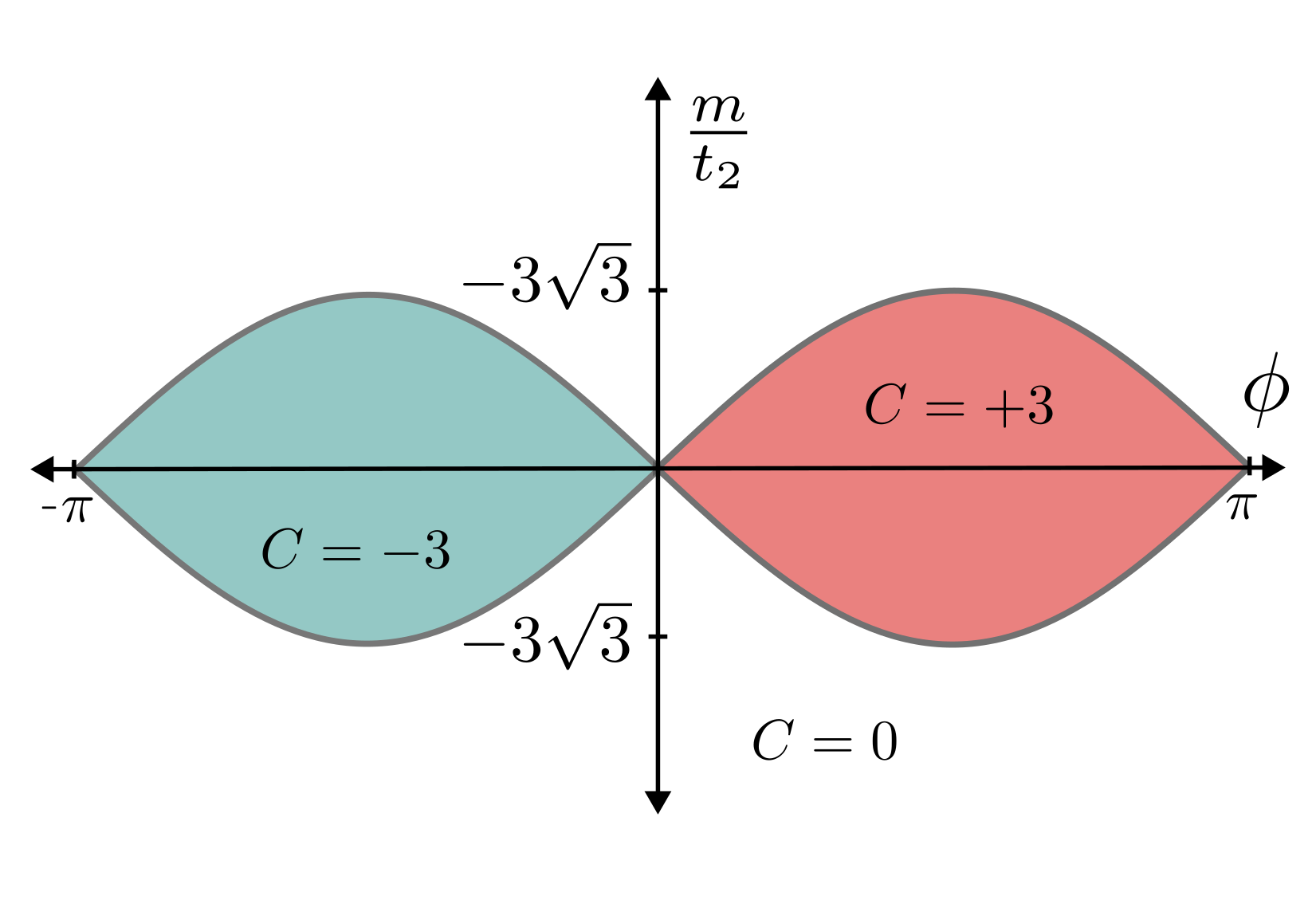}
    \caption{Phase diagram of a modified Haldane model on a triangular lattice  with \( t_1 = 1\) in Eq.~\eqref{equation:triangularH}. The Chern number takes higher values \{-3,3\} as it is 3 times that of the Honeycomb lattice; compare with Fig.~\ref{fig:haldanephase}.}
    \label{fig:triangular3-3}
\end{figure}
Given this base Hamiltonian, we can construct a new Hamiltonian with a distant hopping that realizes a higher Chern number as follows:
\begin{equation}
\begin{aligned}
    h_1^{[C=N^2]}({k}) &= -\frac{1}{2} t_1 \sum_{i=1}^{6} \cos\Bigl({Nk}\cdot{a}_i\Bigr), \quad
    h_2^{[C=N^2]}({k}) = -\frac{1}{2} t_1 \sum_{i=1}^{6} (-1)^i \sin\Bigl({Nk}\cdot{a}_i\Bigr), \\
    h_3^{[C=N^2]}({k}) &= m - t_2 \sum_{i=1}^{6} (-1)^i \sin\Bigl({Nk}\cdot{b}_i\Bigr).
\end{aligned}
\end{equation}\label{equation:triangularN2}

As before, the distant neighbors effectively implement the mapping $f_{d_1,d_2}:T^2\to T^2, (\theta_1,\theta_2)\mapsto (d_1\theta_1, d_2\theta_2)$ whose degree is $d_1d_2$. Here, $d_1=d_2=N$ and the resulting Hamiltonian has a Chern number $C=3N^2$ since we started with a Hamiltonian with degree 3 \S\ref{subsec:momentumspace}. 

The wall crossings are from $0$ to $\pm 3$ given by circles, but now the origin is crossed $3$ times, which leaves three pre-Dirac points after crossing to either of the two phases with a non-zero Chern number.


\subsubsection{A second basic model for triangular lattices}  
We also note that in triangular lattices, we do not need to preserve the whole structure of the original Hamiltonian. We can also construct models that only preserve the structure of the $t_1$ interaction terms while keeping the $t_2$ terms the same. These models do not arise as an effective composition with a higher degree map. They introduce an intrinsically new map that has a higher degree. The new Chern number will be $N$ times the original one. We will discuss the proof of the degree of such maps in the very similar model of honeycomb lattices; see \eqref{equation:haldaneN}. We also note that such models should be of practical value as the complex interactions $t_1$ are harder to implement experimentally. The new Hamiltonian will be:
\begin{equation}
\begin{aligned}
    h_1^{[C=N]}({k}) &= -\frac{1}{2} t_1 \sum_{i=1}^{6} \cos\Bigl({Nk}\cdot{a}_i\Bigr), \quad
    h_2^{[C=N]}({k}) = -\frac{1}{2} t_1 \sum_{i=1}^{6} (-1)^i \sin\Bigl({Nk}\cdot{a}_i\Bigr), \\
    h_3^{[C=N]}({k}) &= m - t_2 \sum_{i=1}^{6} (-1)^i \sin\Bigl({k}\cdot{b}_i\Bigr).
\end{aligned}
\end{equation}\label{equation:triangularN}

\subsubsection{Criteria for choosing distant neighbors in a triangular lattice} 
\label{par:trianglecriteria}
Putting any integer $N$ in Eq.~\eqref{equation:triangularN2} will give a valid Hamiltonian with a Chern number $3N^2$. However, this Hamiltonian will not correspond to the tight-binding Hamiltonian at that new distance as the number and direction of neighbors change with distance in a triangular lattice as presented in Fig.~\ref{fig:triangularnns}. At certain integer multiples of the distances of the original model though, as discussed in \S\ref{subsection:eisenstein}, the distant neighbors have the same number and direction of the original one. These integers are of the form \( N = 3^{b_0} p_1^{b_1} p_2^{b_2} \dots p_c^{b_c} \), where \( N \) is the product of rational inert primes \( p_i \equiv 2 \pmod{3} \). The first few integers are then $N\in \{1,2,3,4,5,6,9,10,\cdots\}$. If we multiply the range of interactions $t_1$ and $t_2$ by these integers we will get a new model that does not break the symmetry of the lattice and that has $C=N^2$ times the original one. If we only multiply the range of $t_1$ terms we get a Hamiltonian that has $C=N$ times the original one.
 
\begin{figure}[h]
    \centering
    \includegraphics[scale=1.3]{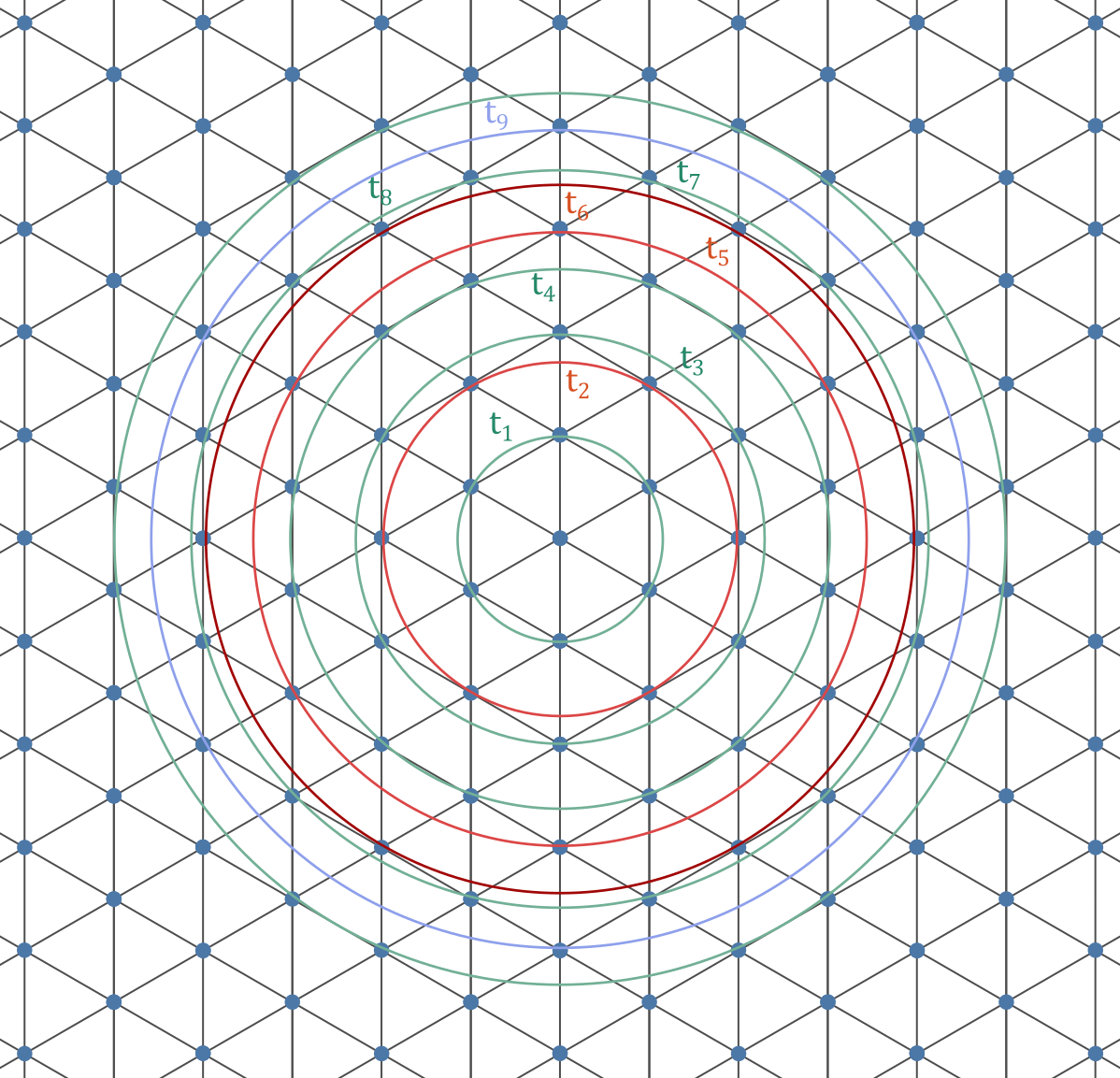}
    \caption{The structure of distant neighbors in a triangular lattice. The number and orientation of nearest neighbors evidently vary with distance. One way to obtain commensurate models with higher degree is to consider neighbors at distances $2$ times that of nearest neighbors. This corresponds to $t_1 \mapsto t_3$ and $t_2 \mapsto t_6$. This implements Eq.~\eqref{equation:triangularN2} with $N=2$ and will result in a model with Chern number $C=(2^2) \cdot 3=12$ since the original model with $N=1$ had $C=3$. Another construction is to leave $t_2$ the same and change the interaction $t_1 \mapsto t_3 \ , \  t_5 \text{ or } t_9 $. This will implement Eq.~\eqref{equation:triangularN} with $N=2,3$ and $4$ respectively, and will result in models with Chern numbers $6,9$ and $12$.  }
    \label{fig:triangularnns}
\end{figure}

\subsubsection{Honeycomb Lattice}
 The honeycomb lattice received a lot of attention as it was the primary example of a tight-binding topological insulator without external magnetic field \cite{haldane88}. Many physics papers studied instances of increasing the range of hopping to obtain higher Chern numbers \cite{Bena2011,Sticlet2013,Mondal2022}. As an example, for the third nearest neighbor we have the following Hamiltonian: 
\begin{align} \label{equation:h3(k)}
      h_{1}^{3^{rd}NN}({k}) &= \sum_{i=1}^3 \cos \left(k \cdot a_i\right)+t_3\sum_{i=1}^3 \cos \left(k \cdot c_i\right) \nonumber\\
    h_{2}^{3^{rd}NN}({k}) &=\sum_{i=1}^3 \sin \left(k \cdot a_i\right)+t_3\sum_{i=1}^3 \sin \left(k \cdot c_i\right) \nonumber\\
    h_{3}^{3^{rd}NN}({k}) &= m - 2 t_2 \sin(\phi) \sum_{i=1}^3 \sin \left(k \cdot b_i\right)
\end{align}
Here, $a_i$ and $b_i$ denote the first and second-nearest neighbors vectors respectively while the new vectors $c_i$ are the third-nearest neighbors. The phase diagram is shown in Figure ~\ref{fig:haldane3nn}. The Hamiltonian remains $2\times2$ and the $h_3$ component remained the same because the third nearest neighbors are of the opposite species to the central atom. The third-nearest neighbor model hosts a phase with a Chern number $C=2$ \cite{Sticlet2013}. 
\begin{figure}[h]
    \centering
    \includegraphics[scale=0.6]{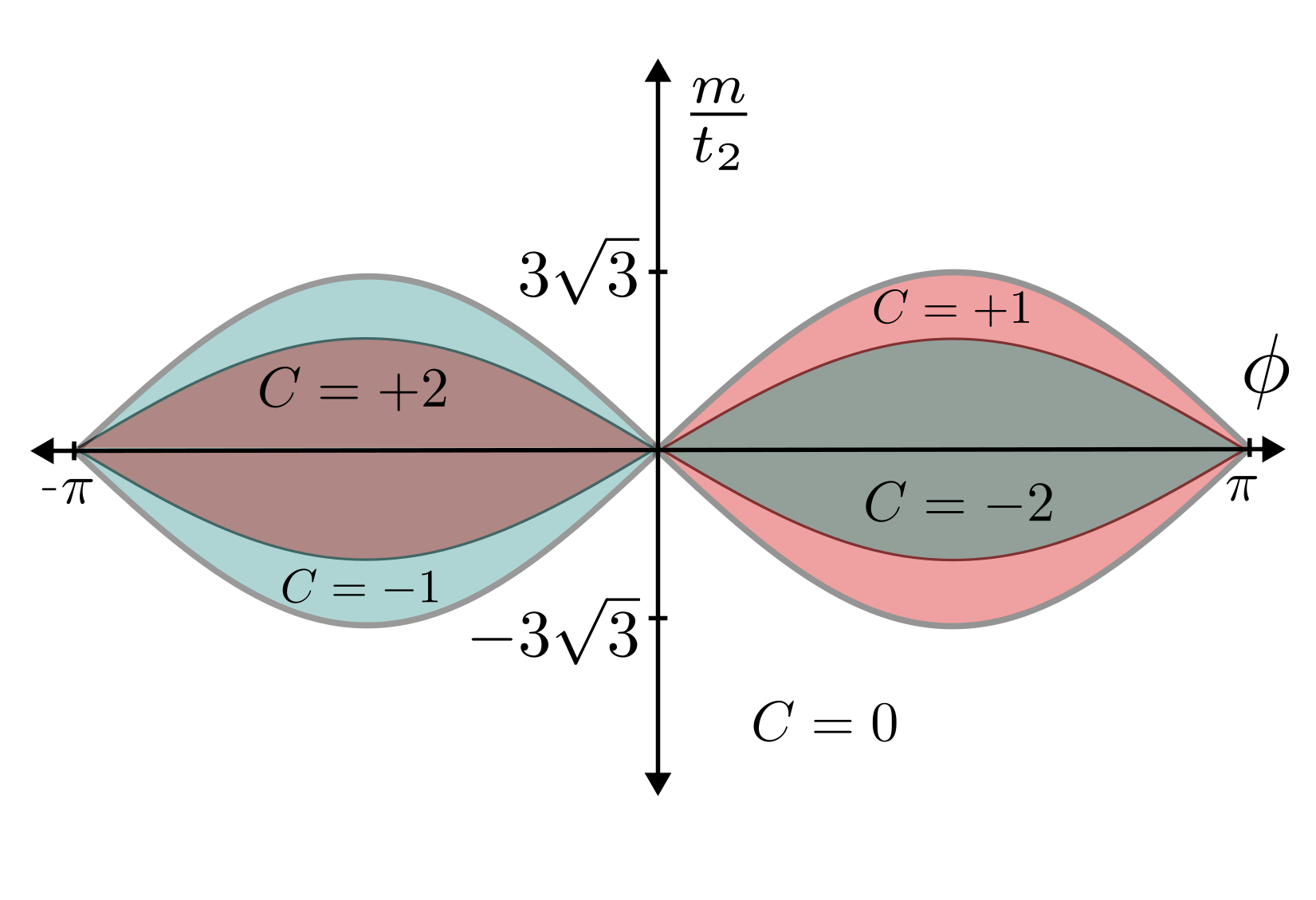}
    \caption{Phase diagram of the Haldane model with third-nearest neighbors \(  t_2 =0.5, \ t_3=0.35 \) in Eq.~\eqref{equation:h3(k)}. The Chern number takes higher values {-2,2} for a specific region of parameters $\phi$ and $m$.}
    \label{fig:haldane3nn}
\end{figure}
 From the correspondence between the Chern number and the mapping degree, the image of the mapping $H(k): T^2\to \R^3$ has to have degree $2$ for the corresponding points in the phase diagram Figure~\ref{fig:haldane3nn}. For example, the image of the map for $m=0$, $\phi=\frac{\pi}{2}$ and $t_3=0.35$ is shown in Figure~\ref{fig:rosec-2}. The map's restriction to the plane $h_3(k)=0$ is isotopic to the rose curve shown in Figure~\ref{fig:D} with $d=1$, $d^{\prime}=-2$ and $t=\frac{1}{4}$.
 Here the isotopy preserves the intersection points and fixes these at $0$.
\begin{figure}[h]
    \centering
    \includegraphics[scale=0.6]{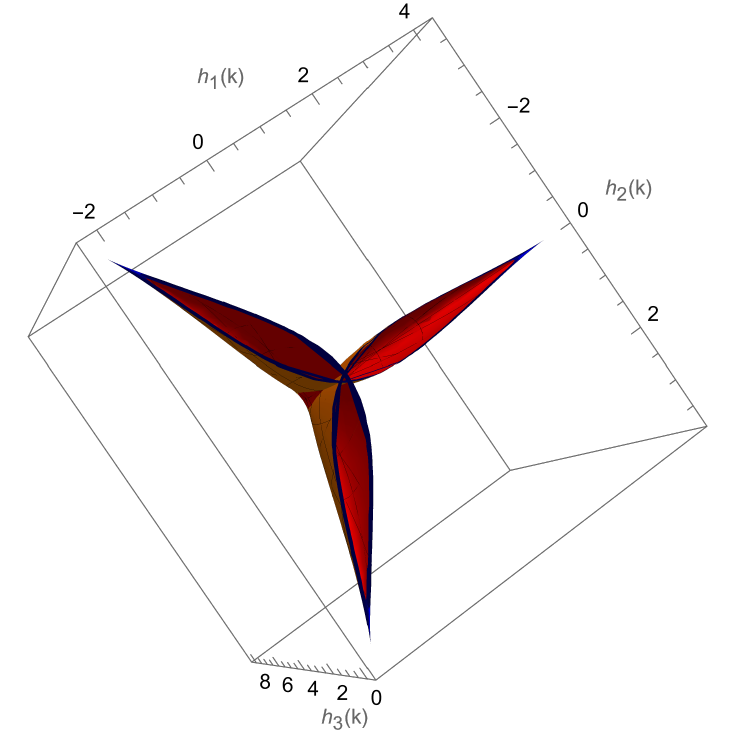}
    \caption{The image of the map $H(k):T^2\to\R^3$ for $m=0$, $\phi=\frac{\pi}{2}$, $t_2=0.5$ and $t_3=0.35$. The map has a degree $2$ as any ray starting from origin will intersect the surface twice. The highlighted plane $h_3(k)=0$ is isotopic to the rose curve in Figure~\ref{fig:D} also fixing the origin as the sole crossing point in the isotopy}
    \label{fig:rosec-2}
\end{figure}

We remark that calculating the Chern number using the connection involves numerically evaluating a complicated integral. Analytically, we can compute the map degree from the ray method. However, this involves solving increasingly higher-degree polynomials which becomes impractical after the the fourth nearest-neighbor \cite{Sticlet2013}. Moreover, the construction of these high Chern number models depend on tuning the interaction parameters to obtain a slice with high Chern numbers. This becomes increasingly harder as the dimension of the parameter space increases with each new interaction. However, applying the method of commensurate sublattices we obtain Hamiltonians with higher Chern number for Graphene in a simpler way. The resulting Hamiltonian when considering neighbors that are $N$ times further than first-nearest neighbors is then: 

\begin{align} \label{equation:haldaneN2}
      h_{1}^{[C=N^2]}({k}) &= \sum_{i=1}^3 \cos \left(k \cdot Na_i\right) = \sum_{i=1}^3 \cos \left(Nk \cdot a_i\right) \nonumber\\
    h_{2}^{[C=N^2]}({k}) &=\sum_{i=1}^3 \sin \left(k \cdot Na_i\right)=\sum_{i=1}^3 \sin \left(Nk \cdot a_i\right)  \nonumber\\
    h_{3}^{[C=N^2]}({k}) &= m - 2 t_2 \sin(\phi) \sum_{i=1}^3 \sin \left(k \cdot Nb_i\right) = m - 2 t_2 \sin(\phi) \sum_{i=1}^3 \sin \left(Nk \cdot b_i\right)   \nonumber\\
\end{align}

 Where we emphasized that the distant neighbors effectively implement the mapping $f_{d_1,d_2}:T^2\to T^2, (\theta_1,\theta_2)\mapsto (d_1\theta_1, d_2\theta_2)$ whose degree is $d_1d_2$. Here, $d_1=d_2=N$ and the resulting Hamiltonian has a Chern number $C=N^2$. Concretely, in the Haldane model, we can take the interactions at a distance $N=4$ multiples of the original. Assuming that the side of the hexagon has length $1$, then we consider the blue atoms at distance $4$ and the red atoms at distance $4\sqrt{3}$. It can be checked explicitly that there are only three blue atoms and six red atoms at these distances. Further, their vectors are just an integer multiple of the original Haldane model with distances $1,\sqrt{3}$; see Figure ~\ref{fig:Haldanefarnns}. 
 
 \subsubsection{Criteria for choosing distant neighbors in honeycomb lattice}
 \label{par:honeycriteria}
 The honeycomb lattice can be viewed as a subset of a triangular lattice with two species (the red and blue circles) while missing the third species (the hollow purple circles) as depicted in Figure ~\ref{fig:Haldanefarnns}. The 
 criteria for choosing $N$ (the integer multiple of the distant neighbors) again follows the recipe for a triangular lattice \S\ref{subsection:eisenstein} but, with 
 the additional constraint that we want to avoid fictitious lattice points at that distance. For example, we cannot extend the $A-B$ $t_1$ interaction by either $N=2,5$ or more generally $N>0$ and $ N  \equiv 2 \pmod{3}$ as at these distances there are no atoms in the same direction as the first nearest neighbors. However, for the negative values $N=-2,-5$, we have atoms at these distances and we can choose $N$ in these numbers. Further, since the Chern number is odd in $N$ and $\phi$, flipping both of them produce models with a positive Chern number.

\begin{figure}[h]
    \centering
    \includegraphics[scale=1.3]{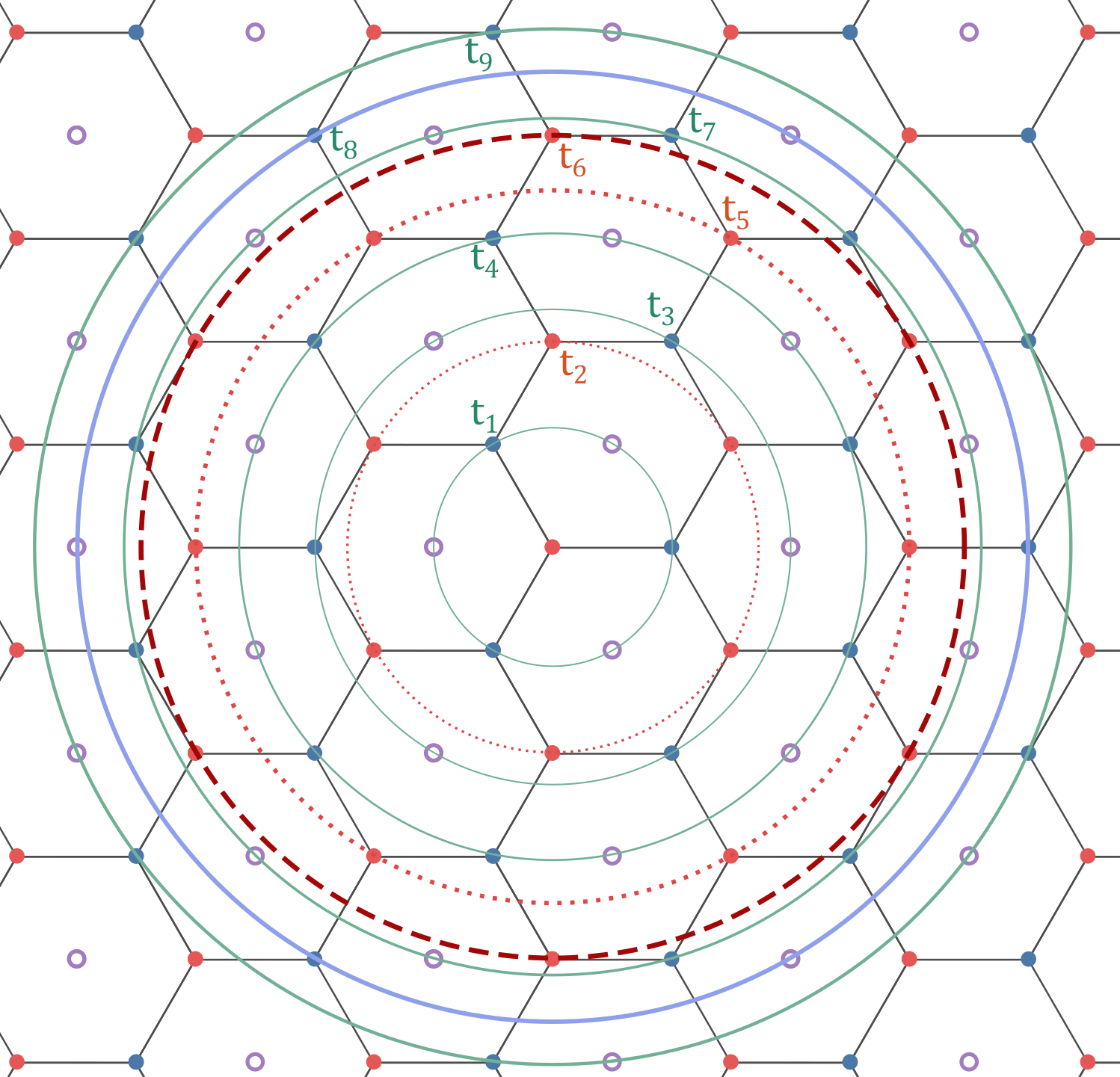}
    \caption{The structure of distant neighbors in a Honeycomb lattice. The number and orientation of nearest neighbors vary with distance. The original Haldane model used $t_1,\ t_2$ interactions. The hollow purple circles are fictitious points drawn to show that the honeycomb lattice is a subset of a triangular lattice. One way to obtain commensurate interactions with higher degree is to consider neighbors at distances $2$ or $4$ times that of nearest neighbors ($t_1$ terms); these correspond to $t_3$ or $t_8$. This will result in Chern number $2,\ 4$, respectively. These are instances of Eq.~\eqref{equation:haldaneN}.}
    \label{fig:Haldanefarnns}
\end{figure}

\subsubsection{A second basic model for Honeycomb lattices} 
In the previous subsection, we implemented the construction of higher-Chern-number models by looking for distant neighbors that preserve the structure of the original Hamiltonian. In honeycomb lattices, we can also construct models that only preserve the $A-B$ interaction terms. Such models should be of practical value as the complex interactions $A-A$ and $B-B$ are harder to implement experimentally. The new Hamiltonian will be:
\begin{align} \label{equation:haldaneN}
      h_{1}^{[C=N]}({k}) &= \sum_{i=1}^3 \cos \left(Nk \cdot a_i\right), \nonumber\quad    h_{2}^{[C=N]}({k}) = \sum_{i=1}^3 \sin \left(Nk \cdot a_i\right)  \nonumber\\
    h_{3}^{[C=N]}({k}) &= m - 2 t_2 \sin(\phi) \sum_{i=1}^3 \sin \left(k \cdot b_i\right)
\end{align}
We can compute the Chern number of this model using the ray method; see \S\ref{subsection:calcdegree}. We draw a ray from origin in the z direction and count how many times it intersects the surface generated by the Hamiltonian $H(k)$ in $\R^3$ as the momentum traverses the $T^2$ Brillouin zone. We can also extend the ray into a straight line and divide by 2. \begin{equation}
C_1 = \sum_{p \in h^{-1}(q)} \operatorname{sgn}(\phi_p)
= \frac{1}{2} \sum_{\substack{k \in \text{BZ} \\ h_1(k)=h_2(k)=0}} 
\operatorname{sgn}\!\Bigl[(\partial_{k_x}h(k) \times \partial_{k_y}h(k))\cdot \hat{z}\Bigr]
\,\operatorname{sgn}\bigl[h_3(k)\bigr].
\end{equation}

The points satisfying $h_1(k)=h_2(k)=0$ are the pre-Dirac points \S\ref{subsection:calcdegree}. In the original model with $N=1$, these pre-Dirac points form on a honeycomb lattice with translation vectors: $\vec{g}_1 = \left( \frac{2\pi}{\sqrt{3}}, \, \frac{2\pi}{3} \right), \quad
\vec{g}_2 = \left( -\frac{2\pi}{\sqrt{3}}, \, \frac{2\pi}{3} \right)$. 
Restricting to the first Brillouin Zone, we find only two terms $K=(\frac{4\pi}{3\sqrt{3}},0)$ and $K'=(\frac{-4\pi}{3\sqrt{3}},0)$ shown as blue and red circles respectively in Figure ~\ref{fig:HaldanefarBZ}. Since the new Hamiltonian terms ($h_1(k) \text{ and }h_2(k) $) are formed by composing the old terms by the map $k \mapsto Nk$, the new pre-Dirac points consist of the old pre-Dirac points but in an extended BZ. Since the new BZ is $N^2$ larger, it has $N^2$ more pre-Dirac points. For a model with a general $N$, the new $N^2$ pre-Dirac points are explicitly:\begin{equation}
\scalebox{0.95} {$
\{k  \mid  h_{1}^{[C=N]}({k})= h_{2}^{[C=N]}({k})=0\}=\left\{
\frac{K + m\,\vec{g}_1 + n\,\vec{g}_2}{N},\; \frac{K' + m\,\vec{g}_1 + n\,\vec{g}_2}{N}
\; \Big| \; 0 \leq m,\, n \leq N-1
\right\}.
$}
\end{equation}

These vectors differ by an unimportant shift from the ones depicted in Figure ~\ref{fig:HaldanefarBZ}. Note that the sign of the first term $sgn\left[(\partial_{k_x}{h}(k)\times \partial_{k_y} {h}(k))\cdot \hat{z}\right]$ is inherited from the original BZ (or inverted if $N$ is negative). The first Chern number for the new Hamiltonian is then:
\begin{equation}
C_1 = \frac{1}{2} \sum_{\substack{0 \leq m \leq N-1 \\ 0 \leq n \leq N-1}}
\left[
\operatorname{sgn}\!\Bigl( h_3\Bigl(\frac{K + m\,\vec{g}_1 + n\,\vec{g}_2}{N} \Bigr)\Bigr)
- \operatorname{sgn}\!\Bigl( h_3\Bigl(\frac{K' + m\,\vec{g}_1 + n\,\vec{g}_2}{N} \Bigr)\Bigr)
\right].
\end{equation}
We can change the BZ for the second sum by $(\vec{g_1},\vec{g_2}) \mapsto (-\vec{g_1},-\vec{g_2})$. This is just a reordering of the old vectors. The sum then becomes \begin{equation}
C_1 = \frac{1}{2} \sum_{\substack{0 \leq m \leq N-1 \\ 0 \leq n \leq N-1}}
\left[
\operatorname{sgn}\!\Bigl( h_3\Bigl(\frac{K + m\,\vec{g}_1 + n\,\vec{g}_2}{N}\Bigr)\Bigr)
- \operatorname{sgn}\!\Bigl( h_3\Bigl(\frac{K' - m\,\vec{g}_1 - n\,\vec{g}_2}{N}\Bigr)\Bigr)
\right].
\end{equation}
We specialize to the case $m=0$ (this is $m$ in ~\eqref{equation:haldaneN} not to be confused with the dummy integer here), $\phi=\frac{\pi}{2}$ and $t_2=\frac{1}{2}$. Thus $h_3 (k) =  \sum_{i=1}^3 \sin \left(k \cdot b_i\right)$ is an odd function of $k$. This, together with the fact that $K'=-K$, simplifies the expression to: $C_1 = \sum_{0\leq m.n \leq N-1}sgn[h_3(\frac{K+m\vec{g_1}+n\vec{g_2}}{N})]$. Using the explicit form of $h_3(k)$, we arrive at:

\begin{equation}
\begin{aligned}
C_1 = \sum_{\substack{0 \leq m \leq N-1 \\ 0 \leq n \leq N-1}}
\operatorname{sgn}\Biggl[ &
\left(\cos\!\left(\frac{(m+n)\pi}{N}\right)
- \cos\!\left(\frac{(2-3m+3n)\pi}{3N}\right)\right) \\
&\times \sin\!\left(\frac{(2-3m+3n)\pi}{3N}\right)
\Biggr].
\end{aligned}
\end{equation}
We note that changing $m \leftrightarrow n$ changes the sign of the expression for $m\neq n$ and $m,n > 0$. Furthermore, the expression is negative if $m=n\neq0$ and positive when $mn=0$. These observations show that $C_1 = N$ for any $N$. It is worth noting that the Haldane model for 3rd-nearest neighbor \eqref{equation:h3(k)} is an example of this family of Hamiltonians since the 3rd-nearest neighbors have the same structure as the 1st-nearest neighbors but with $N=-2$. see Figure ~\ref{fig:haldane3nn}.
\begin{figure}[h]
    \centering
    \includegraphics[scale=1.3]{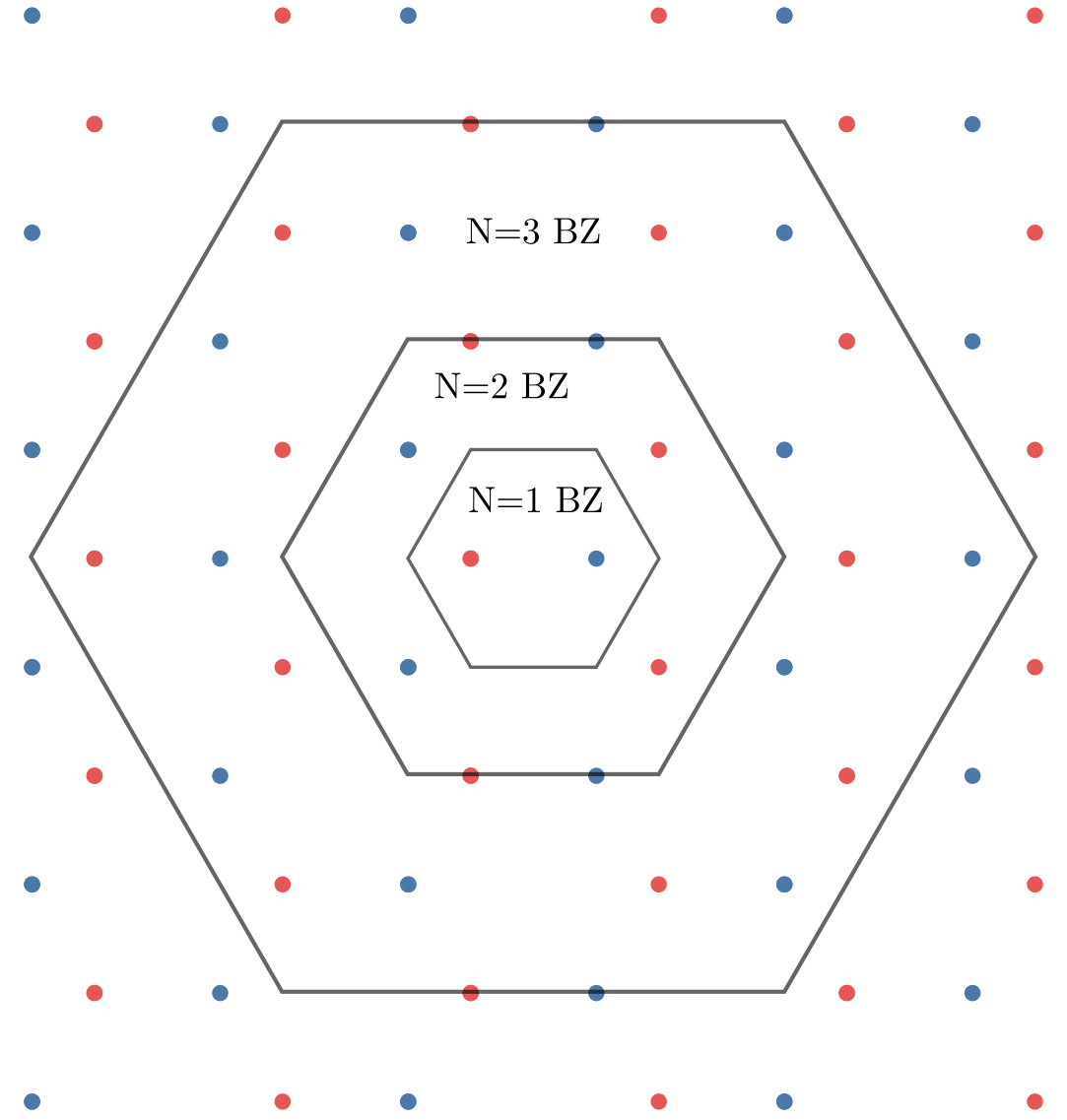}
    \caption{The structure of the pre-Dirac points for different choices of $N$. The $K$-like and $K'$-like points (blue and red respectively) form a honeycomb lattice in momentum space. It is evident geometrically that there are $N^2$ many pre-Dirac points compared to the $N=1$ case (the points at the boundary are identified because of periodic boundary conditions).}
    \label{fig:HaldanefarBZ}
\end{figure}

\subsubsection{Kagome lattice}
\label{par:kagome}
The same methods can be applied to the case of Kagome lattices. As a concrete example, consider the Hamiltonian describing an interaction in a 3-species Kagome lattice, see Figure ~\ref{fig:kagomemodel}: 
\begin{equation}\label{equation:KagomeH}\begin{aligned}
 H =&-t_1\sum_{\langle i, j\rangle} (c_{i,A}^{\dagger} c_{j,B}+c_{i,A}^{\dagger} c_{j,C}+c_{i,B}^{\dagger} c_{j,C})+iu_1 \sum_{\langle i, j\rangle}  (c_{i,A}^{\dagger} c_{j,B}-c_{i,A}^{\dagger} c_{j,C}+c_{i,B}^{\dagger} c_{j,C})\\ &+ h.c.
\end{aligned}\end{equation} 
As before, $\langle i, j\rangle$ denotes summation over first-nearest neighbors. This Hamiltonian is adapted from the Hamiltonian: 
\begin{equation}
 H=-t_1 \sum_{\langle i j\rangle \sigma} c_{i \sigma}^{\dagger} c_{j \sigma}+i u_1 \sum_{\langle i j\rangle \alpha \beta}\left(\mathbf{E}_{i j} \times \mathbf{R}_{i j}\right) \cdot \boldsymbol{\sigma}_{\alpha \beta} c_{i \alpha}^{\dagger} c_{j \beta}   
\end{equation}
 Here, $\alpha$ and $\beta$ represent spin indices,  $\mathbf{R}_{i j}$ are the displacement vectors between sites $i$ and $j$, and $\mathbf{E}_{i j}$ is the electric field from neighbors along $\mathbf{R}_{i j}$. This is a model for Fe$_3$Sn$_2$, and the second term is the spin-orbit coupling from the Sn ion at the center of the hexagon \cite{kagomeprl2011}. We obtain Eq.~\eqref{equation:KagomeH} after restricting to the spin-up electrons. The model is presented in Figure ~\ref{fig:kagomemodel}. For another model that admits the same construction, see~\cite{KagomePRB2000}.

\begin{figure}[h]
    \centering
    \includegraphics[scale=1.5]{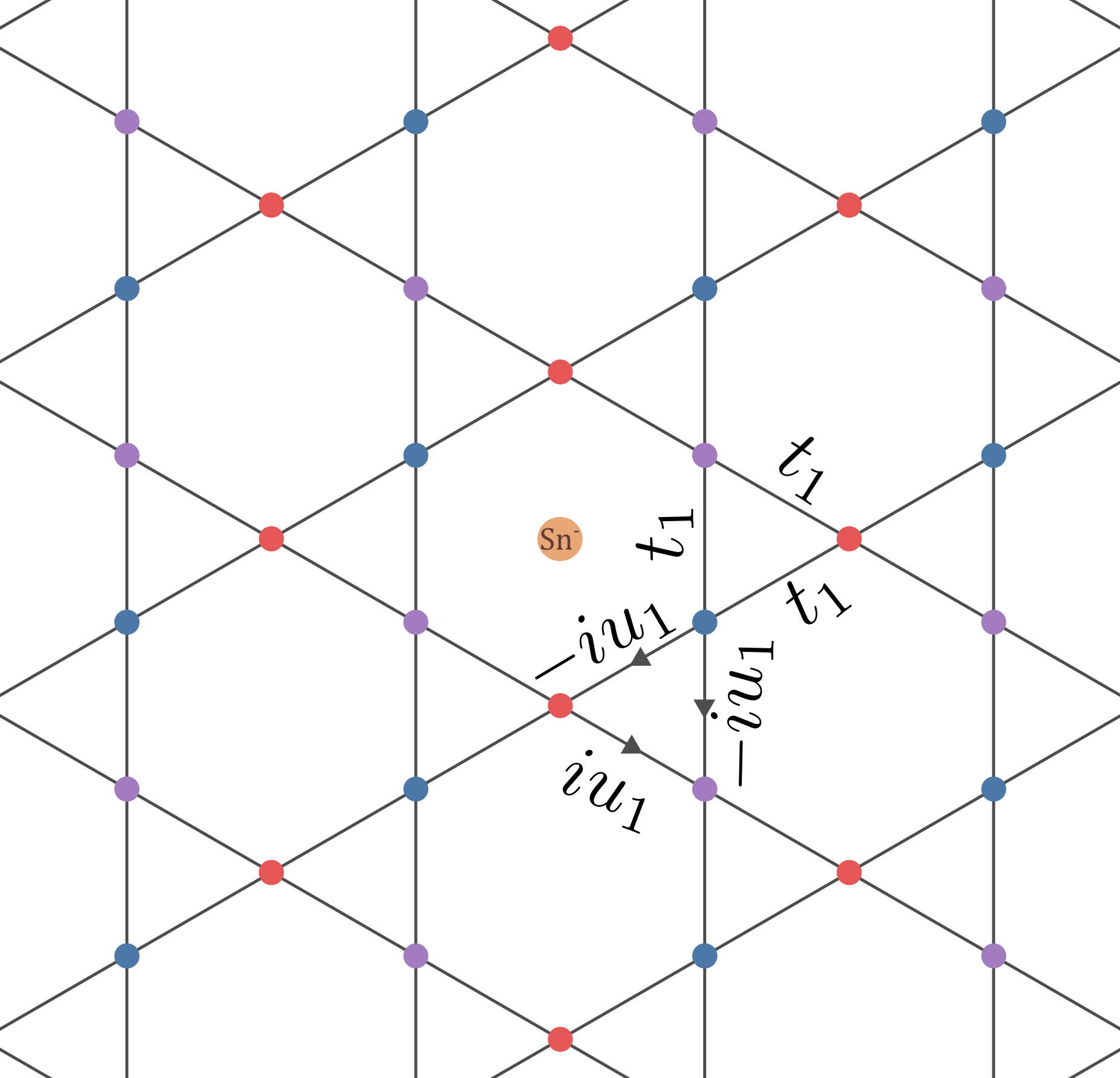}
    \caption{A model for a Chern insulator on a Kagome lattice. First-nearest interactions are both real and imaginary. The complex interactions arise from the spin-orbit coupling of the electric field of the Sn ion at the center. Blue, red and purple circles correspond to $A$, $B$ and $C$ atoms in the Hamiltonian Eq.~\eqref{equation:KagomeH}.}
    \label{fig:kagomemodel}
\end{figure}
As in the honeycomb case, we can do a Fourier transform to the momentum space similar to \eqref{equation:haldanefourier}. We define three creation/annihilation operators as follows:
\begin{equation}
\left(\begin{array}{c}
c_{k}  \\0\\0
\end{array}\right)  \coloneqq c_{k,A} \ket{0}  , \quad \left(\begin{array}{c}
0 \\c_{k} \\0
\end{array}\right)  \coloneqq c_{k,B} \ket{0} , \quad \left(\begin{array}{c}
0 \\0\\c_{k} 
\end{array}\right)  \coloneqq c_{k,C} \ket{0}   \end{equation}   
The $\ket{0}$ represents the vacuum state where no momentum states are occupied. In Figure ~\ref{fig:kagomemodel}, we denoted A, B and C atoms using blue, red and purple circles respectively. After the Fourier transform the resulting Hamiltonian is:
\begin{equation}
\begin{aligned}
H(k)= & -2 t_1\left(\begin{array}{ccc}
0 & \cos (k \cdot a_1) & \cos (k \cdot a_2) \\
\cos (k \cdot a_1) & 0 & \cos (k \cdot a_3) \\
\cos (k \cdot a_2) & \cos (k \cdot a_3) & 0
\end{array}\right) \\
& + 2i u_1\left(\begin{array}{ccc}
0 & \cos (k \cdot a_1) & -\cos (k \cdot a_2) \\
-\cos (k \cdot a_1) & 0 & \cos (k \cdot a_3) \\
\cos (k \cdot a_2) & -\cos (k \cdot a_3) & 0
\end{array}\right)
\end{aligned}
\end{equation}

The vectors $\{a_i\}$ represent the first-nearest neighbors displacements. The Hamiltonian is gapless when $u_1 =0,\pm \sqrt{3}$ otherwise it is gapped. The Chern number for the lowest band can be computed through the connection and it is found to be $-1$ when $t_1 = u_1=1$. A 1-d slice for the parameters is shown in Figure ~\ref{fig:kagome1-1}.

\begin{figure}[h]
    \centering
    \includegraphics[scale=0.6]{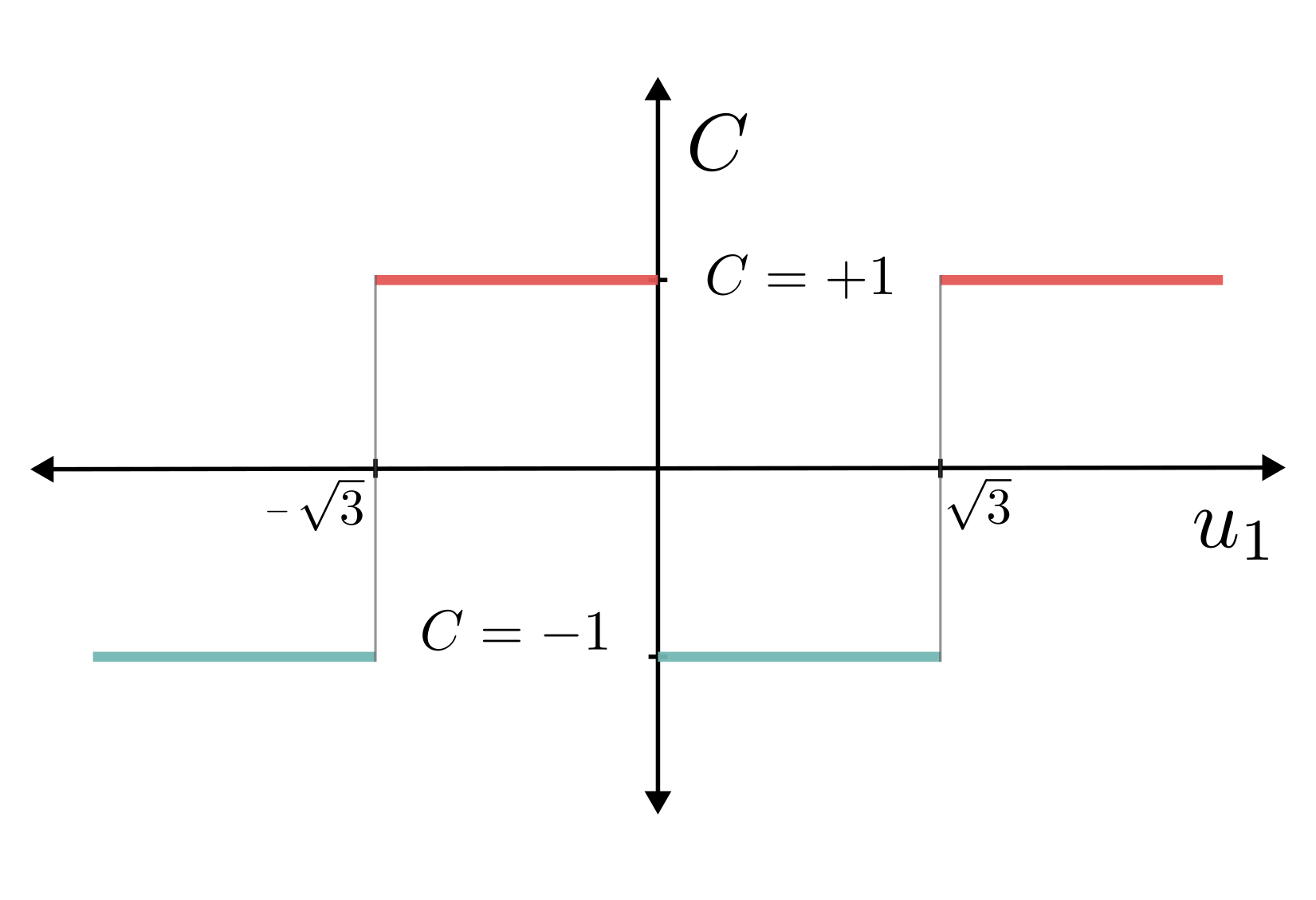}
    \caption{One-dimensional phase diagram of the model for Kagome lattice Eq.~\eqref{equation:KagomeH} with $t_1=1$. The model is gapless for $u_1 \in \{\pm\sqrt{3},0\}$, otherwise it is gapped. }
    \label{fig:kagome1-1}
\end{figure}

Since this is a three band system, there could be a jump by $2$ stemming from three bands crossing, a jump by two from a higher local charge in a two band crossing, or, generically, by two Dirac points with local charge $1$. 
The three band crossing is not  realized, but (fixing $t_1=1$) for the lowest two bands at $u_1=\pm \sqrt{3}$ there is one singular point whose local charge is $2$. At $u_1=0$ 
the situation is the generic case of two Dirac points.
This can seen from a local analysis; see Figure \ref{fig:Kagome}.
Note that the 2nd and 3rd band have the flipped behavior, i.e.\ two points at $\pm\sqrt{3}$ and one local charge $2$ point at  $0$.

\begin{figure}[h]

\begin{subfigure}[t]{0.45\textwidth}
{\includegraphics[height=1.1in]{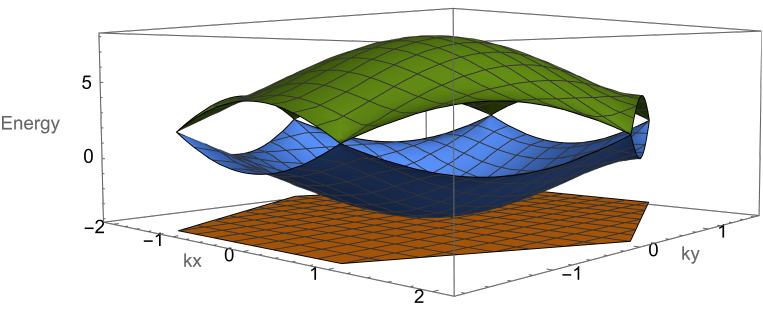}}
    \caption{$u_1=\sqrt{3}$}
\label{figK:A}
\end{subfigure}\hfill
\begin{subfigure}[t]{0.45\textwidth}
{\includegraphics[height=1.2
in]{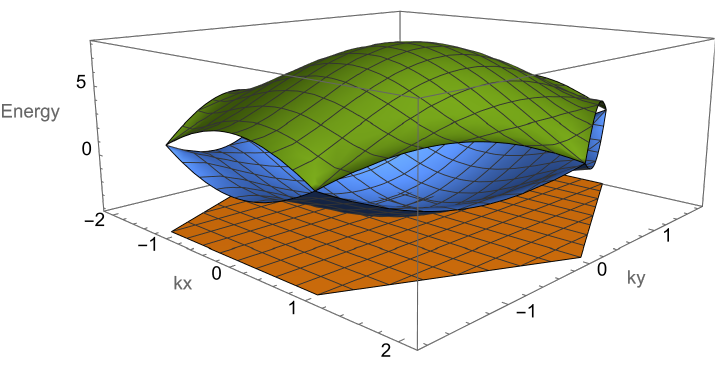}}
    \caption{$u_1=-\sqrt{3}$}
\label{figK:B}
\end{subfigure}\hfill

\begin{center} \begin{subfigure}[t]{0.45\textwidth}
{\includegraphics[height=1.2
in]{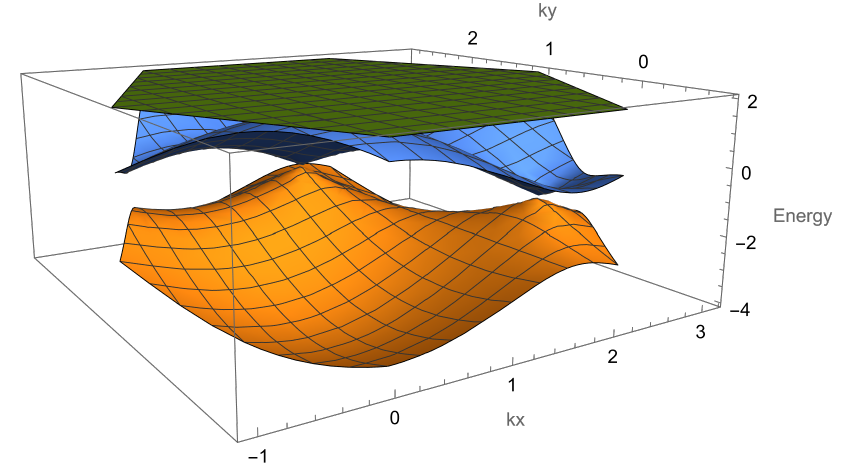}}
    \caption{$u_1=0$}
\label{figK:C}
\end{subfigure}
\end{center}
\caption{The bands and critical points for $t_1=1$. For the lowest and 2nd lowest band: (A) at $u_1=-\sqrt{3}$,  one charge $2$  point plotted at the center of the unit cell, (B) at $u_1=\sqrt{3}$ one charge $2$  point plotted at the center of the unit cell, and  (C)  at $u_1=0$ two Dirac points.  \label{fig:Kagome}}
\end{figure}

We can also expand the Hamiltonian in terms of Gell--Mann matrices since they span traceless Hermitian $3\times 3$ matrices. We used the following ordering for the Gell--Mann matrices:
\begin{equation}
\begin{aligned}
\lambda_1 &= \begin{pmatrix} 0 & 1 & 0 \\ 1 & 0 & 0 \\ 0 & 0 & 0 \end{pmatrix}, \quad
\lambda_2 = \begin{pmatrix} 0 & -i & 0 \\ i & 0 & 0 \\ 0 & 0 & 0 \end{pmatrix}, \quad
\lambda_3 = \begin{pmatrix} 1 & 0 & 0 \\ 0 & -1 & 0 \\ 0 & 0 & 0 \end{pmatrix}, \\
\lambda_4 &= \begin{pmatrix} 0 & 0 & 1 \\ 0 & 0 & 0 \\ 1 & 0 & 0 \end{pmatrix}, \quad
\lambda_5 = \begin{pmatrix} 0 & 0 & -i \\ 0 & 0 & 0 \\ i & 0 & 0 \end{pmatrix}, \quad
\lambda_6 = \begin{pmatrix} 0 & 0 & 0 \\ 0 & 0 & 1 \\ 0 & 1 & 0 \end{pmatrix}, \\
\lambda_7 &= \begin{pmatrix} 0 & 0 & 0 \\ 0 & 0 & -i \\ 0 & i & 0 \end{pmatrix}, \quad
\lambda_8 = \frac{1}{\sqrt{3}} \begin{pmatrix} 1 & 0 & 0 \\ 0 & 1 & 0 \\ 0 & 0 & -2 \end{pmatrix}.
\end{aligned}
\end{equation}
The Hamiltonian in momentum space can be written as $H(k) = h'(k) \cdot \mathbf{\lambda}$, where $\mathbf{\lambda}$ is an 8-dimensional vector of Gell--Mann matrices.
\begin{multline}
  h'(k)= 
(-2 t_1 \cos k_x,
-2 u_1 \cos k_x ,
0,
-2 t_1 \cos (\frac{k_x+\sqrt{3} k_y}{2}),\\
2 u_1 \cos (\frac{k_x+\sqrt{3} k_y}{2}),
-2 t_1 \cos (\frac{k_x-\sqrt{3} k_y}{2}),
-2 u_1 \cos (\frac{k_x-\sqrt{3} k_y}{2}),
0)^T
  \end{multline}

It should be noted that the coefficients can be arranged as: \begin{equation}
\begin{aligned}H(k) = &\cos(k_x) [-2t_1 \lambda_1-2u_1 \lambda_2] + \cos(\frac{k_x+\sqrt{3}k_y}{2})[-2t_1 \lambda_4 +2u_1 \lambda_5] \\ &+\cos(\frac{k_x-\sqrt{3}k_y}{2})[-2t_1 \lambda_6 -2u_1 \lambda_7]\end{aligned}
\end{equation} 
In this form it is evident that the target space is 3-dimensional subspace of $\R^8$ which depends on $t_1$ and $u_1$. Furthermore, since the eigenfunctions do not depend on the overall scale of the Hamiltonian, in the non--degenrate case, we can normalize $h(k)$ to $\hat{h}(k)$ establishing the target space as $S^2\subset \R^3\subset \R^8$. 
We can then use the results of \S\ref{subsection:calcdegree} to calculate the Chern number as the degree of this map. Using the method of commensurate sublattices, the Hamiltonian is modified to:
\begin{equation}\label{equation:kagomeN2}
\begin{aligned} 
H^{[C=N^2]}(\mathbf{k}) =\; & -2t_1\left(\begin{array}{ccc}
0 & \cos (N k \cdot a_1) & \cos (N k \cdot a_2) \\
\cos (N k \cdot a_1) & 0 & \cos (N k \cdot a_3) \\
\cos (N k \cdot a_2) & \cos (N k \cdot a_3) & 0
\end{array}\right) \\
& + 2i u_1\left(\begin{array}{ccc}
0 & \cos (N k \cdot a_1) & -\cos (N k \cdot a_2) \\
-\cos (N k \cdot a_1) & 0 & \cos (N k \cdot a_3) \\
\cos (N k \cdot a_2) & -\cos (N k \cdot a_3) & 0
\end{array}\right).
\end{aligned}
\end{equation}
This Hamiltonian will have a Chern number  $C=N^2$. 
\begin{figure}[h]
    \centering
    \includegraphics[scale=1.6]{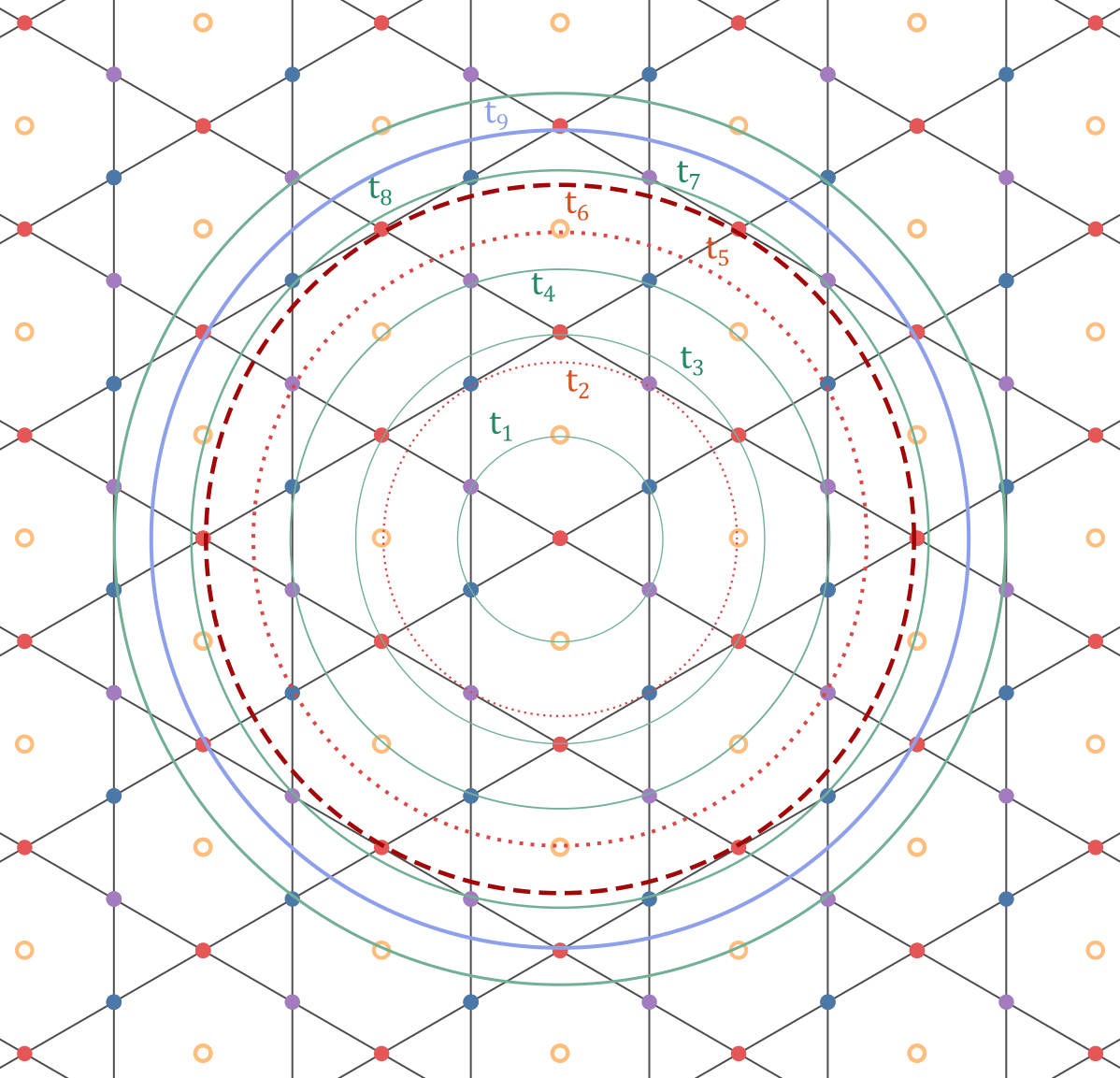}
    \caption{The structure of the nearest-neighbor interactions in Kagome lattice. The gold hollow circles represent fictitious points that complete the lattice into a triangular lattice. For certain ranges the structure of the first-nearest neighbors appears again. For example, $t_5$ has the same structure as $t_1$. Thus using $N=3$ in Eq.~\eqref{equation:kagomeN2} is a valid instance. }
    \label{fig:kagomefarnns}
\end{figure}

 \subsubsection{Criteria for choosing distant neighbors in Kagome lattice} 
 \label{par:Kagomecriteria}
 The Kagome lattice can be viewed as a subset of a triangular lattice with three species (blue, red and purple circles) while missing the hollow gold circles as depicted in Figure ~\ref{fig:kagomefarnns}. To avoid breaking the lattice symmetry, $N$ should be chosen such that the resulting neighbors at distance $N$ have the same number and type as the original model. The criteria for choosing $N$ (the integer multiple of the distant neighbors) again follows the recipe for a triangular lattice \S\ref{subsection:eisenstein} but, with the additional constraint that we want to  have the same number of fictitious points at that distance as the original model. For example, we cannot extend the $t_1$ interaction by either $N=2,4$ or more generally $N>0$ and $ N  \equiv 0 \pmod{2}$ as at these distances the atoms in the same direction as the first nearest neighbors are of the same species as the central atom. The criteria is then to have $ N  \equiv 1 \pmod{2}$ on top of the usual constraint: \( N = 3^{b_0} p_1^{b_1} p_2^{b_2} \dots p_c^{b_c} \), where \( N \) is the product of rational inert primes \( p_i \equiv 2 \pmod{3} \) as discussed in \S\ref{subsection:eisenstein}.
 
\section{Conclusion and Outlook}
We have studied the possibility of having phase diagrams with arbitrary Chern numbers and arbitrary crossings between them. In addition, we reviewed different methods of computing Chern numbers in 2D topological insulators. It was shown that for two--band systems  there exists a simple equivalence between the fiber-bundle picture and the mapping degree picture using the standard $spin\frac{1}{2}$ family. For higher band systems this is also possible, but one has to know both the pull--back map and the bundle structure of the pulled--back Hamiltonian.
The description in terms of pull--backs also  results in the rather simple ray method for computation of Chern classes and phase diagrams of given systems, as opposed to the integration of the connection over the whole $T^2$ Brillouin zone. Physically, this means that we do not need to compute the ground state wave function in order to know its Chern number \S\ref{subsection:calcdegree}. 

This relationship also provided insight into constructing new models with higher Chern numbers. There are two main obvious ways to construct higher degree Hamiltonians: composing the domain space with a mapping to itself (e.g. $T^2\to T^2$ in the case of a 2D periodic Brillouin zone) or composing with a mapping from the target space to itself (e.g. $S^2\to S^2$ in the case of a Bloch sphere).
 The latter allowed us to construct arbitrary tame phase diagrams by pulling back the standard  $spin\frac{1}{2}$ family of Hamiltonians by families of maps with prescribed mapping degree.
 This has potential to applications in designing materials and quantum information theory, as now the design question is given by gluing together  standard systems in a prescribed way using a finite amount of data. One feature is a minimal number of Dirac points appearing for each wall--crossing, which is what is expected generically, but now realized explicitly.

The former formalism, that is coverings of the torus by itself, can be thought of as moving to super--cells that for instance appear in stacking. This can be naturally realized for lattice models by distant-hopping interaction terms in real space.  The requirement that the new interactions have the same structure, in order to have simple covering maps, led to constraints on the range of hoppings admitted. The lattices at these distances should be commensurate, that is simply scalings of the original lattice. This classification problem was solved for those lattices which are also lattices of  quadratic integers. The study then translates to questions about the existence of a representation of primes in certain quadratic forms \cite{cox1989primes}. This is controlled by the behavior of rational primes after extension as explained in the appendix \ref{par:quadratic}. 

Another direction was to provide new maps for honeycomb and triangular lattices that implement intrinsically new mappings with higher degree that are not obtained by composing. These later maps should be of experimental value as they only extend the real interaction terms. 

While the analysis here is readily applicable for lattices that correspond to quadratic integer domains such as triangular and square lattices, many other types of lattices can be treated as easily with minor modifications. This is because such lattices are often a subset of the former lattices. For example, honeycomb, Kagome, dice lattices can be completed into a triangular lattice, while Lieb, checkerboard (Planar Pyrochlore), rectangular lattices (with rational ratio between sides) can be completed into a square lattices. In this paper, we showed how honeycomb and Kagome lattices, as examples, can be treated using this construction. It should also be noted that one can use the methods here for lattices that consist of two superimposed lattices as well (e.g. a triangular lattice on top of a square lattice \dots etc). Another direction for generalization that is readily accounted for by this construction is the number of bands. Since the main method of using bigger commensurate sublattices to achieve higher-degree maps solely relies on composing in domain space (the Brillouin Zone), the construction applies as well to $n$-band systems, as evident by the three-band Kagome case which did not require any new modifications. These investigations present a unified way to look at many constructions in the physics literature where higher-Chern number models were constructed by distant-hopping \cite{Bena2011,Sticlet2013,Wang2015,Mondal2022,circuits2023,Woo2024}. It also presents new simpler models with high Chern numbers that should be useful in experiments for Chern insulators or fractional topological insulators.

\section{Acknowledgments}
The authors are grateful for the funding support from the College of Science at Purdue University. R.K would like to acknowledge funding from the Simons foundation.
We would also like to thank Shawn Cui and Sabre Kais for motivating and very helpful initial discussions as well as Garth Simpson for inquiring about the possibility of direct phase transitions which are now realized via rose-curves.

\appendix
\section{Quadratic Integers}
\label{par:quadratic}
\subsection{Introduction}
\label{par:qintro}
 A natural setting to study 2D lattices is through the extension of the field of rationals $\Q$ by the square root of a negative square free integer $-d \in \Z^-$,  which is denoted $\Q[\sqrt{-d}]$.  The algebraic integers of the field $\Q[\sqrt{-d}]$ are those numbers that are roots of a polynomial equation of second degree in two variables with all integer coefficients they are denoted by $\mathcal{O}_{{\mathbb{Q}}[\sqrt{-d}]}$. They form an integral domain and further $\mathcal{O}_{{\mathbb{Q}}[\sqrt{-d}]}= \Z[\omega]=\{a+\omega b: a, b \in \mathbb{Z}\}$, where $\omega = \frac{1+\sqrt{-d}}{2}$ if $-d \equiv 1 \pmod{4}$ and $\omega = \sqrt{-d}$ otherwise. This ensures that the domain $\Z[\omega]$ is integrally closed. We will use $\Z[\omega]$ in this definition, this is not to be confused with the third root of unity which we discuss later. 
 
Choosing an embedding of the algebraic closure $\overline{\Q}\to \C$ and using correspondence between $\mathbb{C}= \mathbb{R}^2$ as real vector spaces, the quadratic integers will form a lattice in $\C$ that can be used to study physical lattices. For $-d \neq0$, the resulting lattice $\Gamma$ will be a  lattice in $\R^2$ of the form:
\begin{equation} \label{eq:quadtolattice}
\G =
\begin{cases}
\mathbb{Z}(1,0) + \mathbb{Z}\Bigl(\frac{1}{2},\frac{\sqrt{d}}{2}\Bigr), & \text{if } -d \equiv 1 \pmod{4}, \\[1mm]
\mathbb{Z}(1,0) + \mathbb{Z}(0,\sqrt{d}), & \text{Otherwise.}
\end{cases}
\end{equation}

A main tool in the study is the norm function $N : \Z[\omega] \to \mathbb{N}$ which is defined as $N(a+b\omega) = (a+b\omega)\overline{(a+b\omega)} $. Where the over-bar denotes conjugation: $\overline{(a+b \sqrt{h})} \coloneqq (a-b \sqrt{h})$ for any $h\in \Z$. After embedding into $\C$  the conjugation corresponds to complex conjugation. Consequently, the norm function is non-negative, multiplicative and equals the square of the Euclidean distance from the origin. This completes the identification of the quadratic integers $\Z[\omega]$ with the 2D lattice $\G$ \cite{cox1989primes,conway1999sphere,IrelandRosen1990}.

Quadratic integers inherit a lot of properties from the ring of integers $\Z$ while certain properties fail. For example, the domain $\Z[\sqrt{-5}]$ is not a unique-factorization domain (UFD). This is evidenced by the factorization: $6=2 \cdot 3=(1+\sqrt{-5})(1-\sqrt{-5})$. Additionally, in the same domain, we have $5= (-1) \cdot \sqrt{-5}\cdot \sqrt{-5}$ which means that the rational prime $5 \in \Z$ is now reducible. However, $\Z[w]$ are Dedekind domains and any ideal factors uniquely into prime ideals \cite{IrelandRosen1990}. A rational prime $p$ can have three behaviors after the extension. First, it can remain a prime and its principal ideal $(p)$ remains a prime ideal. In this case it is called an inert prime. Second, it can ramify as the square of a single prime ideal $(p) = \mathfrak{p}_1^2$ with $\mathfrak{p}_1$ a prime ideal in $\Z[\omega]$. Third, it can split into two distinct prime ideals $(p) = \mathfrak{p}_1\mathfrak{p}_2$. The behavior of rational primes after extension is well-studied and we mention the following theorem without a proof \cite{cox1989primes,Neukirch1999,IrelandRosen1990}. 

\begin{theorem}\label{theorem:oddprimes} Let \(K=\mathbb{Q}(\sqrt{d})\) be a quadratic field, an odd rational prime $p$ ramifies, splits or remains inert in the quadratic extension to the field $K=$ iff $\left(\frac{d}{p}\right)=0,1,-1$ respectively, where the Legendre symbol can be computed as follows:
\[
\left(\frac{d}{p}\right) =
\begin{cases}
0, & \text{if } p \mid d, \text{ $p$ ramifies,} \\
1, & \text{if } d^{\frac{p-1}{2}} \equiv 1 \pmod{p},  \text{ $p$ splits,} \\
-1, & \text{if } d^{\frac{p-1}{2}} \equiv -1 \pmod{p}  \text{, $p$ is inert}.
\end{cases}
\]
\end{theorem}
For ease of presentation we treat the $p=2$ case separately. We have the following theorem\cite{IrelandRosen1990}.
\begin{theorem}\label{theorem:evenprimes}
Let \(K=\mathbb{Q}(\sqrt{d})\) be a quadratic field with field discriminant 
\[
D=\begin{cases}
d, & \text{if } d\equiv 1\pmod{4},\\[1mm]
4d, & \text{if } d\not\equiv 1\pmod{4}.
\end{cases}
\]
Then the rational prime \(2\) behaves as follows in \(\mathcal{O}_{{\mathbb{Q}}[\sqrt{-d}]}\):
\[
\begin{cases}
 \text{if } 2 \mid D, \text{ then \(2\) ramifies,} \\
 \text{if }2 \nmid D \text{ and } d \equiv 1 \pmod{8},  \text{ then \(2\) splits,} \\
 \text{if }2 \nmid D \text{ and } d \equiv 3 \pmod{8},  \text{ then \(2\) remains inert.}
\end{cases}
\]
\end{theorem}

\subsection{Isolated norms}

In the following, Greek letters are  used for quadratic integers while Latin letters are reserved for rational numbers. The letters $p,$ and $\pi$ are used for rational primes and quadratic primes in their respective ring. We start by three lemmata that illustrate the structure of the norm function in relation to inert primes. 

\begin{lemma}\label{lem:inert1}
In a quadratic integer ring $\Z[\omega]$, the norm function cannot represent any inert prime $p\in \Z$. (i.e. there is no number with norm $p$).
\end{lemma}
\begin{proof}
Assume $\exists \alpha \in \Z[\omega] $ such that $N(\alpha) = p$. This entails, however, $N(\alpha)=\alpha \cdot \overline{\alpha}=p$ which is a contradiction since $p$ is an inert prime and neither $\alpha$ nor its conjugate are units. \end{proof}

\begin{lemma}\label{lem:inert2}
In a quadratic integer ring $\Z[\omega]$, the only numbers with norm $p^2$ for an inert prime $p \in \Z$ are associates of $p$ (ie. $p$ has an isolated norm).
\end{lemma}
\begin{proof}
Assume $\exists \alpha \in \Z[\omega] $ such that $N(\alpha) = p^2$ and $p \nmid \alpha$. Let's write the norm of $\alpha$ as $N(\alpha)=\alpha \cdot \overline{\alpha}=p^2$. Since $p$ is an inert prime, $p \mid \alpha\cdot\overline{\alpha} \Rightarrow p\mid \alpha \text{ or } p \mid \overline{\alpha}$. The first case leads to a contradiction, we are then left with $p \mid \overline{\alpha}$. Write $\overline{\alpha}=p\cdot\mu$ for some number $\mu \in \Z[\omega]$. Taking conjugates, We have ${\alpha}=p \cdot \overline{\mu}$. Then again $p\mid \alpha$: Take the norm of $\alpha = p \cdot \overline{\mu}$ and use the multiplicativity of the norm $N(\alpha) = N(p)\cdot N(\overline{\mu}) = p^2 \cdot N(\overline{\mu})=p^2$ and hence $\mu$ is a unit. This completes the proof. \end{proof}
Note that the norm of irreducible elements can be a composite number in a non-UFD. For example, in $\Z[\sqrt{-5}]$, $N(1+\sqrt{-5})=(1+\sqrt{-5})(1-\sqrt{-5})=6$. We, however, do have the following lemma.
\begin{lemma}\label{lem:inert3}
In a quadratic integer ring $\Z[\omega]$, if an inert prime $p \in \Z$ divides the norm of an irreducible $\alpha \in \Z[\omega]$, then $N(\alpha) = p^2$ and $\alpha$ is an associate of $p$.
\end{lemma}
\begin{proof}
Assume $\exists \alpha \in \Z[\omega] $ such that $p \mid N(\alpha)$ and $p \nmid \alpha$. Let's write the norm of $\alpha$ as $N(\alpha)=\alpha \cdot \overline{\alpha}=p \cdot \mu$. Since $p$ is an inert prime, $p \mid \alpha\cdot\overline{\alpha} \Rightarrow p\mid \alpha$ from the multiplicativity of conjugation or equivalently the proof of Lemma~\ref{lem:inert2}. We can then write ${\alpha}=p\cdot\nu$. However, since $\alpha$ is an irreducible then $\nu$ is a unit, and $N(\alpha)=p^2$. \end{proof}

Ramified primes obey a modified versions of Lemmata~\ref{lem:inert1}, \ref{lem:inert2}, and \ref{lem:inert3}. We first note that in some cases the ramified prime $p$ ramifies as the square of non-principal ideals. For example, in $\Z[\sqrt{-5}]$, $(2) = (2,\,1+\sqrt{-5})^2$. In this case, there is no number whose norm is $2$. This is easily seen from the norm function $N(a+b\sqrt{-5})=a^2+5b^2$.
\begin{lemma}\label{lem:ramified1}
In a quadratic integer ring $\Z[\omega]$, if the norm function represents a ramified prime $p\in \Z$, then $p=u\cdot \pi^2$ for some prime $\pi \in \Z[\omega]$ and unit $u$.
\end{lemma}
\begin{proof}
Since the norm function represents $p$, $\exists\alpha$ such that $N(\alpha) = \alpha \cdot \overline{\alpha}=p$. Taking ideals  $(\alpha)(\overline{\alpha})=(p)$. Since $p$ ramifies, the ideal generated by it has the form $(p) = \mathfrak{p}_1^2$ for some prime ideal $\mathfrak{p}_1$. We then have $(\alpha)(\overline{\alpha})=\mathfrak{p}_1^2$. Because quadratic integers are Dedekind domains, ideals factor uniquely into prime ideals. We then have $(\alpha)=(\overline{\alpha})=\mathfrak{p}_1$ which means $\alpha$ is a prime (in Dedekind domains, a principal ideal is prime iff its generator is a prime), and $\alpha = u \cdot \overline{\alpha}$.
\end{proof}
Similar to the inert prime case where the square of an inert prime was an inert norm, if a ramified prime is represented by the norm function then it has an isolated norm.

\begin{lemma}\label{lem:ramified2}
In a quadratic integer ring $\Z[\omega]$, if the norm function represents a ramified prime $p\in \Z$, then $p$ is an isolated norm (all numbers with norm $p$ are associates of each other).
\end{lemma}
\begin{proof}
If $N(\alpha)= N(\beta) = p$ then $\alpha,\beta$ are primes by Lemma~\ref{lem:ramified1} and further $\overline{\alpha}=u\cdot \alpha, \ \overline{\beta}=w\cdot \beta$ for units $u,w \in \Z[\omega]$. We then have $u\cdot\alpha^2=w\cdot\beta^2$ and since $\alpha$ and $\beta$ are primes they have to be associates. \end{proof}
The last lemma for inert primes fails to work in general for ramified primes. In $\Z[\sqrt{-5}]$, consider $N(1+\sqrt{-5}) = 6$. Even though the ramified prime $2$ divides the norm of the irreducible $(1+\sqrt{-5})$, the norm is not exactly $2$. We then have the modified version of Lemma~\ref{lem:inert3}.

\begin{lemma}\label{lem:ramified3}
In a quadratic integer ring $\Z[\omega]$, if a ramified prime $p \in \Z$ divides the norm of an irreducible $\alpha \in \Z[\omega]$, and this ramified prime is represented by the norm function, then $N(\alpha)=p$, $\alpha$ is a prime.
\end{lemma}
\begin{proof}
Start with an irreducible $\alpha$ such that $N(\alpha)=\alpha\cdot \overline{\alpha}= p \cdot c$ for some non-unit integer $c$. Using Lemma~\ref{lem:ramified1}, $p=u\cdot\pi^2$ for some prime $\pi \in \Z[\omega]$ with $\overline{\pi}=u\cdot\pi$. We then have $\alpha\cdot \overline{\alpha}=u\cdot\pi^2 \cdot c$. Since $\pi$ is a prime, it divides $\alpha$ and also $\overline{\alpha}$ (because $\overline{\pi}=u\cdot\pi$). Thus $\alpha=\pi \cdot w$ for some number $w$, however the norm function fixes $w$ as a unit.   \end{proof}

The ring $\Z[\omega]$ is a Dedekind domain. In particular, it is a Noetherian domain which implies it is an atomic domain and each number can be written in terms of a finite set of irreducibles. The factorization into irreducibles can still be non-unique \cite{IrelandRosen1990}. Using the previous Lemmata~\ref{lem:inert1},~\ref{lem:inert2},~\ref{lem:inert3},~\ref{lem:ramified1},~\ref{lem:ramified2} and ~\ref{lem:ramified3}, we can identify a family of numbers that have isolated norms in any $\Z[\omega]$.

\begin{theorem}\label{theorem:isolatednorms}
 In a quadratic integer ring $\Z[\omega]$, the number \( d =  \prod_{i=1}^{r} p_i^{b_i} \prod_{j=1}^{s}{\pi_j}^{c_j} \)  has an isolated norm, where $p_i \in \Z$ is an inert prime and $\pi_j \in \Z[\omega]$ is a prime whose norm is a ramified prime $q_j \in \mathbb{N}$ and $b_i,c_j,r,s \in \mathbb{N} $.
\end{theorem}

\begin{proof}
Assume $\exists \alpha \in \Z[\omega] $ such that $N(\alpha) = N(d)$. Since $\Z[\omega]$ is an atomic domain, we can write $\alpha$ in terms of (possible non-unique) non-assosciate irreducibles $\mu_k$ and a unit $u$: $\alpha = u\cdot \prod_{k=1}^{t} {\mu_k}^{m_k}$ with $t,m_k \in \mathbb{N}$. From the multiplicativity of the norm we have: $N(\alpha) = \prod_{k=1}^{t} N({\mu_k})^{m_k} = N(d) = \prod_{i=1}^{r} p_i^{2b_i} \prod_{j=1}^{s}{q_j}^{c_j} $. Since the norm is a map to $\mathbb{N}$ there is a unique factorization for the norm of each irreducible $\mu_k$ in terms of rational primes. In this factorization, only inert or ramified primes should appear as otherwise if a split prime $g$ appears this will imply $g \mid \prod_{i=1}^{r} p_i^{2b_i}\prod_{j=1}^{s}{q_j}^{c_j}$ which is rejected since all $p_i$ and $q_j$ are either inert or ramified. This means that the norm of every irreducible $\mu_k$ only consists of powers of inert or ramified primes. Using Lemmata~\ref{lem:inert3} and~\ref{lem:ramified3}, the norm of any such irreducible is $N(\mu_k) = {p^{\prime}}_i^2$ for some inert prime $p_i^{\prime} \in \Z$ or $N(\mu_k) = q^{\prime}_k \in \Z$ for a ramified prime $q^{\prime}_k$. Since the irreducibles are non-associates by construction, every norm contribute a distinct prime: $N(\mu_k)  \neq N(\mu_l) $ for $k \neq l$. Let's reorder the irreducibles in one factorization of $\alpha$ according to their norm type: $N(\alpha) = \prod_{k=1}^{x} {p^{\prime}_k}^{2m_k}\prod_{k=x+1}^{t} {q^{\prime}_k}^{m_k}$ for some integer $0\leq x\leq  t $. We then have $ \prod_{k=1}^{x} {p^{\prime}_k}^{2m_k}\prod_{k=x+1}^{t} {q^{\prime}_k}^{m_k} =\prod_{i=1}^{r} p_i^{2b_i} \prod_{j=1}^{s}{q_j}^{c_j} $. Finally, since $\mathbb{N}$ is a UFD, we have $x=r$, $t=x+s$, $m_k = b_i$ for $k \leq x$, and $m_k = c_j$ for $k>x$, and $p^{\prime}_k = p_i$, $q^{\prime}_k = q_j$  up to reordering of indices. This fixes $\alpha = u \cdot \prod_{i=1}^{r} p_i^{b_i} \prod_{j=1}^{s}{\pi_j}^{c_j} $ which is an associate of $d$.
\end{proof}

Note that while this necessarily gives an infinite set of isolated norms, it is only exhaustive in UFDs. A counter example in $\Z[\sqrt{-14}]$ is $3,5$ which are split primes but, they have isolated norms. Moreover, in non-UFDs, it can happen that the product of two isolated norms is not an isolated norm. For an example, consider again the product $3\cdot5$. Each factor has an isolated norm as can be checked explicitly, however, the numbers $1\pm4\sqrt{-14}$ and $13\pm2\sqrt{-14}$ are non-associates and also have the same norm as $15$. This cannot happen in a UFD (as discussed later). Evidently, $\Z[\sqrt{-14}]$ has a class number of $4$ and is thus a non-UFD. Another example in a domain with a class number of $3$ is found in $Z[\frac{1+\sqrt{-23}}{2}]$. The two numbers $2,3$ with isolated norms $4,9$ multiply to give a number with a non-isolated norm.

The results of this section can be restated in terms of the structure of the nearest neighbors for lattices described by a quadratic integer ring Eq.~\eqref{eq:quadtolattice} which works only if the norm function can be interpreted as a distance for field extensions $\Q[\sqrt{-d}]$.

\begin{corollary}\label{corollary:generalisolated}
    In a full 2D lattice $\Gamma$ described by a quadratic integer ring $\Z[\omega]$, there are exactly $|U|$ nearest neighbors to any lattice point lying at a Euclidean distance  \( d = \prod_{i=1}^{r} p_i^{b_i}\prod_{j=1}^{s} q_j^{c_j/2} \) where each $p_i$ is an inert prime $p_i \in \Z$ in $\Z[\omega]$, and $q_j \in \mathbb{N}$ is a ramified prime represented by the norm function, $b_i,c_j,r,s \in \mathbb{N}$, and $|U|$ is the cardinality of the set of units in $\Z[\omega]$.
\end{corollary}

 The set of units $U$ for each ring $\Z[\omega]$ is given by \cite{IrelandRosen1990}:
\begin{equation}
\label{eq:units}
U =
\begin{cases}
\{1, i, -1, -i\} & \text{for } \mathbb{Z}[i], \\[1mm]
\{\pm 1,\, \pm \omega,\, \pm \omega^2\} & \text{for } \mathbb{Z}[\omega] \text{ with } \omega = \frac{-1+\sqrt{-3}}{2}, \\[1mm]
\{1,-1\} & \text{otherwise}.
\end{cases}
\end{equation}

\subsection{Unique Factorization Domains}\label{par:ufd}
In the case where $\Z[\omega]$ is a UFD, stronger results apply. As any number admits a unique factorization into primes, the prime factors determines if the norm is isolated or not. The cases where a quadratic integer ring with an imaginary field is a UFD is given by the following Stark-Heegner theorem \cite{Stark1969a,Stark1969b}.
\begin{theorem}[Stark-Heegner]\label{theorem:Heegner}
 The quadratic imaginary number field $\Q$ $\Z[\sqrt{-d}]$ has a ring of integers which is a unique-factorization domain iff \[ 
d \in \{-1, -2, -3, -7, -11, -19, -43, -67, -163\}.\]

\end{theorem}

Lemmata \ref{lem:inert1},\ref{lem:inert2} and \ref{lem:inert3} discussing inert rational primes carry over without generalization. However, we have the following two lemmata for ramified and split primes. 
\begin{lemma}\label{lem:ufd1}
In a quadratic integer ring $\Z[\omega]$ which is a UFD, any ramified prime $p \in \Z$ is represented by the norm function, and this norm is isolated.
\end{lemma}
\begin{proof}
The ideal generated by the rational prime $p$ splits as $(p) = \mathfrak{p}_1^2$ for a prime ideal $  \mathfrak{p}_1 \subset\Z[\omega]$. Since every Dedekind domain that is a UFD is also a PID (Principal Ideal Domain), we have $p=\pi_1^2$ with a prime $\pi_1\ \in \Z[\omega]$. The norm function $N(p)=p^2=N(\pi_1)^2$ fixes $N(\pi_1)=p$. This proves the first part. For the second part, use Lemma~\ref{lem:ramified2}. 
\end{proof}
On the other hand, all split primes have non-isolated norms.
\begin{lemma}\label{lem:ufd2}
In a quadratic integer ring $\Z[\omega]$ which is a UFD, any split prime $p \in \Z$ is represented by the norm function, and this norm is not isolated.
\end{lemma}
\begin{proof}
The ideal generated by the rational prime $p$ splits as $(p) = \mathfrak{p}_1\mathfrak{p_2}$ for two distinct prime ideals $  \mathfrak{p}_1,\mathfrak{p}_2 \subset\Z[\omega]$. Since every Dedekind domain that is a UFD is also a PID, we have $p=\pi_1\pi_2$ for distinct primes $\pi_1,\pi_2 \in \Z[\omega]$. The norm function $N(p)=p^2=N(\pi_1)N(\pi_2)$ fixes $N(\pi_1)=N(\pi_2)=p$ as otherwise one of them has to be a unit. This means in particular, $p=\pi_1\pi_2=\pi_1 \cdot \overline{\pi_1}$. Hence, $\pi_2=\overline{\pi_1}$. Finally, $\pi_1 \neq u\cdot\overline{\pi_1}$ for some unit $u$ because, if this is the case then $p$ ramifies.
\end{proof}

Combining these lemmata with theorem~\ref{theorem:isolatednorms}, we have the following theorem.
\begin{theorem}\label{theorem:ufdisolatednorms}
In a quadratic integer ring $\Z[\omega]$ which is a UFD, a number has an isolated norm iff its norm is not divisible by any split prime $p \in \Z$.
\end{theorem}
\begin{proof}
Call such a number $\alpha$. It has a unique factorization, $\alpha = \prod_{i=1}^r\pi_i^{m_i}$ into distinct primes $\pi_i \in \Z[\omega]$, and its norm is $N(\alpha)=\prod_{i=1}^r[{N(\pi_i)}]^{m_i}$. However, since each norm $N(\pi_i)$ only consists of inert or ramified primes, by Lemmata~\ref{lem:inert3} and \ref{lem:ufd2}, this norm only comes from a unique number up to associates. This fixes any number with the same norm up to units. For the reverse direction, assume $p \mid N(\alpha)$ where $p \in \mathbb{N}$ is a split prime in $\Z[\omega]$. Then there is at least one prime $\pi_j$ in the factorization of $\alpha$ such that $N(\pi_j)=p$. Write $\alpha = \pi_j\cdot \gamma$ for some number $\gamma \in \Z[\omega]$. By Lemma~\ref{lem:ufd1}, the number $\alpha^{\prime} = \overline{\pi_j}\cdot \gamma$ has the same norm as and is not an associate to $\alpha$. 
\end{proof}

\bibliographystyle{amsplain}
\bibliography{main}

\end{document}